\documentclass{vldb}

\usepackage{times}
\usepackage{color}
\usepackage{balance}
\usepackage{url}
\usepackage{siunitx}
\usepackage{epsfig,tabularx,amssymb,amsmath,subfigure,multirow, algorithm, algorithmic,graphicx}

\long\def\comment#1{}

\newcommand{\nop}[1]{}

\newtheorem{definition}{\bf Definition}

\newtheorem{lemma}{\bf Lemma}[section]

\newcommand\abs[1]{\left\lvert #1 \right\rvert}

\title{Reliable Diversity-Based Spatial Crowdsourcing by Moving Workers}

\author{
	{Peng Cheng{\small $~^{\#}$}, Xiang Lian{\small $~^{*}$},  Zhao Chen{\small $~^{\#}$}, Rui Fu{\small $~^{\#}$}, Lei Chen{\small $~^{\#}$}, Jinsong Han{\small $~^{\dagger}$}, Jizhong Zhao{\small $~^{\dagger}$} } %
	\vspace{1.6mm}\\
	\fontsize{10}{10}\selectfont\itshape
	$^{\#}$\,Hong Kong University of Science and Technology, Hong Kong, China\\
	\fontsize{9}{9}\selectfont\ttfamily\upshape
	\{pchengaa, zchenah,  leichen\}@cse.ust.hk, rfu@ust.hk
	\vspace{1.2mm}\\
	\fontsize{10}{10}\selectfont\rmfamily\itshape
	$^{*}$\,University of Texas Rio Grande Valley, Texas, USA\\
	\fontsize{9}{9}\selectfont\ttfamily\upshape
	xiang.lian@utrgv.edu
	\vspace{1.2mm}\\
	\fontsize{10}{10}\selectfont\rmfamily\itshape
	$^{\dagger}$\,Xi'an Jiaotong University, Shaanxi, China\\
	\fontsize{9}{9}\selectfont\ttfamily\upshape
	\{hanjinsong, zjz\}@mail.xjtu.edu.cn
}

\begin{document}

\maketitle

\begin{abstract}
	With the rapid development of mobile devices and the crowdsourcing
	platforms, the spatial crowdsourcing has attracted much attention
	from the database community, specifically,  spatial crowdsourcing
	refers to sending a location-based request to workers according to
	their positions. In this paper, we consider an important spatial
	crowdsourcing problem, namely \textit{reliable diversity-based
		spatial crowdsourcing} (RDB-SC), in which spatial tasks (such as
	taking videos/photos of a landmark or firework shows, and checking
	whether or not parking spaces are available) are time-constrained,
	and workers are moving towards some directions. Our RDB-SC problem
	is to assign workers to spatial tasks such that the completion
	reliability and the spatial/temporal diversities of spatial tasks
	are maximized. We prove that the RDB-SC problem is NP-hard and
	intractable. Thus, we propose three effective approximation
	approaches, including greedy, sampling, and divide-and-conquer
	algorithms. In order to improve the efficiency, we also design an
	effective cost-model-based index, which can dynamically maintain
	moving workers and spatial tasks with low cost, and efficiently
	facilitate the retrieval of RDB-SC answers. Through extensive
	experiments, we demonstrate the efficiency and effectiveness of our
	proposed approaches over both real and synthetic datasets.
\end{abstract}

\section{Introduction}
\label{sec:introduction}

Recently, with the ubiquity of smart mobile devices and high-speed
wireless networks, people can now easily work as moving sensors to
conduct sensing tasks, such as taking photos and recording
audios/videos. While data submitted by mobile users often contain
spatial-temporal-related information, such as real-world scenes
(e.g., street view of Google Maps \cite{GoogleMapStreetView}), video
clips (e.g., MediaQ \cite{mediaq}), local hotspots (e.g., Foursquare
\cite{foursquare}), and traffic conditions (e.g., Waze \cite{waze}),
the \textit{spatial crowdsourcing} platform
\cite{deng2013maximizing, kazemi2012geocrowd} has nowadays drawn
much attention from both academia (e.g., the data\-base community) and
industry (e.g., Amazon's AMT \cite{amt}).

\begin{figure}[ht]\vspace{-2ex}\centering
	\scalebox{0.28}[0.28]{\includegraphics{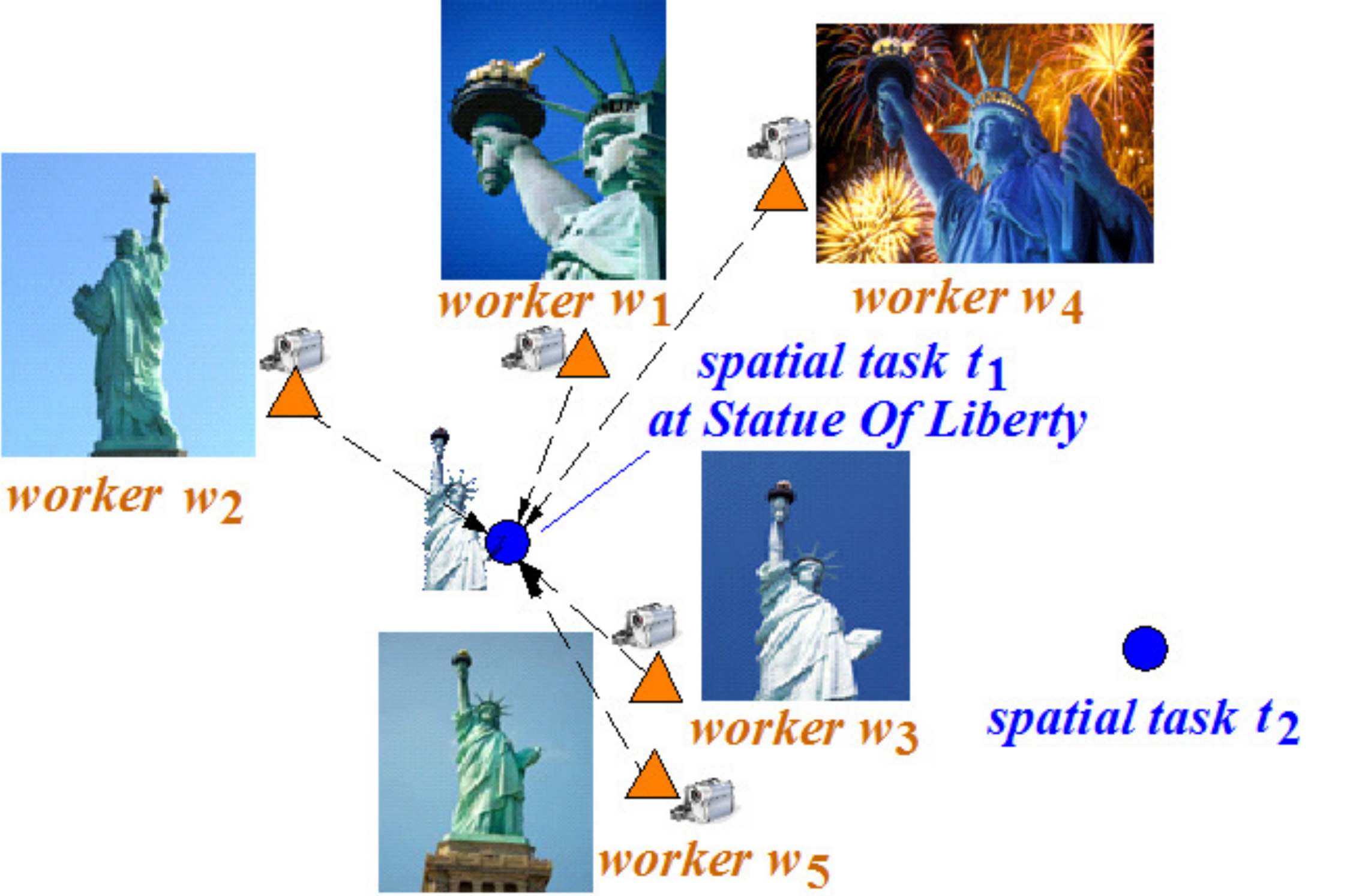}}\vspace{-2ex}
	\caption{\small An Example of Taking Photos/Videos of a Landmark
		(Statue of Liberty) in the Spatial Crowdsourcing
		System.}\vspace{-5ex}
	\label{fig:SC_example}
\end{figure}

Specifically, a spatial crowdsourcing platform
\cite{deng2013maximizing,kazemi2012geocrowd} is in charge of
assigning a number of workers to nearby \textit{spatial tasks}, such
that workers need to physically move towards some specified
locations to finish tasks (e.g., taking photos/videos).

\vspace{0.5ex}\noindent {\bf Example 1 (Taking Photos/Videos of a
	Landmark).} {\it    Consider a scenario of the spatial crowdsourcing
	in Figure \ref{fig:SC_example}, in which there are two spatial tasks
	at locations $t_1$ and $t_2$, and 5 workers, $w_1 \sim w_5$. In
	particular, the spatial task $t_1$ is for workers to take 2D
	photos/videos of a landmark, the Statue of Liberty, while walking
	from ones' current locations towards it. The resulting 2D
	photos/videos are useful for real applications such as virtual tours
	and 3D landmark reconstruction. Therefore, the task requester is
	usually interested in obtaining a full view of the landmark from
	diverse directions (e.g., photos from the back of the statue that
	not many people saw before).
	
	As shown in Figure \ref{fig:SC_example}, workers $w_1$ and $w_4$ can
	take photos from left hand side of the statue, worker $w_2$ can take
	photos from the back of the statue, and workers $w_3$ and $w_5$ can
	take photos from the front of the statue. Since photos/videos from
	similar directions are not informative for the 3D reconstruction or
	virtual tours, the spatial crowdsourcing system needs to select
	those workers who can take photos of the statue with as diverse
	directions as possible, and then assign task $t_1$ to them.
	
	Note that, when worker $w_4$ arrives at the location of $t_1$,
	he can take a photo of the landmark at night with fireworks. As
	a result, by assigning task $t_1$ to worker $w_4$, we can obtain a
	quite diverse photo at night for virtual tours, compared with that
	taken by worker $w_1$ in the daytime (even though they are taken
	from similar angles). Thus, it is also important to consider the
	temporal diversity of taking the photos, in terms of the arrival
	times of workers at $t_1$. \qquad $\blacksquare$
	
}

\vspace{0.5ex}\noindent {\bf Example 2 (Available Parking Space
	Monitoring over a Region).} {\it    In the application of monitoring
	free parking spaces in a spatial region, it is important to analyze
	the photos taken from diverse directions and at different time
	periods of the day, and predict the trend of available parking
	spaces in the future. This is because some available parking spaces
	might be hidden by other cars for photos just from one single
	direction (or multiple similar directions), and moreover, photos
	taken at different timestamps have richer information than those
	taken at a single time point, for the purpose of predicting the
	availability of parking spaces. Therefore, in this case, the spatial
	crowdsourcing system needs to assign such a task to those workers
	with diverse walking directions to it and arrival times. \qquad
	$\blacksquare$
	
}

In this paper, we will investigate a realistic scenario of spatial
crowdsourcing, where workers are dynamically moving towards some
directions, and spatial tasks are constrained by valid time periods.
For example, a worker may want to do spatial
tasks on the way home, and thus he/she tends to accept tasks only along
the direction to home, rather than an opposite direction. Similarly,
a spatial task of taking photos of the statue, together with
fireworks, is restricted by the period of the firework show time.
Therefore, a worker can only conduct a task within the constrained
time range and prefer to accepting the tasks close to his/her moving
direction. In this paper, we characterize features of moving workers
and time-constrained spatial tasks, which have not been studied
before.

Under the aforementioned realistic scenario, we propose the problem
of dynamic task-and-worker assignment, by considering the answer 
quality of spatial tasks, in terms
of two measures: \textit{spatial-temporal diversities} and
\textit{reliability}. In particular, inspired by the two applications above (Examples 1
and 2), for spatial tasks, photos/videos from diverse angles or
timestamps can provide a comprehensive view of the landmark, which
is more preferable than those taken from a single boring
direction/time; similarly, photos taken by different directions and
periods are more useful for the trend prediction of available
parking spaces. Therefore, in this paper, we will introduce the
concepts of \textit{spatial diversity} and \textit{temporal
	diversity} to spatial crowdsourcing, which capture the diversity
quality of the returned answers (e.g., photos taken from different
angles and at different timestamps) to spatial tasks.

Furthermore, in reality, it is possible that answers provided by
workers are not always correct, for example, the uploaded
photos/videos might be fake ones, or workers may deny the assigned
tasks. Thus, we will model the confidence of each worker, and in
turn, guarantee high \textit{reliability} of each spatial task, which is defined as the confidence
that at least one worker assigned to this task can give a high quality answer.

Note that, while existing works on spatial crowdsourcing
\cite{deng2013maximizing, kazemi2012geocrowd} focused on the
assignment of workers and tasks to maximize the total number of
completed tasks, they did not consider much about the constrained features
of workers/tasks (workers' moving directions and tasks' time
constraints). Most importantly, they did not take into
account the quality of the returned answers.

By considering both quality measures of spatial-temporal diversities
and reliability, in this paper, we will formalize the problem of
\textit{reliable diversity-based spatial crowdsourcing} (RDB-SC),
which aims to assign moving workers to time-constrained spatial
tasks such that both reliability and diversity are maximized. To the
best of our knowledge, there are no previous works that study
reliability and spatial/temporal diversities in the spatial
crowdsourcing. However, efficient processing of the RDB-SC problem
is quite challenging. In particular, we will prove that the RDB-SC
problem is NP-hard, and thus intractable. Therefore, we propose
three approximation approaches, that is, the greedy, sampling, and
divide-and-conquer algorithms, in order to efficiently tackle the
RDB-SC problem. Furthermore, to improve the time efficiency, we design a
cost-model-based index to dynamically maintain moving workers and
time-constrained spatial tasks, and efficiently facilitate the
dynamic assignment in the spatial crowdsourcing system. Finally,
through extensive experiments, we demonstrate the efficiency and
effectiveness of our approaches.

To summarize, we make the following contributions.\vspace{-2ex}

\begin{itemize}
	\item We formally propose the problem of reliable diversity-based spatial
	crowdsourcing (RDB-SC) in Section \ref{sec:problem_def}, by
	introducing the reliability and diversity to guarantee the quality
	of spatial tasks.\vspace{-1.5ex}
	
	\item We prove that the RDB-SC problem is NP-hard, and thus
	intractable in Section \ref{sec:reduction}.\vspace{-1.5ex}
	
	\item We propose three approximation approaches, greedy, sampling, and divide-and-conquer algorithms, in Sections \ref{sec:greedy}, \ref{sec:sampling} and
	\ref{sec:D&C}, respectively, to tackle the RDB-SC
	problem.\vspace{-1.5ex}

	\item We conduct extensive experiments in Section \ref{sec:exper} on both real and synthetic datasets and show the efficiency and effectiveness of our
	approaches.\vspace{-1.5ex}
\end{itemize}

We design a cost-model-based index structure to dynamically 
maintain workers and tasks in Section \ref{sec:gridIndexCostModel}. 
Section \ref{sec:related} overviews previous works on (spatial)
crowdsourcing. Finally, Section \ref{sec:conclusion} concludes this paper.

\vspace{-1ex}
\section{Problem Definition}
\label{sec:problem_def}\vspace{-1ex}

\subsection{Time-Constrained Spatial Tasks} \vspace{-1ex}

We first define the time-constrained spatial tasks in the
crowdsourcing applications.

\begin{definition}
	$($Time-Constrained Spatial Tasks$)$ Let $T=\{t_1,$ $t_2, ...,
	t_m\}$ be a set of $m$ time-constrained spatial tasks. Each spatial
	task $t_i$ ($1\leq i\leq m$) is located at a specific location
	$l_i$, and associated with a valid time period $[s_i, e_i]$. \qquad
	$\blacksquare$ \label{definition:task}\vspace{-2ex}
\end{definition}

In this paper, we consider spatial and time-constrained tasks, such
as ``taking 2D photos/videos for the Statue of Liberty together with
fireworks'', or ``taking photos of parking places during open hours
of the parking area in a region''. Therefore, in such scenarios,
each task can only be accomplished at a specific location, and,
moreover, satisfy the time constraint. For example, photos should be
taken by people in person and within the period of the firework
show. Therefore, in Definition \ref{definition:task}, we require
each spatial task $t_i$ be accomplished at a spatial location $l_i$
(for $1\leq i\leq m$), and within a valid period $[s_i, e_i]$.

The set of spatial tasks is dynamically changing. That is, the newly
created tasks keep on arriving, and those completed (or
expired) tasks are removed from the crowdsourcing system.

\subsection{Dynamically Moving Workers} 

Next, we consider dynamically moving workers.\vspace{-2ex}

\begin{definition} $($Dynamically Moving Workers$)$ Let
	$W=\{w_1,$ $w_2,$ $..., w_n\}$ be a set of $n$ workers. Each worker
	$w_j$ ($1\leq j\leq n$) is currently located at position $l_j$,
	moving with velocity $v_j$, and towards the direction with angle
	$\alpha_j\in [\alpha_j^-, \alpha_j^+]$. Each worker $w_j$ is
	associated with a confidence $p_j\in [0, 1]$, which indicates the
	reliability of the worker that can do the task.
	
	\qquad $\blacksquare$ \label{definition:worker}\vspace{-2ex}
\end{definition}

Intuitively, a worker $w_j$ ($1\leq j\leq n$) may want to do tasks
on the way to some place during the trip. Thus, as mentioned in
Definition \ref{definition:worker}, the worker can pre-register the
angle range, $[\alpha_j^-, \alpha_j^+]$, of one's current moving
direction. In other words, the worker is only willing to accomplish
tasks that do not deviate from his/her moving direction
significantly. For other (inconvenient) tasks (e.g., opposite to the
moving direction), the worker tends to ignore/reject the task
request, thus, the system would not assign such tasks to this
worker. In the case that the worker has no targeting destinations
(i.e., free to move), he/she can set $[\alpha_j^-, \alpha_j^+]$ to
$[0, 2 \pi]$.

After being assigned with a spatial task, a worker sometimes may not
be able to finish the task. For example, the worker might reject the
task request (e.g., due to other tasks with higher prices), do the
task incorrectly (e.g., taking a wrong photo), or miss the deadline
of the task. Thus, as given in Definition \ref{definition:worker},
each worker $w_j$ is associated with a confidence $p_j\in [0, 1]$,
which is the probability (or reliability) that $w_j$ can
successfully finish a task (inferred from historical data of this
worker). In this paper, we consider the model of server assigned tasks (SAT) \cite{kazemi2012geocrowd}. That is, we assume that once a worker accepts the 
assigned task, the worker will voluntarily do the task. Similar to spatial 
tasks, workers can freely register or leave the
crowdsourcing system. Thus, the set of workers is also dynamically
changing.

\subsection{Reliable Diversity-Based Spatial Crowdsourcing}
\label{subsec:RDB-SC}

With the definitions of tasks and workers, we are now ready to
formalize our spatial crowdsourcing problem (namely, RDB-SC), which
assigns dynamically moving workers to time-constrained spatial tasks
with high accuracy and quality.

Before we provide the formal problem definition, we first quantify
the criteria of our task assignment during the crowdsourcing, in
terms of two measures, the \textit{reliability} and
\textit{spatial/temporal diversity}. The reliability indicates the
confidence that at least some worker can successfully complete the
task, whereas the spatial/tem\-poral diversity reflects the quality of
the task accomplishment by a group of workers, in both spatial and
temporal dimensions (e.g., taking photos from diverse angles and at
diverse timestamps).

\vspace{0.5ex}\noindent {\bf Reliability}. Since not all workers are
trustable, we should consider the reliability, $p_j$, of each workers,
$w_j$, during the task assignment. For example, some workers might
take a wrong photo, or fail to reach the task location before the
valid period. In such cases, the goal of our task assignment is to
guarantee that tasks $t_i$ can be accomplished by those assigned
workers with high confidence.\vspace{-2ex}

\begin{definition} $($Reliability$)$ Given a spatial task $t_i$ and
	its assigned set, $W_i$, of workers, the \textit{reliability},
	$rel(t_i, W_i)$, of a worker assignment w.r.t. $t_i$ is given
	by:\vspace{-2ex}
	
	{\scriptsize
		\begin{eqnarray}
		rel(t_i, W_i) = 1-\prod_{\forall w_j \in W_i} (1-p_j).
		\label{eq:eq1}
		\end{eqnarray}\vspace{-3ex}
	}
	
	\noindent where $p_j$ is the probability that worker $w_j$ can
	reliably complete task $t_i$. \qquad $\blacksquare$
	\label{definition:reliability}\vspace{-2.5ex}
\end{definition}

Intuitively, Eq.~(\ref{eq:eq1}) gives the probability (reliability)
that there exists some worker who can accomplish the task $t_i$
reliably (e.g., taking the right photo, or providing a reliable
answer). In particular, the second term (i.e., $\prod_{\forall w_j
	\in W_i} (1-p_j)$) in Eq.~(\ref{eq:eq1}) is the probability that all
the assigned workers in $W_i$ cannot finish the task $t_i$. Thus,
$1-\prod_{\forall w_j \in W_i} (1-p_j)$ is the probability that task
$t_i$ can be completed by at least one assigned worker in $W_i$.

High reliability usually leads to good confidence of the task
completion. In this paper, we aim to maximize the reliability for
each individual task.

\underline{\it Possible Worlds of the Task Completion}. Since not
all the assigned workers in $W_i$ can complete the task $t_i$, it is
possible that only a subset of workers in $W_i$ can succeed in
accomplishing the task $t_i$ in the real world. In practice, there
are an exponential number (i.e., $O(2^{|W_i})$) of such possible
subsets.

Following the literature of probabilistic databases \cite{Dalvi07},
we call each possible subset in $W_i$ a \textit{possible world},
denoted as $pw(W_i)$, which contains those workers who may finish
the task $t_i$ in reality. Each possible world, $pw(W_i)$, is
associated with a probability confidence,\vspace{-3ex}

{\scriptsize
	\begin{eqnarray}
	Pr\{pw(W_i)\} = \prod_{\forall w_j\in pw(W_i)} \hspace{-2ex}p_j \cdot
	\prod_{\forall w_j\in (W_i-pw(W_i))} \hspace{-4ex}(1-p_j),\label{eq:eq2}
	\end{eqnarray}\vspace{-3ex}
}

\noindent which is given by multiplying probabilities that workers
in $W_i$ appear or do not appear in the possible world $pw(W_i)$.

\begin{figure}[ht]\vspace{-2ex}\centering
	\subfigure[][{\scriptsize Spatial Diversity}]{
		\scalebox{0.28}[0.28]{\includegraphics{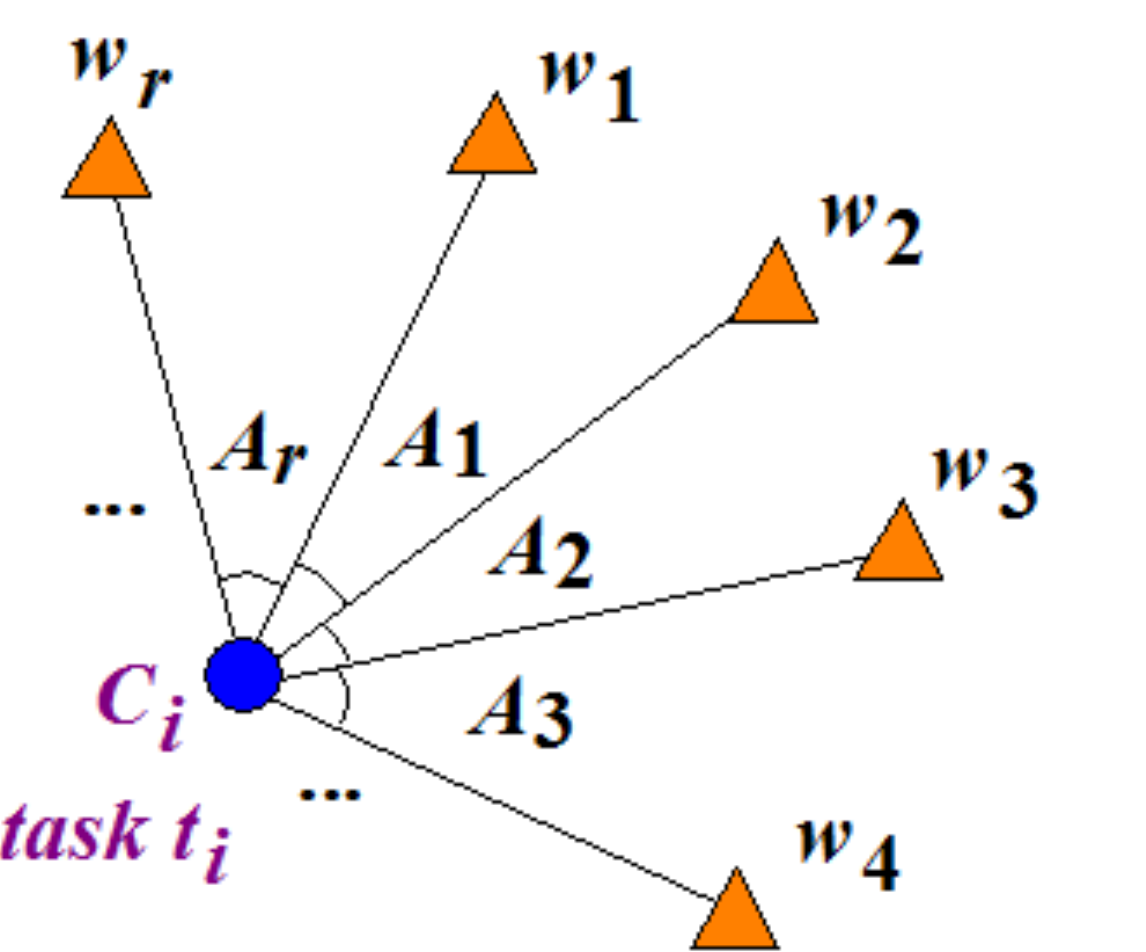}}
		\label{subfig:SD}}
	\subfigure[][{\scriptsize Temporal Diversity}]{
		\scalebox{0.28}[0.28]{\includegraphics{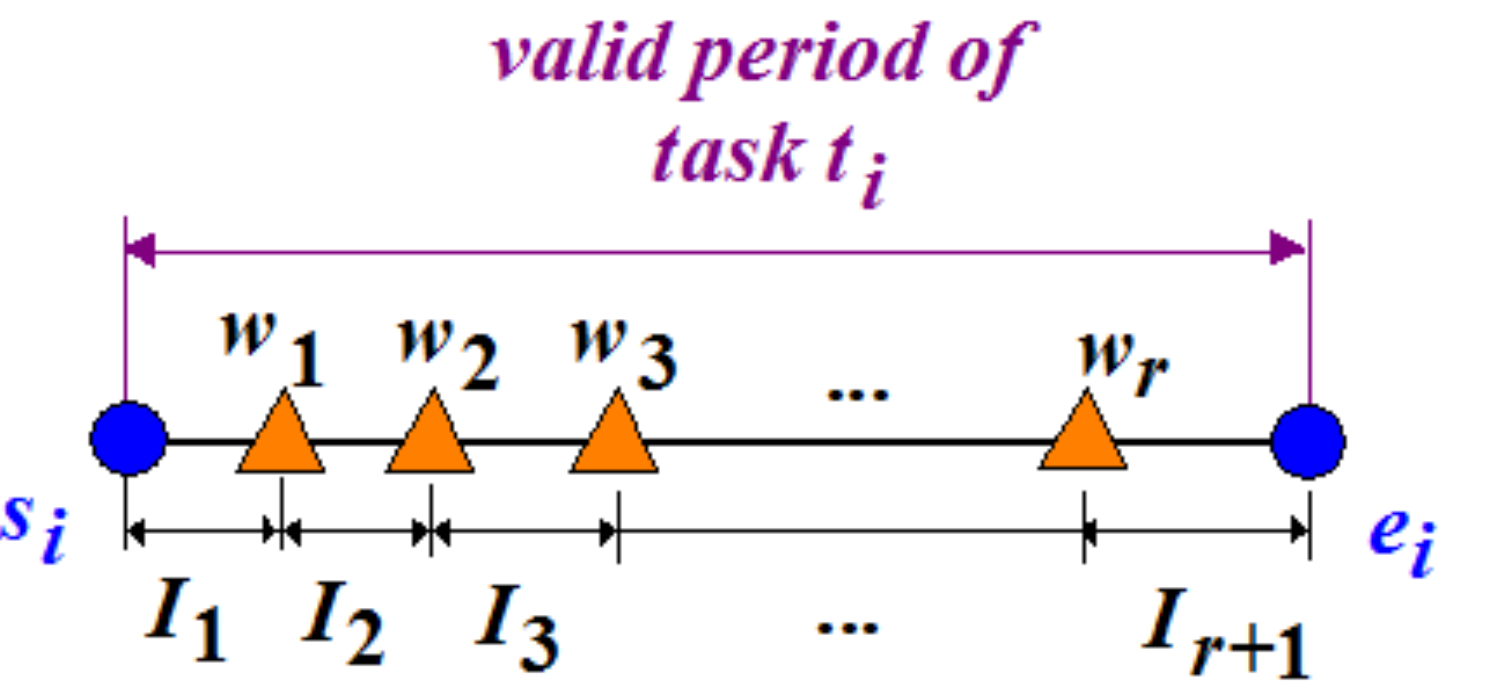}}
		\label{subfig:TD}}\vspace{-3ex}
	\caption{\small Illustration of Spatial/Temporal
		Diversity.}\vspace{-3ex}
	\label{fig:diversity}
\end{figure}

\vspace{0.5ex}\noindent {\bf Spatial/Temporal Diversity}. As
mentioned in Example 1 of Section \ref{sec:introduction} (i.e., take
photos of a statue), it would be nice to obtain photos from
different angles and at diverse times of the day, and get the full
picture of the statue for virtual tours. Similarly, in Example 2
(Section \ref{sec:introduction}), it is desirable to obtain photos
of parking areas from different directions at diverse times, in
order to collect/analyze data of available parking spaces during
open hours. Thus, we want workers to accomplish spatial tasks (e.g.,
taking photos) from different angles and over timestamps as diverse
as possible. We quantify the quality of the task completion by
\textit{spatial/temporal diversities}.

Specifically, the \textit{spatial diversity} (SD) is defined as
follows. As illustrated in Figure \ref{subfig:SD}, let $C_i$ be a
point at the location $l_i$ (or the centroid of a region) of task
$t_i$. Assume that $r$ workers $w_j \in W_i$ ($1\leq j \leq r$) do
tasks (i.e., take pictures) at $l_i$ from different angles. We draw
$r$ rays from $C_i$ to the directions of $r$ workers who take
photos. Then, with these $r$ rays, we can obtain $r$ angles, denoted
as $A_1$, $A_2$, ..., and $A_r$, where $\sum_{j=1}^r A_j=2\pi$.
Intuitively, the entropy was used as an expression of the disorder,
or randomness of a system. Here, higher diversity indicates more
discorder, which can be exactly captured by the entropy. That is,
when the answers come from diverse angles and timestamps, the
answers have high entropy. Thus, we use the entropy to define the
spatial diversity (SD) as follows:\vspace{-4ex}

{\scriptsize
	\begin{eqnarray}
	SD(t_i) = - \sum_{j=1}^r \frac{A_j}{2\pi} \cdot log
	\left(\frac{A_j}{2\pi}\right)\label{eq:eq3}
	\end{eqnarray}\vspace{-3ex}
}

Similarly, we can give the \textit{temporal diversity} for the
arrival times of workers (to do tasks), by using the entropy of time
intervals. As shown in Figure \ref{subfig:TD}, assume that the
arrival times of $r$ workers divide the valid period $[s_i, e_i]$ of
task $t_i$ into $(r+1)$ sub-intervals of lengths $I_1$, $I_2$, ...,
and $I_{r+1}$. We define the temporal diversity (TD)
below:\vspace{-4ex}

{\scriptsize
	\begin{eqnarray}
	TD(t_i) = - \sum_{j=1}^{r+1} \frac{I_j}{e_i-s_i} \cdot log
	\left(\frac{I_j}{e_i-s_i}\right)\label{eq:eq4}
	\end{eqnarray}\vspace{-3ex}
	
}

Intuitively, larger SD or TD value indicates higher diversity of
spatial angles or time distribution, which is more desirable by task
requesters. We combine these 2 diversity types, and obtain the
spatial/temporal diversity (STD) w.r.t. $W_i$, below:\vspace{-3.5ex}

{\scriptsize
	\begin{eqnarray}
	STD(t_i, W_i)=\beta \cdot SD(t_i) +(1-\beta) \cdot
	TD(t_i),\label{eq:eq5}
	\end{eqnarray}\vspace{-5ex}
}

where parameter $\beta \in [0, 1]$. Here, $\beta$ is a
weight balancing between SD and TD, which depends on the application
requirement specified by the task requester of $t_i$. When
$\beta=0$, we consider TD only; when $\beta = 1$, we require SD only
in the spatial task $t_i$.

Therefore, under possible worlds semantics, in this paper, we will
consider the \textit{expected spatial/temporal diversity} (defined
later) in our crowdsourcing problem.

\vspace{0.5ex}\noindent {\bf The RDB-SC Problem.} We define our
RDB-SC problem below.\vspace{-2ex}

\begin{definition} $($Reliable Diversity-Based Spatial Crowdsourcing, RDB-SC$)$ Given $m$
	time-constrained spatial tasks in $T$, and $n$ dynamically moving
	workers in $W$, the problem of \textit{reliable diversity-based
		spatial crowdsourcing} (RDB-SC) is to assign each task $t_i\in T$
	with a set, $W_i$, of workers $w_j \in W$, such that: \vspace{-1ex}
	\begin{enumerate}
		\item each worker $w_j\in W$ is assigned with a spatial
		task $t_i\in T$ such that his/her arrival time at location $l_i$
		falls into the valid period $[s_i, e_i]$,\vspace{-1.5ex}
		\item the minimum reliability, $\min_{i=1}^m rel(t_i, W_i)$, of all tasks
		$t_i$ is maximized, and\vspace{-1.5ex}
		\item the summation, $total\_STD$, of the expected spatial/temporal diversities, $E(STD(t_i))$, for all tasks $t_i$, is
		maximized,\vspace{-1ex}
	\end{enumerate}
	
	\noindent where the expected spatial/temporal diversity
	$E(STD(t_i))$ is:\vspace{-2ex}
	
	{\scriptsize
		\begin{eqnarray}
		E(STD(t_i)) = \sum_{\forall pw(W_i)} Pr\{pw(W_i)\} \cdot STD(t_i,
		pw(W_i)),
		\text{ and}\label{eq:eq6} \vspace{-3ex}
		\end{eqnarray}\vspace{-6ex}
	}
	{\scriptsize
		\begin{eqnarray}
		total\_STD = \sum_{i=1}^m E(STD(t_i)). \label{eq:eq7}
		\end{eqnarray}\vspace{-4ex}
	}

	\label{definition:crowdsourcing}\vspace{-2ex}
\end{definition}

The RDB-SC problem is to assign workers to collaboratively
accomplish each task with two optimization goals: (1) the smallest
reliability among all tasks is maximized (intuitively, if the
smallest reliability of tasks is maximized, then the reliability of
all tasks must be high), and (2) the summed expected diversity for
all tasks is maximized. These two goals aim to guarantee the
confidence of the task completion and the diversity quality of the
tasks, respectively.

\nop{
	\subsection{Answer Aggregation}
	\label{subsec:aggregation}
	
	Up to now, we only describe the problem of assigning tasks to
	optimal workers, which is the main work of this paper, but how to
	return the result to a task requester is also an important problem.
	Here we just propose a simple method to aggregate the answers and
	guarantee the diversity and reliability of the task at the same
	time.
	
	We first divide the answers based on their angles and timestamps
	into $g$ groups. For example, we divide the angles into 4 groups
	with angle range of $[\frac{i\cdot\pi}{2},
	\frac{(i+1)\cdot\pi}{2})$, where $i = \{0,1,2,3\}$. Each group
	contains answers within the direction range. Similarly, we divide
	the time period into several groups and classify the answers to
	those groups. Then, we return to the requester one answer from each
	angle and time group with the highest confidence in that group.
	
}

\vspace{0.5ex}\noindent {\bf Answer Aggregation for a Spatial Task.}
After assigning a set, $W_i$, of workers to spatial task $t_i$, we
can obtain a set of answers, for example, photos in Example 1 of
Section \ref{sec:introduction} with different angles and at diverse
timestamps. It is also important to present these photos to the task
requester. Due to too many photos, we can apply aggregation
techniques to group those photos with similar spatial/temporal
diversities, and show the task requester only one representative
photo from each group. Moreover, since we consider taking 
photos as tasks, the task requester can choose 
the photos with high quality (e.g., resolution or sharpness) from all the answers if he/she wants to. 
This is, however, not the focus of this
paper, and we would like to leave it as future work.

\subsection{Challenges}
\label{subsec:challenges}

According to Definition \ref{definition:crowdsourcing}, the RDB-SC
problem is an optimization problem with two objectives. The
challenges of tackling the RDB-SC problem are threefold. First, with
$m$ time-constrained spatial tasks and $n$ moving workers, in the
worst case, there are an exponential number of possible task-worker
assignment strategies, that is, with time complexity $O(m^n)$. In
fact, we will later prove that the RDB-SC problem is NP-hard. Thus,
it is inefficient or even infeasible to enumerate all possible
assignment strategies.

Second, the second objective in the RDB-SC problem considers maximizing the summed
expected spatial/temporal diversities, which involves an exponential
number of possible worlds for diversity computations. That is, the 
time complexity of enumerating
possible worlds is $O(2^{|W_i|})$, where $W_i$ is a set of assigned
workers to task $t_i$. Therefore, it is not efficient to
compute the spatial/temporal diversity by taking into account all
possible worlds.

Third, in the RDB-SC problem, workers move towards some directions,
whereas spatial tasks are restricted by time constraints (i.e.,
valid period $[s_i, e_i]$). Both workers and tasks can enter/quit
the spatial crowdsourcing system dynamically. Thus, it is also
challenging to dynamically decide the task-worker assignment.

Inspired by the challenges above, in this paper, we first prove that
the RDB-SC problem is NP-hard, and design three efficient
approximation algorithms, which are based on greedy, sampling, and
divide-and-conquer approaches. To tackle the second challenge, we
reduce the computation of the expected diversity under possible
worlds semantics to the one with cubic cost, which can greatly
improve the problem efficiency. Finally, to handle the cases that
workers and tasks can dynamically join and leave the system freely,
we design effective grid index to enable dynamic updates of
worker/tasks, as well as the retrieval of assignment pairs.

Table \ref{table0} summarizes the commonly used symbols.

\begin{table}
	\begin{center}\vspace{-2ex}
		\caption{Symbols and descriptions.} \label{table0}\vspace{-2ex}
		{\small\scriptsize
			\begin{tabular}{l|l} 
				{\bf Symbol} & {\bf \qquad \qquad \qquad Description} \\ \hline \hline
				$T$   & a set of $m$ time-constrained spatial tasks $t_i$ \\ 
				$W$   & a set of $n$ dynamically moving workers $w_j$ \\ 
				$[s_i, e_i]$   & the time constraint of accomplishing a task $t_i$\\ 
				$l_i$ (or $l_j$)   & the position of task $t_i$ (or worker $w_j$) \\ 
				$v_j$   & the velocity of moving worker $w_j$ \\ 
				$p_j$   & the confidence of worker $w_j$ that reliably do the task \\ 
				$[\alpha_j^-, \alpha_j^+]$   & the interval of moving direction (angle)\\ 
				$rel(t_i, W_i)$ & the reliability that task $t_i$ can be
				completed by workers
				in $W_i$ \\
				$R(t_i, W_i)$ & a (equivalent) variant of reliability $rel(t_i, W_i)$ \\
				$STD(t_i, W_i)$ & the spatial/temporal diversity of task $t_i$\\
				$E(STD(t_i))$ & the expected spatial/temporal diversity of task $t_i$\\
				$total\_STD$ & the sum of the expected spatial/temporal diversities for all  tasks\\
				$K$   & the sample size\\ \hline
				\hline
			\end{tabular}
		}
	\end{center}\vspace{-6ex}
	
\end{table}

\section{Problem Reduction}
\label{sec:reduction}\vspace{-1ex}

\subsection{Reduction of Reliability}
\label{subsec:reduction1}

As mentioned in Definition \ref{definition:crowdsourcing}, the first
optimization goal of the RDB-SC problem is to maximize the minimum
reliability among $m$ tasks. Since it holds that the reliability
$rel(t_i, W_i) = 1-\prod_{w_j\in W_i} (1-p_j)$ (given in
Eq.~(\ref{eq:eq1})), we can rewrite it as:\vspace{-2.5ex}

{\scriptsize
	\begin{eqnarray}
	R(t_i, W_i) = - ln (1- rel(t_i, W_i)) = \sum_{w_j\in W_i} -ln
	(1-p_j).\label{eq:eq8}
	\end{eqnarray}\vspace{-3ex}
}

From Eq.~(\ref{eq:eq8}), our goal (w.r.t. reliability) of maximizing
the smallest $rel(t_i, W_i)$ for all tasks $t_i$ is equivalent to
maximizing the smallest $- ln (1- rel(t_i, W_i))$ (i.e., LHS of
Eq.~(\ref{eq:eq8})) for all $t_i$. In turn, we can maximize the
smallest $\sum_{w_j\in W_i} -ln (1-p_j)$ (i.e., RHS of
Eq.~(\ref{eq:eq8})) among all tasks.

Intuitively, we associate each worker $w_j$ with a positive
constant $-ln(1-p_j)$. Then, we divide $n$ workers into $m$
disjoint partitions (subsets) $W_i$ ($1\leq i\leq n$), and each
subset $W_i$ has a summed value $\sum_{w_j\in W_i} -ln (1-p_j)$
(i.e., RHS of Eq.~(\ref{eq:eq8})). As a result, our equivalent
reliability goal is to find a partitioning strategy such that the
smallest summed value among $m$ subsets is maximized.

\subsection{Reduction of  Diversity}
\label{subsec:reduction2}

As discussed in Section \ref{subsec:challenges}, the direct
computation of the expected diversity involves an exponential number
of possible worlds $pw(W_i)$ (see Eq.~(\ref{eq:eq6})). It is thus
not efficient to enumerate all possible worlds. In this subsection,
we will reduce such a computation to the problem with polynomial cost.

Specifically, in order to compute the expected spatial diversity,
$E(SD(t_i))$, of a task $t_i$, we introduce a spatial diversity
matrix, $M_{SD}$, in which each entry $M_{SD} [j] [k]$ stores a
value, given by the multiplication of the probability that an angle
$A_{j,k} = (\sum_{x=j}^{k+r} A_{(x \% r)}) \% 2\pi$ exists (in possible worlds) and
entropy $-(\frac{A_{j,k}}{2\pi}) \cdot log(\frac{A_{j,k}}{2\pi})$, where $x \% y$ is
$x$ mod $y$.

In particular, we have:\vspace{-4.5ex}

{\scriptsize
	\begin{eqnarray}
	M_{SD} [j] [k] = -(\frac{A_{j,k}}{2\pi}) \cdot
	log(\frac{A_{j,k}}{2\pi}) \cdot p_j \cdot p_k   \cdot
	\prod_{x=j+1}^{(k+r-1)\% r}(1-p_x). \label{eq:eq9}
	\end{eqnarray}\vspace{-2ex}
}

The time complexity of computing $M_{SD} [j] [k]$ is $O(r)$. Thus,
the total cost of spatial diversity matrix is $O(r^3)$.

Similarly we can compute the temporal diversity matrix, $M_{TD}$, in
which each entry $M_{TD}[j][k]$ ($j\leq k$) is given by the
multiplication of the probability that a time interval
$I_{j,k} = \bigcup_{x=j}^k I_x$ exists in possible worlds and entropy
$-(\frac{I_{j,k}}{e_i-s_i}) \cdot log (\frac{I_{j,k}}{e_i-s_i})$,
for $j\leq k$; moreover, $M_{TD}[j][k] = 0$,if $j > k$.

Formally, for $j\leq k$, we have:\vspace{-3ex}

{\scriptsize
	\begin{eqnarray}
	M_{TD} [j] [k] = -(\frac{I_{j,k}}{e_i-s_i}) \cdot log
	(\frac{I_{j,k}}{e_i-s_i}) \cdot p_k \cdot \prod_{x=j+1}^{(k-1)}
	(1-p_x) \label{eq:eq10}.
	\end{eqnarray}\vspace{-3ex}
}

To compute $E(STD(t_i))$, the expected spatial/temporal diversity,
we have the following lemma. \vspace{-2ex}

\begin{lemma} (Expected Spatial/Temporal Diversity)
	The expected spatial/temporal diversity, $E(STD(t_i))$, of task
	$t_i$ is given by:\vspace{-2ex}
	
	{\scriptsize
		\begin{eqnarray}
		E(STD(t_i))&=& \beta \cdot E(SD(t_i)) + (1-\beta) \cdot E(TD(t_i))\label{eq:eq11}\\
		&=& \beta \cdot \sum_{\forall j, k} M_{SD} [j] [k] + (1-\beta) \cdot
		\sum_{\forall j, k}  M_{TD} [j] [k].\notag
		\end{eqnarray}\vspace{-2ex}
	}
	
	\label{lemma:lem1}
\end{lemma}
\begin{proof}
	Please refer to Appendix A of the technical report \cite{aixivReport}. 
\end{proof}
\text{}\vspace{-3ex}

In this subsection, we prove that the hardness of our RDB-SC problem
is NP-hard. Specifically, we can reduce the problem of the
\textit{number partition problem} \cite{mertens2006easiest} (which
is known to be an NP-hard problem) to our RDB-SC problem. This way,
our RDB-SC problem is also an NP-hard problem: \vspace{-2ex}

\begin{lemma} (Hardness of the RDB-SC Problem)
	The problem of the reliable diversity-based spatial crowdsourcing
	(RDB-SC) is NP-hard. \label{lemma:lem2}\vspace{-1ex}
\end{lemma}
\begin{proof}
	Please refer to Appendix B of the technical report \cite{aixivReport}.\vspace{-1ex}
\end{proof}

From Lemma \ref{lemma:lem2}, we can see that the RDB-SC problem is
not tractable. Therefore, in the sequel, we aim to propose
approximation algorithms to find suboptimal solution efficiently.

\vspace{-1ex}
\section{The Greedy Approach}
\label{sec:greedy}\vspace{-1ex}

\subsection{Properties of Optimization Goals}
\label{subsec:properties}\vspace{-1ex}

In this subsection, we provide the properties about the reliability
and the expected spatial/temporal diversity. Specifically, assume
that a task $t_i$ is assigned with a set, $W_i$, of $r$ workers
$w_j$ ($1\leq j\leq r$). Let $w_{r+1}$ be a new worker (with
confidence $p_{r+1}$ who is also assigned to task $t_i$.

\vspace{0.5ex}\noindent {\bf Reliability.} We first give the
property of the reliability upon a newly assigned worker.\vspace{-2.5ex}

\begin{lemma} (Property of the Reliability)
	Let $R(t_i, W_i)$ be the reliability of task $t_i$ (given in
	Eq.~(\ref{eq:eq8}), in the reduced goal of Section
	\ref{subsec:reduction1}), associated with a set, $W_i$, of $r$
	workers. If a new worker $w_{r+1}$ is assigned to $t_i$, then we
	have: \vspace{-3.5ex}
	
	{\scriptsize
		\begin{eqnarray}
		R(t_i, W_i\cup \{w_{r+1}\}) = R(t_i, W_i) - ln
		(1-p_{r+1}).\label{eq:eq12}
		\end{eqnarray}\vspace{-7ex}
	}
	
	\label{lemma:lem3}
\end{lemma}

\begin{proof}
	Please refer to Appendix C of the technical report \cite{aixivReport}. 
\end{proof}
\text{}\vspace{-5ex}

\nop{
	\begin{proof} According to the equivalent definition of reliability in
		Eq.~(\ref{eq:eq8}), it holds that: $$R(t_i, W_i) =  \sum_{w_j\in
			W_i} -ln (1-p_j).$$ Moreover, for the new set
		$(W_i\cup\{w_{r+1}\})$, we have: $$R(t_i, W_i\cup \{w_{r+1}\}) =
		\sum_{w_j\in (W_i\cup \{w_{r+1}\})} \hspace{-4ex}-ln (1-p_j).$$
		
		By combining the two formulae above, we can infer that: $R(t_i,
		W_i\cup \{w_{r+1}\}) = R(t_i, W_i) - ln (1-p_{r+1})$. Hence, the
		lemma holds.
	\end{proof}
	
}

From Eq.~(\ref{eq:eq12}) in Lemma \ref{lemma:lem3}, we can see that
the second term (i.e., $- ln (1-p_{r+1})$) is a positive value.
Thus, it indicates that when we assign more workers (e.g., $w_{r+1}$
with confidence $p_{r+1}\leq 0$) to task $t_i$, the reliability,
$R(t_i, \cdot)$, of the task $t_i$ is always increasing (at least
non-decreasing).

\vspace{0.5ex}\noindent {\bf Diversity.} Next, we give the property
of the expected spatial/temporal diversity, upon a newly assigned
worker.\vspace{-2.5ex}

\begin{lemma} (Property of the Expected Spatial/Temporal Diversity)
	Let $E(STD(t_i))$ be the expected spatial/temporal diversity of task
	$t_i$. Upon a newly assigned worker $w_{r+1}$ with confidence
	$p_{r+1}$, the expected diversity $E(STD(t_i))$ of task $t_i$ is
	always non-decreasing, that is, $E(STD(t_i, W_i\cup\{w_{r+1}\}))\geq
	E(STD(t_i, W_i))$. \label{lemma:lem4}\vspace{-4ex}
\end{lemma}
\begin{proof}
	Please refer to Appendix D of the technical report \cite{aixivReport}.\vspace{-2.5ex}
\end{proof}

Lemma \ref{lemma:lem4} indicates that when we assign a new worker to
a spatial task $t_i$, the expected spatial/temporal diversity is
non-decreasing.

\begin{figure}[ht]
	\begin{center}
		\begin{tabular}{l}
			\parbox{3.1in}{
				\begin{scriptsize}
					\begin{tabbing}
						12\=12\=12\=12\=12\=12\=12\=12\=12\=12\=12\=\kill
						{\bf Procedure {\sf RDB-SC\_Greedy}} \{ \\
						\> {\bf Input:} $m$ time-constrained spatial tasks in $T$ and $n$ workers in $W$\\
						\> {\bf Output:} a task-and-worker assignment strategy, $\mathbb{S}$, \\
						\>\>\>\>\hspace{2ex}with high reliability and diversity\\
						\> (1) \> \> $\mathbb{S} = \emptyset$\\
						\> (2) \> \> compute all the valid task-and-worker pairs $(t_i, w_j)$\\
						\> (3) \> \> for $i=1$ to $n$ \\
						\>\>\>\textit{// in each round, select one best task-and-worker pair}\\
						\> (4) \> \> \> for each pair $(t_i, w_j)$ ($w_j\in W$)\\
						\> (5) \> \> \> \> compute the increase pair ($\Delta R(t_i, w_j), \Delta STD(t_i, w_j)$) \\
						\> (6) \> \> \> prune ($\Delta R(t_i, w_j), \Delta STD(t_i, w_j)$) dominated by others\\
						\> (7) \> \> \> rank the remaining pairs by their scores (i.e., the number \\
						\>\>\>\>of dominated pairs) \\
						\> (8) \> \> \> select a pair, $(t_i, w_j)$, with the highest score and add it to $\mathbb{S}$\\
						\> (9) \> \> \> $W = W-\{w_j\}$\\
						\> (10) \> \> return $\mathbb{S}$\\
						\}
					\end{tabbing}
				\end{scriptsize}
			}
		\end{tabular}
	\end{center}\vspace{-8ex}
	\caption{\small RDB-SC Greedy Algorithm.}
	\label{alg:greedy}\vspace{-5ex}
\end{figure}

\subsection{The Greedy Algorithm}
\label{subsec:greedy_algorithm}

As mentioned in Section \ref{subsec:properties}, when we assign more
workers to a spatial task, the reliability and diversity of the
assignment strategy is always non-decreasing. Based on these
properties, we propose a greedy algorithm, which iteratively assigns
workers to spatial tasks that can always achieve high ranks (w.r.t.
reliability and diversity).

Figure \ref{alg:greedy} illustrates the pseudo code of our RDB-SC
greedy algorithm, namely {\sf RDB-SC\_Greedy}, which returns one
best strategy, $\mathbb{S}$, containing task-and-worker assignments
with high reliability and diversity. Specifically, our greedy
algorithm iteratively finds one pair of task and worker such that
the assignment with this pair can increase the reliability and
diversity most.

Initially, there is no task-and-worker assignment, thus, we set
$\mathbb{S}$ to empty (line 1). Next, we identify all the valid
task-and-worker pairs $(t_i, w_j)$ in the crowdsourcing system (line
2). Here, the validity of pair $(t_i, w_j)$ means that worker $w_j$
can reach the location of task $t_i$, under the constraints of both
moving directions and valid period. Then, among these pairs, we want
to incrementally select $n$ best task-and-worker assignments such
that the increases of reliability and diversity are always maximized
(lines 3-9).

In particular, in each iteration, for every task-and-worker pair
$(t_i, w_j)$ ($w_j \in W$ is a worker who has no task), if we allow the assignment of worker $w_j$ to $t_i$, we
can calculate the increases of the reliability and diversity
($\Delta R(t_i, w_j), \Delta STD(t_i,$ $w_j)$) (lines 4-5), where
$\Delta R(t_i, w_j) = R(\mathbb{S}\cup \{(t_i,
w_j\})-R(\mathbb{S})$, and $\Delta STD(t_i, w_j) =
STD(\mathbb{S}\cup \{(t_i, w_j\})-STD(\mathbb{S})$. Note that, as
guaranteed by Lemmas \ref{lemma:lem3} and \ref{lemma:lem4}, here the
two optimization goals are always non-decreasing (i.e., $\Delta
R(\cdot)$ and $\Delta STD(\cdot)$ are positive).

Since some increase pairs may be \textit{dominated}
\cite{Borzsonyi01} by others, we can safely filter out such false
alarms with both lower reliability and diversity (line 6). We say
assignment $\mathbb{S}_i$ dominates $\mathbb{S}_j$ when
$R(\mathbb{S}_i) > R(\mathbb{S}_j)$ and $STD(\mathbb{S}_i) \geq
STD(\mathbb{S}_j)$, or $R(\mathbb{S}_i) \geq R(\mathbb{S}_j)$ and
$STD(\mathbb{S}_i) > STD(\mathbb{S}_j)$. If there are more than one
remaining pair, we rank them according to the number of pairs that
they are dominating \cite{Yiu07b} (line 7). Intuitively, the pair
(i.e., assignment) with higher rank indicates that this assignment
is better than more other assignments. Thus, we add a pair $(t_i,
w_j)$ with the highest rank to $\mathbb{S}$, and remove the worker
$w_j$ from $W$ (lines 8-9).

The selection of pairs (assignments) repeats for $n$ rounds (line
3). In each round, we find one assignment pair that can locally
increase the maximum reliability and diversity. Finally, we return
$\mathbb{S}$ as the best RDB-SC assignment strategy (line 10).

\vspace{0.5ex}\noindent {\bf The Time Complexity.} The time
complexity of computing the best task-and-worker pair in each
iteration is given by $O(m\cdot n)$ in the worst case (i.e., each of
$n$ worker can be assigned to any of the $m$ tasks). Since we only
need to select $n$ task-and-worker pairs (each worker can only be
assigned to one task at a time), the total time complexity of our
greedy algorithm is given by $O(m\cdot n^2)$.

\text{}\vspace{-3ex}
\subsection{Pruning Strategies}
\label{subsec:pruning}

Note that, to compute the exact increase of the reliability $R(t_i,
W_i)$, we can immediately obtain the reliability increase of task
$t_i$: $\Delta R(t_i,$ $w_j) = -ln(1-p_j)$. For the diversity,
however, it is not efficient to compute the exact increase, $\Delta
STD(t_i, w_j)$, since we need to update the
diversity matrices (as mentioned in Section \ref{subsec:reduction2})
before/after the worker insertion. Therefore, in this subsection, we
present an effective pruning method to reduce the search space
without calculating the expected diversity for every $(t_i, w_j)$
pair.

Our basic idea of the pruning method is as follows. For any
task-and-worker pair $(t_i, w_j)$, assume that we can quickly
compute its lower and upper bounds of the increase for the expected
spatial/temporal diversity, denoted as $lb\_\Delta D(t_i, w_j)$ and
$ub\_\Delta D(t_i,$ $w_j)$, respectively.

Then, for two pairs $(t_i, w_j)$ and $(t_i', w_j')$, if it holds
that $lb\_\Delta D(t_i,$ $w_j)> ub\_\Delta D(t_i', w_j')$, then the
diversity increase of pair $(t_i', w_j')$ is inferior to that of
pair $(t_i, w_j)$.

We have the pruning lemma below.\vspace{-2ex}

\begin{lemma} (Pruning Strategy) Assume that $lb\_\Delta D(t_i, w_j)$ and $ub\_\Delta D(t_i,$ $w_j)$ are lower
	and upper bounds of the increase for the expected spatial/temporal
	diversity, respectively. Similarly, let $\Delta min\_R(t_i, w_j)$ be
	the increase of the smallest reliability among $m$ tasks after assigning worker $w_i$ to $t_i)$,
	respectively. Then, given two pairs $(t_i, w_j)$ and $(t_i', w_j')$,
	if it holds that: (1) $\Delta min\_R(t_i, w_j) \\ 
	\geq \Delta min\_R(t_i', w_j')$, and (2) $lb\_\Delta D(t_i, w_j) > ub\_\Delta
	D(t_i', w_j')$, then we can safely prune the pair $(t_i', w_j')$.
	\label{lemma:lem5}
\end{lemma}

\begin{proof}
	Please refer to Appendix E of the technical report \cite{aixivReport}.\vspace{-1ex}
\end{proof}

\vspace{0.5ex}\noindent {\bf The Computation of Lower/Upper Bounds
	for the Diversity Increase.} From Eq.~(\ref{eq:eq6}), we can
alternatively compute the lower/upper bounds of $E(STD(t_i))$ before
and after assigning worker $w_j$ to $t_i$. From Lemma
\ref{lemma:lem3}, we know that the maximum diversity is achieved
when the maximum number of workers are assigned to task $t_i$. Thus,
in each possible world $pw(W_i)$, the upper bound of diversity
$STD(t_i, pw(W_i))$ is given by $STD(t_i, W_i)$. Thus, we have
$ub\_E(STD(t_i)) = STD(t_i, W_i)$.

Moreover, from Eq.~(\ref{eq:eq6}), the lower bound,
$lb\_E(STD(t_i))$, of the diversity is given by the probability that
$STD(t_i, pw(W_i))$ is not zero in possible worlds times the minimum
possible non-zero diversity. Note that, $STD(t_i, pw(W_i))$ is zero,
when none or one worker is reliable. Thus, the minimum possible
non-zero spatial diversity is achieved when we assign two workers to
task $t_i$. In this case, two angles, $\min_{j=1}^r A_j$ and
$(1-\min_{j=1}^r A_j)$ can achieve the smallest diversity (i.e.,
entropy), which can be computed with $O(r)$ cost. The smallest
non-zero temporal diversity is achieved when one worker is assigned
to the task. The computation cost is also $O(r)$, where $r$ is
the maximum number of workers for task $t_i$.

After obtaining lower/upper bounds of the expected diversity, we can
thus compute bounds of the diversity increase. We use subscript
``b'' and ``a'' to indicate the measures before/after the worker
assignment, respectively. The bounds are:\vspace{-3ex}

{\scriptsize
	$$lb\_\Delta D(t_i, w_j) = lb\_E_a(STD(t_i)) -
	ub\_E_b(STD(t_i)),$$\vspace{-5.5ex}
	$$ub\_\Delta D(t_i, w_j) = ub\_E_a(STD(t_i)) - lb\_E_b(STD(t_i)).$$\vspace{-5ex}
}

Therefore, instead of computing the exact diversity values for all
task-and-worker pairs with high cost, we now can utilize their
lower/upper bounds to derive bounds of their increases, and in turn
filter out false alarms by Lemma \ref{lemma:lem5}.

\begin{figure}[ht]
	\centering \vspace{-2ex}
	\scalebox{0.32}[0.28]{\includegraphics{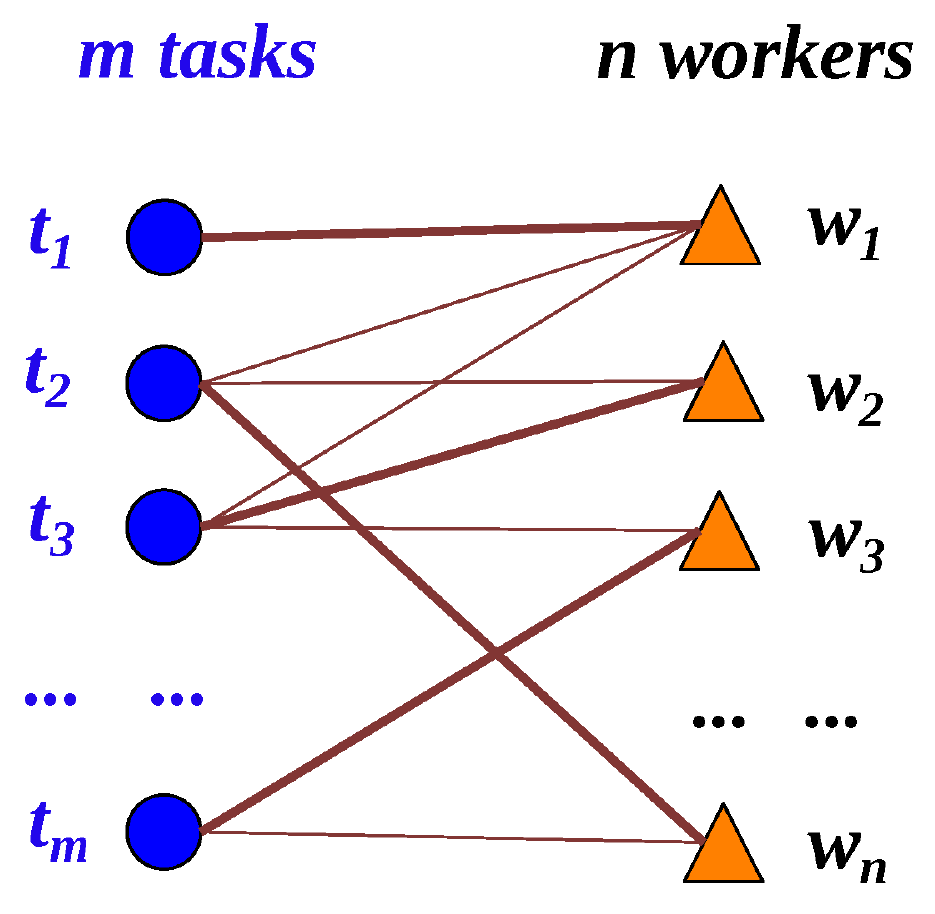}}\vspace{-3ex}
	\caption{\small Illustration of the Task-and-Worker
		Assignment.}\vspace{-4ex}
	\label{fig:assignment}
\end{figure}

\section{The Sampling Approach}
\label{sec:sampling}\vspace{-1ex}

\subsection{Random Sampling Algorithm}
\label{subsec:sampling}

In this subsection, we illustrate how to obtain a good
task-and-worker assignment strategy with high reliability and
diversity by random sampling. Specifically, in our RDB-SC problem,
all possible task-and-worker assignments correspond to the
population, where each assignment is associated with a value (i.e.,
reliability or diversity). As shown in Figure \ref{fig:assignment}, we
denote $m$ tasks $t_i$ and $n$ workers $w_j$ by nodes of circle and
triangle shapes, respectively. The edge between two types of nodes
indicates that worker $w_j$ can arrive at the location of task $t_i$
within the time period (and with correct moving direction as well).
Since each worker can be only assigned to one task, for each worker
node $w_j$, we can select one of $deg(w_j)$ edges connecting to it
(represented by bold edges in the figure), where $deg(w_j)$ is the
degree of the worker node $w_j$. As a result, we can obtain $n$
selected edges (as shown in Figure \ref{fig:assignment}), which
correspond to one possible assignment of workers to tasks.

Due to the exponential number ($O(\prod_{j=1}^n deg(w_j))$) of
possible assignments (i.e., large size of the population), it is not
feasible to enumerate all assignments, and find an optimal
assignment with high reliability and diversity. Alternatively, we
adopt the sampling techniques, and aim at obtaining $K$ random
samples from the entire population such that among these $K$ samples
there exists a sample with error-bounded ranks of reliability or
diversity.

Figure \ref{alg:sampling} illustrates the pseudo code of our
sampling algorithm, {\sf RDB-SC\_Sampling}, to tackle the
RDB-SC problem. Specifically, we obtain each random sample (i.e.,
task-and-worker assignment), $S_h$ ($1\leq h\leq K$), from the
entire population as follows. For each worker $w_j$ ($1\leq j\leq
n$), we first randomly generate an integer $x$ between $[1,
deg(w_j)]$, and then select the $x$-th edge that connecting two
nodes $t_i$ and $w_j$ (lines 5-7). After selecting $n$ edges for $n$
workers, respectively, we can obtain one possible assignment, which
is exactly a random sample $S_h$ ($1\leq h\leq K$). Here, the random
sample $S_h$ is chosen with a probability $p=\prod_{j=1}^n
\frac{1}{deg(w_j)}$. Given the sample (assignment) $S_h$, we can
compute its reliability and diversity.

We repeat the sampling process above, until $K$ samples (i.e.,
assignments) are obtained. After obtaining $K$ samples, we rank them
with the ranking scores \cite{Yiu07b} (w.r.t. reliability and
diversity) (line 8). Let $S_h$ be the sample with the highest score
(line 9). We can return this sample (assignment) as the answer to
the RDB-SC problem (line 10).

\begin{figure}[ht]
	\begin{center}\vspace{-3ex}
		\begin{tabular}{l}
			\parbox{3.1in}{
				\begin{scriptsize}
					\begin{tabbing}
						12\=12\=12\=12\=12\=12\=12\=12\=12\=12\=12\=\kill
						{\bf Procedure {\sf RDB-SC\_Sampling}} \{ \\
						\> {\bf Input:} $m$ time-constrained spatial tasks in $T$ and $n$ workers in $W$\\
						\> {\bf Output:} a task-and-worker assignment strategy, $\mathbb{S}$,with high reliability and diversity\\
						\> (1) \> \> $\mathbb{S} = \emptyset$\\
						\> (2) \> \> compute all the valid task-and-worker pairs $(t_i, w_j)$\\
						\> (3) \> \> for $h=1$ to $K$  \textit{// in each round, obtain one random sample $S_h$}\\
						\> (4) \> \> \> $S_h = \emptyset$\\
						\> (5) \> \> \> for each worker $w_j$ ($\in W$)\\
						\> (6) \> \> \> \> randomly select a task $t_i$ with probability $\frac{1}{deg(w_j)}$\\
						\> (7) \> \> \> \> $S_h = S_h \cup \{(t_i, w_j\}$\\
						\> (8) \> \> rank $S_h$ ($1\leq h\leq K$) by dominating scores among samples \\
						\> (9) \> \> let $\mathbb{S}$ be the sample, $S_h$, with the highest score\\
						\> (10) \> \> return $\mathbb{S}$\\
						\}
					\end{tabbing}
				\end{scriptsize}
			}
		\end{tabular}
	\end{center}\vspace{-8 ex}
	\caption{\small RDB-SC Sampling Algorithm.}
	\label{alg:sampling}\vspace{-2ex}
\end{figure}

Intuitively, when the sample size $K$ is approaching the population
size (i.e., $\prod_{j=1}^n deg(w_j)$), we can obtain RDB-SC answers
close to the optimal solution. However, since RDB-SC is NP-hard and
intractable (as proved in Lemma \ref{lemma:lem2}), we alternatively
aim to find approximate solution via samples with bounded rank
errors. Specifically, our target is to determine the sample size $K$
such that the sample with the maximum optimization goal (reliability
or diversity) has the rank within the $(\epsilon,
\delta)$-error-bound.

\subsection{Determination of Sample Size}
\label{subsec:sample_size}

Without loss of generality, assume that we have the population of
size $N$ (i.e., $N=\prod_{j=1}^n deg(w_j)$), $V_1$, $V_2$, ..., and
$V_N$, which correspond to the reliabilities/diversities of all
possible task-and-worker assignments, where $V_1\leq V_2 \leq ...
\leq V_N$. Then, for each value $V_i$ ($1\leq i\leq N$), we flip a
coin. With probability $p$, we accept value $V_i$ as the selected
random sample; otherwise (i.e., with probability $(1-p)$), we reject
value $V_i$, and repeat the same sampling process for the next value
(i.e., $V_{i+1}$). This way, we can obtain $K$ samples, denoted as
$S_1$, $S_2$, ..., and $S_K$, where $S_1\leq S_2 \leq ... \leq S_K$.

Our goal is to estimate the required minimum number of samples,
$\widehat{K}$, such that the rank of the largest sample $S_K$ is
bounded by $\epsilon N$ in the population (i.e., within
$((1-\epsilon) \cdot N, N]$) with probability greater than $\delta$.

Let variable $X$ be the rank of the largest sample, $S_K$, in the
entire population. We can calculate the probability that $X =
r$:\vspace{-3ex}

{\scriptsize
	\begin{eqnarray}
	Pr\{X=r\}&\hspace{-2ex}=&\hspace{-2ex}\dbinom{r-1}{K-1} \cdot
	p^{K-1}\cdot (1-p)^{r-K}\cdot
	p\cdot (1-p)^{N-r} \hspace{-5ex}\notag\\
	&\hspace{-2ex}=& \hspace{-2ex}\dbinom{r-1}{K-1} \cdot p^K \cdot
	(1-p)^{N-K}.\label{eq:eq13}
	\end{eqnarray}\vspace{-4ex}
}

Intuitively, the first 3 terms above is the probability that $(K-1)$
out of $(r-1)$ values are selected from the population before
(smaller than) $S_K$ (i.e., $V_1 \sim V_{r-1}$). The fourth term
(i.e., $p$) is the probability that the $r$-th largest value $V_r$
($=S_K$) is selected. Finally, the last term is the probability that
all the remaining $(N-r)$ values (i.e., $V_{r+1}\sim V_N$) are not
sampled.

With Eq.~(\ref{eq:eq13}), the cumulative distribution function of
variable $X$ is given by:\vspace{-6ex}

{\scriptsize
	\begin{eqnarray}
	Pr\{X\leq r\}= \sum_{i=1}^r Pr\{X=i\}.\label{eq:eq14}
	\end{eqnarray}\vspace{-3ex}
	
}

Now our problem is as follows. Given parameters $p$, $\epsilon$, and
$\delta$, we want to decide the value of parameter $K$ with high
confidence. That is, we have:$Pr\{X > (1-\epsilon)\cdot N\} > \delta. $

By applying the combination theory and Harmonic series, we can rewrite the formula $Pr\{X >(1-\epsilon) \cdot N\} > \delta$, and derive the following formula w.r.t. $K$:
\vspace{-3ex}

{\scriptsize
	\begin{eqnarray}
	K>\frac{p\cdot M\cdot e -1+p}{1-p+e\cdot p},\label{eq:eq20}
	\end{eqnarray} 
}

\vspace{-3ex}

\noindent where $M = (1-\epsilon)\cdot N$, and $e$ is the base of the natural logarithm. Please refer the detailed derivation to Appendix F of the technical report \cite{aixivReport}.

Since $K \leq M$ holds, and the probability $Pr\{X\leq (1-\epsilon)\cdot N\}$ decreases with the increase of K, we can thus conduct a binary search for
$\widehat{K}$ value within $\left(\frac{p\cdot M\cdot e
	-1+p}{1-p+e\cdot p}, M\right]$, such that $\widehat{K}$ is the
smallest $K$ value such that $Pr\{X \leq (1-\epsilon) \cdot N\} \leq 1-\delta$, where $p =
\prod_{j=1}^n \frac{1}{deg(w_j)}$.

This way, we can first calculate the minimum required sample size,
$\widehat{K}$, in order to achieve the $(\epsilon, \delta)$-bound.
Then, we apply the sampling algorithm mentioned in Section
\ref{subsec:sampling} to retrieve samples. Finally, we calculate one
sample with the highest reliability and diversity. Note that, in the
case no sample dominates all other samples, we select one sample
with the highest ranking score (i.e., dominating the most number of
other samples) \cite{Yiu07b}.

\section{The Divide-and-Conquer Approach}
\label{sec:D&C}\vspace{-1ex}

\subsection{Divide-and-Conquer Algorithm}
\label{subsec:D&CAlgorithm}

We first illustrate the basic idea of the divide-and-conquer approach. As
discussed in Section \ref{subsec:sampling}, the size of all possible
task-and-worker assignments is exponential. Although RDB-SC is NP-hard, we still can
speed up the process of finding the RDB-SC answers.
By utilizing divide-and-conquer approach, the problem space is
dramatically reduced.

Figure \ref{alg:DC} illustrates the main framework for our
divide-and-conquer approach, which includes three stages: (1)
recursively divide the RDB-SC problem into two smaller subproblems,
(2) solve two subproblems, and (3) merge the answers of two
subproblems. In particular, for Stage (1), we design a partitioning
algorithm, called {\sf BG\_Partition}, to divide the RDB-SC problem
into smaller subproblems. In Stage (2), we use either the greedy or
sampling algorithm, introduced in Section \ref{sec:greedy} and
Section \ref{sec:sampling}, respectively, to get an approximation
result. Moreover, for Stage (3), we propose an algorithm, called
{\sf SA\_Merge}, to obtain RDB-SC answers by combining answers to
subproblems.

\begin{figure}[ht]
	\begin{center}\vspace{-5ex}
		\begin{tabular}{l}
			\parbox{3.1in}{
				\begin{scriptsize}
					\begin{tabbing}
						12\=12\=12\=12\=12\=12\=12\=12\=12\=12\=12\=\kill
						{\bf Procedure {\sf RDB-SC\_DC}} \{ \\
						\> {\bf Input:} $m$ time-constrained spatial tasks in $T$, $n$ workers in $W$, and a threshold $\gamma$\\
						\> {\bf Output:} two sparse and balanced tasks-workers set pairs  \\
						\> (1) \> \> if $Size(T) \leq \gamma$ \\
						\> (2) \> \> \> solve problem ($T$,$W$) to get the result $\mathbb{S}$ directly\\
						\> (3) \> \> else \\
						\> (4) \> \> \> {\sf BG\_Partition} ($T$,$W$) to ($T_1$,$W_1$) and ($T_2$,$W_2$)\\
						\> (5) \> \> \> {\sf RDB-SC\_DC} ($T_1$,$W_1$) to get answer $\mathbb{S}_1$\\
						\> (6) \> \> \> {\sf RDB-SC\_DC} ($T_2$,$W_2$) to get answer $\mathbb{S}_2$\\
						\> (7) \> \> \> {\sf SA\_Merge} ($\mathbb{S}_1$, $\mathbb{S}_2)$ to get the result $\mathbb{S}$\\
						\> (8) \> \> return $\mathbb{S}$\\
						\}
					\end{tabbing}
				\end{scriptsize}
			}
		\end{tabular}
	\end{center}\vspace{-8ex}
	\caption{\small Divide and Conquer Algorithm.}
	\label{alg:DC}\vspace{-3ex}
\end{figure}

\subsection{Partition the Bipartite Graph}
\label{subsec:DivideMethod}

As shown in Figure \ref{fig:assignment}, the task-and-worker
assignment is a bipartite graph. We first need to iteratively divide
the whole graph into two subgraphs such that few edges crossing the
cut (sparse) and close to bisection (balanced). Unfortunately, this
problem is NP-hard \cite{arora2009expander}. Here we just provide a
heuristic algorithm, namely {\sf BG\_Partition}, which is shown in
Figure \ref{alg:divide}. After running BG\_Partition we can get two
subproblems $RDB\text{-}SC_1$ and $RDB\text{-}SC_2$.

\begin{figure}[ht]
	\begin{center} \vspace{-4ex}
		\begin{tabular}{l}
			\parbox{3.1in}{
				\begin{scriptsize}
					\begin{tabbing}
						12\=12\=12\=12\=12\=12\=12\=12\=12\=12\=12\=\kill
						{\bf Procedure {\sf BG\_Partition}} \{ \\
						\> {\bf Input:} $m$ time-constrained spatial tasks in $T$ and $n$ workers in $W$\\
						\> {\bf Output:} two sparse and balanced tasks-workers set pairs  \\
						\> (1) \> \> $W_1 = \varnothing, W_2 = \varnothing$\\
						\> (2) \> \> partition tasks into two even set $T_1$ and $T_2$ with KMeans\\
						\> (3) \> \> for $w_i$ in $W$\\
						\> (4) \> \> \>if the tasks that $w_i$ can do are all included in $T_1$\\
						\> (5) \> \> \> \> put $w_i$ into $W_1$ and $W = W-\{w_i\}$\\
						\> (6) \> \> \>if the tasks that $w_i$ can do are all included in $T_2$\\
						\> (7) \> \> \> \> put $w_i$ into $W_2$ and $W = W-\{w_i\}$\\
						\> (8) \> \> add $W$ into $W_1$  \\
						\> (9) \> \> add $W$ into $W_2$\\
						\> (10) \> \> return $(T_1, W_1)$ and $(T_2,W_2)$\\
						\}
					\end{tabbing}
				\end{scriptsize}
			}
		\end{tabular}
	\end{center}\vspace{-8ex}
	\caption{\small Bipartite Graph Partitioning Algorithm.}
	\label{alg:divide}\vspace{-1ex}
\end{figure}

First, we partition tasks into two almost even subsets, $T_1$ and
$T_2$, based on their locations, which can be done through
clustering the tasks to two set (i.e. KMeans.). Then we find out the workers who can
reach tasks totally included in some subset, and add them to the
corresponding worker subset, $W_1$ or $W_2$. By doing this, these
workers are isolated in the corresponding subproblems. For the rest
workers who can do tasks both in $T_1$ and $T_2$, we add them to
both $W_1$ and $W_2$. As Figure \ref{fig:partition} shows, $w_1$ and
$w_5$ are isolated in $RDB\text{-}SC_1$ and $RDB\text{-}SC_2$
respectively while $w_2,w_3,w4$ are added in both two subproblems.

\begin{figure}[ht]\centering \vspace{-3ex}
	\subfigure[][{\scriptsize Origin Problem}]{
		\scalebox{0.23}[0.23]{\includegraphics{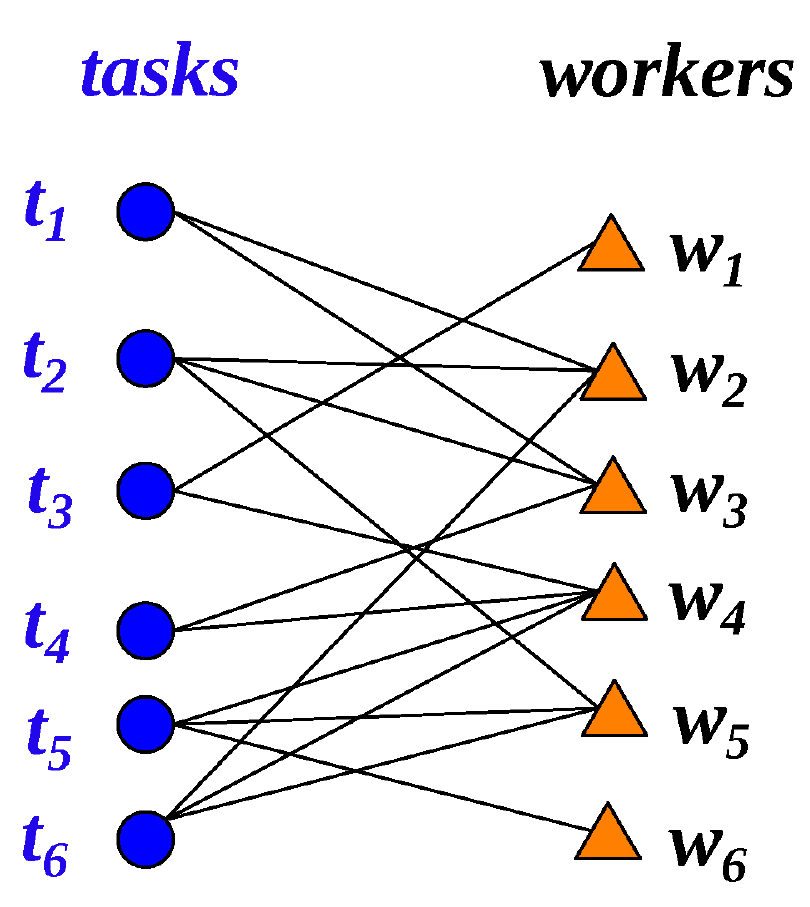}}\hspace{2ex}
		\label{subfig:BeforePartition}}
	\subfigure[][{\scriptsize Partitioned Problem}]{
		\scalebox{0.23}[0.23]{\includegraphics{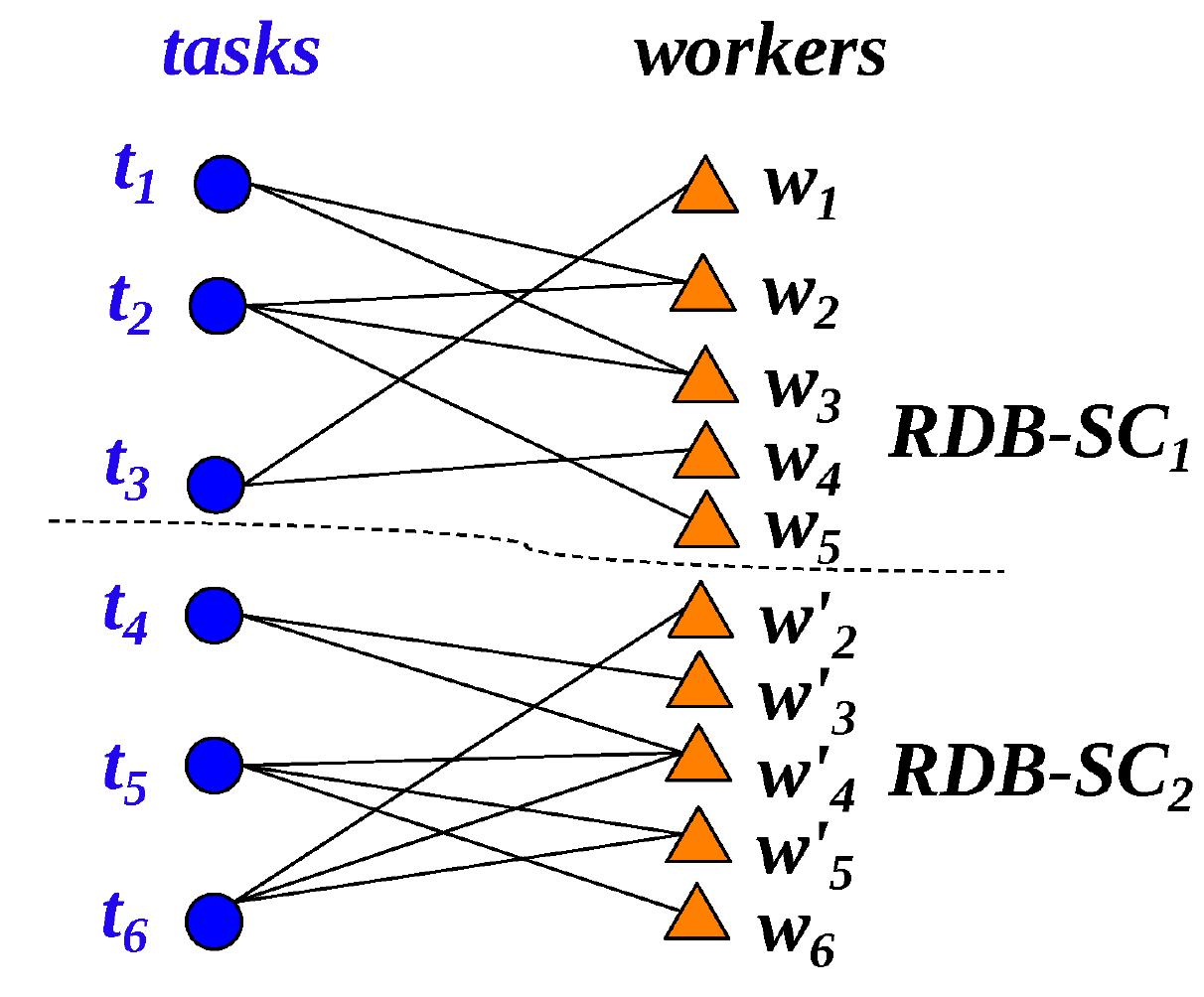}}
		\label{subfig:AfterPartition}}
	\subfigure[][{\scriptsize ICW and DCW}]{
		\scalebox{0.23}[0.23]{\includegraphics{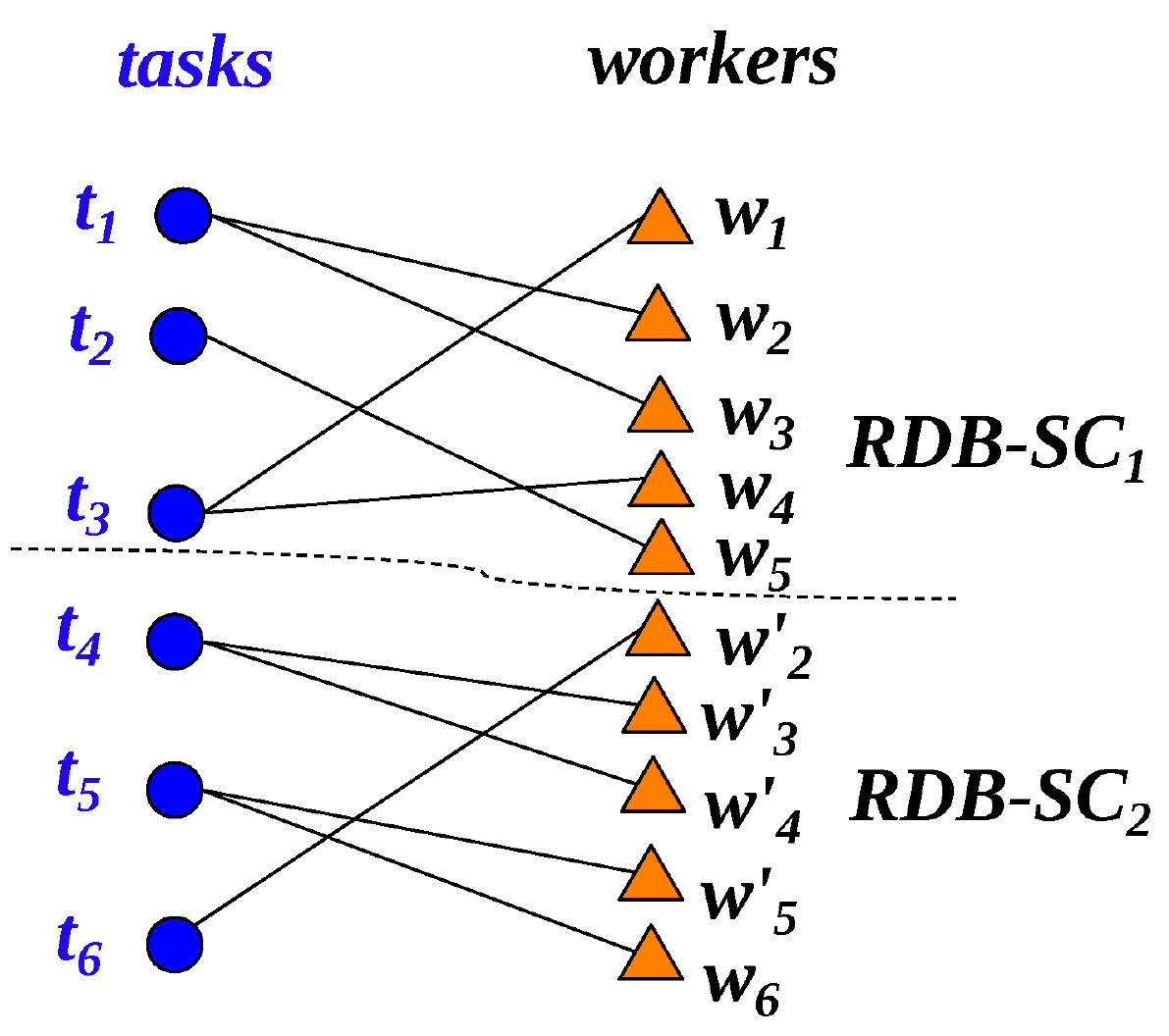}} 
		\label{subfig:subproblemAssignment}}\vspace{-3ex}
	\caption{\small Illustration of the Task-and-Worker Partitioning}
	\vspace{-2ex}
	\label{fig:partition}
\end{figure}

Note that, even if some workers are duplicated and added to both
$W_1$ and $W_2$, each of them can only be assigned to one task.
Moreover, the duplicated workers in each subproblem can only do a
part of the tasks that he can do in the whole problem. The
complexity of $RDB\text{-}SC_1$ or $RDB\text{-}SC_2$ are much lower
than $RDB\text{-}SC$. Each time before calling {\sf BG\_Partition}
algorithm, we check whether the size of the tasks is greater than a
threshold $\gamma$, otherwise that problem is small enough to solve
directly. The threshold $\gamma$ is set before running the
divide-and-conquer algorithm.

\subsection{Merge the Answers of the Subproblems}
\label{subsec:MergeMethod}
To merge the answers of the subproblems, we just need to
solve the conflicts of the workers who are added to
both two subproblems. Those duplicated workers in are called
\textit{conflicting workers}, whereas others
are called \textit{non-conflicting workers}. We first give the
property of deleting one copy of a conflicting worker.

A conflicting worker $w_i$ is called \textit{independent conflicting
	worker} (ICW), when $w_i$ is assigned to tasks $t_{i1}$ and $t_{i2}$
in the optimal assignments for $RDB\text{-}SC_1$ and
$RDB\text{-}SC_2$, respectively, and no other conflicting workers
are assigned to either $t_{i1}$ or $t_{i2}$. Otherwise, $w_i$ is
called \textit{dependent conflicting worker} (DCW). For example, in
Figure \ref{subfig:subproblemAssignment} worker $w_5$ is a ICW and worker
$w_2$ is a DCW.

\vspace{-2ex}
\begin{lemma} (Non-conflict Stable) The deletion of those conflicting workers'
	copies will not change the assignments of non-conflicting workers who are assigned
	to a same task with any deleted worker. \vspace{-1.5ex}
	\label{lemma:lem6}
\end{lemma}

\begin{proof}
	Please refer to Appendix G of the technical report \cite{aixivReport}. 
\end{proof}

\vspace{-2.5ex}

\begin{lemma} (Deletion of copies of ICWs and DCWs) The deletion of copies of ICWs can be
	done independently while DCWs' deletion need to be considered integrally. \vspace{-1ex}
	\label{lemma:lem7}
\end{lemma}
\begin{proof}
	Please refer to Appendix H of the technical report \cite{aixivReport}. 
\end{proof}
\text{}\vspace{-4ex}

\nop{
	\begin{proof}
		According to Lemma \ref{lemma:lem6}, deleting copies of conflicting
		workers will not affect the assignments of non-conflicting workers.
		The only situation we need to consider is that some conflicting
		workers are assigned to some tasks. In other words, they can connect
		with each other through some tasks. For these DCWs, we just need to
		enumerate all the possible combinations and pick the best one. If
		the size of a group of DCWs is $k$, we have $2^k$ combinations to
		check.
	\end{proof}
	
}

With Lemmas \ref{lemma:lem6} and \ref{lemma:lem7}, the algorithm for
merging subproblems is shown in Figure \ref{alg:merge}, called {\sf
	SA\_Merge}.

\begin{figure}[ht]
	\begin{center}\vspace{-3ex}
		\begin{tabular}{l}
			\parbox{3.1in}{
				\begin{scriptsize}
					\begin{tabbing}
						12\=12\=12\=12\=12\=12\=12\=12\=12\=12\=12\=\kill
						{\bf Procedure {\sf SA\_Merge}} \{ \\
						\> {\bf Input:} two subproblems ($T_1, W_1$) and ($T_2, W_2$), and their local answer $\mathbb{S}_1$ and $\mathbb{S}_2$ \\
						\> {\bf Output:} one merged problem's answer $\mathbb{S}$ for ($T_1 \cup T_2, W_1 \cup W_2$) \\
						\> (1) \> \> $W^\prime = W_1 \cap W_2$\\
						\> (2) \> \> while $W^\prime$ is not empty\\
						\> (3) \> \> \> pick first worker in $W^\prime$ as $w_t$\\
						\> (4) \> \> \> find out dependent workers for $w_t$ as $W_d$\\
						\> (5) \> \> \> add $w_t$, $W_d$ into $W_t$\\
						\> (6) \> \> \> remove one copy of each worker in $W_t$ from $\mathbb{S}_1$ and $\mathbb{S}_2$ integrally\\
						\> (7) \> \> \> $W^\prime = W^\prime - W_t$\\
						\> (8) \> \> $\mathbb{S} = \mathbb{S}_1 \cup \mathbb{S}_2$\\
						\> (9) \> \> return $\mathbb{S}$\\
						\}
					\end{tabbing}
				\end{scriptsize}
			}
		\end{tabular}
	\end{center}\vspace{-8ex}
	\caption{\small Algorithm of Merging Answers to Subproblems.}
	\label{alg:merge}\vspace{-3ex}
\end{figure}

\section{Cost-Model-Based Grid Index}
\label{sec:gridIndexCostModel}\vspace{-1ex}

\subsection{Index Structure}
\label{subsec:gridIndex}

We first illustrate the index structure, namely RDB-SC-Grid, for the
RDB-SC system. Specifically, in a 2-dimensional data space
$[0,1]^2$, we divide the space into $1/\eta^2$ square cells with
side length $\eta$, where $\eta < 1$ and we discuss in
Appendix I of the technical report \cite{aixivReport} about how to set $\eta$ based on a cost
model.

Each cell has a unique ID, $cellid$, and contains a task list and a
worker list which store tasks and workers in it, respectively. In
each task list, we maintain quadruples $(tid, l, s, e)$, where $tid$
is the task ID, $l$ is the position of the task, and $[s, e]$ is the
valid period of the task. In each worker list, we keep records in
the form $(wid, l, v, \alpha^-, \alpha^+, p)$, where $wid$ is the
worker ID, $l$ and $v$ represent the location and velocity of the
worker, respectively, $[\alpha^-, \alpha^+]$ indicates the angle
range of moving directions, and $p$ is the reliability of the
worker. For each cell, we also maintain bounds, $[v_{min},
v_{max}]$, for velocities of all workers in it, $[\alpha_{min},
\alpha_{max}]$, for all workers' moving directions, and $[s_{min},
e_{max}]$ of tasks' time constraints, where $s_{min}$ is the
earliest start time of tasks stored in the cell, and $e_{max}$ is
the latest deadline in the cell. In addition, each cell is
associated with a list, $tcell\_list$, which contains all the IDs of
cells that can be reachable to at least one worker in that cell.

\vspace{0.5ex}\noindent {\bf Pruning Strategy on the Cell Level.}
One straightforward way to construct $tcell\_list$ for cell $cell_i$
is to check all cells $cell_j$, and add those ``reachable'' cells to
$tcell\_list$. This is however quite time-consuming to check each
pair of worker and task from $cell_i$ and $cell_j$, respectively. In
order to accelerate the efficiency of building $tcell\_list$, we
propose a pruning strategy to reduce the search space. That is, for
cell $cell_i$, if a cell, $cell_j$, is in the \textit{reachable
	area} (in workers' moving directions) within two rays starting from
$cell_i$ (i.e., reachable by at least one worker in $cell_i$), then
we add $cell_j$ to list $tcell\_list$ of $cell_i$. Therefore, our
pruning strategy can effectively filter out those cells that are
definitely unreachable.

Specifically, we can prune $cell_j$ as follows. First, we calculate
the minimum and maximum distances, $d_{min}$ and $d_{max}$,
respectively, between any two points in $cell_i$ and $cell_j$. As a
result, any worker who moves from $cell_i$ will arrive at $cell_j$
with time at least $t_{min} = \frac{d_{min}}{v_{max}(cell_i)}$,
where $v_{max}(cell_i)$ is the maximum speed in $cell_i$. Thus, if
$t_{min} > e_{max}(cell_i)$, we can safely prune $cell_j$, where
$e_{max}(cell_i)$ represents the latest deadline of tasks in
$cell_i$. After pruning these unreachable cells, we further check
the rest cells one by one to build the final $tcell\_list$ for
$cell_i$.

Please refer to Appendix I of the technical report \cite{aixivReport}.

\subsection{Dynamic Maintenance}
\label{subsec:maintenance}

To insert a worker $w_i$ into RDB-SC-Grid, we first find the cell
$cell_k$ where $w_i$ locates, which uses $O(1)$ time. Moreover, we
also need to update the $tcell\_list$ for $cell_k$, which requires
$O(cost_{update})$ time in the worst case. The case of removing a
worker is similar.

To insert a task $t_j$ into RDB-SC-Grid, we obtain the cell,
$cell_k$, for the insertion, which requires $O(1)$ time cost.
Furthermore, we need to check all the cells that do not contain
$cell_k$ in their $tcell\_list$'s, which needs to check all workers
in the worst case  (i.e., $O(n)$). When removing a task from
$cell_k$, we check all the cells containing it in their
$tcell\_list$s, which also requires to check every worker in the
worst case (i.e., $O(n)$).

\section{Experimental Study}
\label{sec:exper}
\vspace{-1ex}

\subsection{Experimental Methodology}
\label{subsec:expMethods}

\begin{table}[t]
	\begin{center}\vspace{-2ex}
		\caption{Experiments setting.} \label{table2}\vspace{-2ex}
		{\small\scriptsize
			\begin{tabular}{l|l} 
				{\bf Parameter} & {\bf \qquad \qquad \qquad Values} \\ \hline \hline
				range of expiration time $rt$  & [0.25, 0.5], [0.5, 1], \textbf{[1, 2]}, [2, 3]\\ 
				reliability of workers $[p_{min}, p_{max}]$ & (0.8, 1), (0.85, 1), \textbf{(0.9, 1)}, (0.95, 1)\\
				number of tasks $m$ & 5K, 8K, \textbf{10K}, 50K, 100K \\ 
				number of workers $n$ & 5K, 8K, \textbf{10K}, 15K, 20K \\ 
				velocities of workers $[v^-, v^+]$   & [0.1, 0.2],\textbf{[0.2, 0.3]}, [0.3, 0.4], [0.4, 0.5]\\ 
				range of moving angles ($\alpha_j^+ - \alpha_j^-$) & (0, $\pi$/8], (0, $\pi$/7],
				\textbf{(0, $\pi$/6]}, (0, $\pi$/5], (0, $\pi$/4]\\
				balancing weight $\beta$   & (0, 0.2], (0.2, 0.4], \textbf{(0.4, 0.6]}, (0.6, 0.8], (0.8, 1)\\ 
				\hline
			\end{tabular}
		}
	\end{center}\vspace{-7ex}
	
\end{table}

\vspace{0.5ex}\noindent {\bf Data Sets.} We use both real and
synthetic data to test our approaches. Specifically, for real data,
we use the POI (Point of Interest) data set of China
\cite{bclChinaPOI2008} and T-Drive data set \cite{yuan2011driving,
	yuan2010t}. The POI data set of China contains over 6 million POIs
of China in 2008, whereas T-Drive data set includes GPS trajectories
of 10,357 taxis within Beijing during the period from Feb. 2 to Feb.
8, 2008. We test our approaches in the area of Beijing (with
latitude from \ang{39.6}  to \ang{40.25} and longitude from
\ang{116.1} to \ang{116.75}), which covers 74,013 POIs. After
filtering out short trajectories from T-Drive data set, we obtain
9,748 taxis' trajectories. We use POIs to initialize the locations
of tasks. From the trajectories, we extract workers' locations,
ranges of moving directions, and moving speeds. For workers'
confidences $p$, tasks' valid periods $[s,e]$, and parameter $\beta$
to balance spatial and temporal diversity, we follow the same
settings as that in synthetic data (as described below).

For synthetic data, we generate locations of workers and tasks in a
2D data space $[0, 1]^2$, following either Uniform (UNIFORM) or
Skewed (SKEWED) distribution. In particular, similar to
\cite{deng2013maximizing}, we generate tasks and workers with Skewed
distribution by letting 90\% of tasks and workers falling into a
Gaussian cluster (centered at (0.5, 0.5) with variance = $0.2^2$).
For each moving worker $w_j$, we randomly produce the angle range,
$[\alpha_j^-, \alpha_j^+]$, where $\alpha_j^-$ is uniformly chosen
within $[0, 2\pi]$ and $(\alpha_j^+ - \alpha_j^-)$ is uniformly
distributed in a range of angle (e.g., $(0, \pi/6]$). Moreover, we
also generate check-in times of each worker with Uniform or Skewed
distribution, and compute one's confidence (reliability) following
Gaussian distribution within the range $[p_{min}, p_{max}]$ (with
mean $\frac{p_{min}+p_{max}}{2}$, and variance $0.02^2$).
Furthermore, we obtain the velocity of each worker with either
Uniform or Gaussian distribution within range $[v^-, v^+]$, where
$v^-, v^+ \in (0, 1)$. Regarding spatial tasks, we generate their
valid periods, $[s,e]$, within a time interval $[st, st+rt]$, where
$st\in [0, 24]$ follows either Uniform or Gaussian distribution, and
$rt$ follows the Uniform distribution. To balance $SD$ and $TD$, we
test parameter $\beta$ following Uniform distribution within
$[0,1]$. The case of $rt$ or $\beta$ following other distributions
is similar, and thus omitted due to space limitations.

\vspace{0.5ex}\noindent {\bf Configurations of a Customized gMission
	Spatial Crowdsourcing Platform.} We tested the performance of our
proposed algorithms with an incrementally updating strategy, on a
real spatial crowdsourcing platform, namely gMission
\cite{gmissionhkust, chen2014gmission}. In particular, gMission is a
general laboratory application, which has over 20 existing active
users in the Hong Kong University of Science and Technology. In gMission, users can ask/answer spatial crowdsourcing
questions (tasks), and the platform pushes tasks to users based on
their spatial locations. As gMission has been released, the recruitment, 
monitoring and compensation of workers are already provided. A credit 
point system of gMission can record the contribution of workers. 
Workers can use their credit point to redeem coupons of book stores
or coffee shops. Moreover, when they are using gMission, their 
trajectories are recorded, which is informed to the users. 
To evaluate our RDB-SC model and algorithms, we will modify/adapt 
gMission to an RDB-SC system.

Specifically, in order to build user profiles, we set up the
peer-rating over 613 photos taken by active users. They are asked to
rate peers' photos based on their resolutions, distances, and
lights. The score of each photo is given by first removing the
highest and lowest scores, and then averaging the rest.
Moreover, the score of each user is given by the average score of
all photos taken by this user. Intuitively, a user receiving a
higher score is more reliable. Thus, we set the user's peer-rating
score as one's reliability value.

In addition, we also provide a setting option for users to configure
their preferred working area. For example, a user is going home, and
wants to do some tasks on the way home. In general, one may
just want to go to places not deviating from the direction towards his/her
home too much. Thus, a fan-shaped working area is reasonable.
Due to the maximum possible speed of the user, this fan-shaped
working area is also constrained by the maximum moving distance.

Furthermore, we enable gMission to detect the accuracy of answers. When a worker $w_j$ is taking a photo to answer a task $t_i$, we record his/her instantaneous information, like the facing direction, the location, and the timestamp, through his/her smart device. By comparing the information with the required angle and time constraint of $t_i$, we can calculate the error of angle $\Delta \theta_{ij}$ and the error of time $\Delta t_{ij}$. Then, the accuracy of the answer is $Accuracy_{ij} = \beta_i \cdot \frac{\Delta \theta_{ij}}{\pi} + (1 - \beta_i) \cdot \frac{\Delta t_{ij}}{e_i - s_i}$, where $\beta_i$ is the balancing weight of $t_i$, and  $s_i$ and $e_i$ are the starting and ending times, respectively. In particular, $0 \leq \Delta \theta_{ij} \leq \pi$ and $0 \leq \Delta t_{ij} < e_i - s_i$. Then, the accuracy of a task is given by the average value of all accuracy values of the answers to the task. For the accuracy control, it could be quite interesting and challenging, and we would like to leave it as our future work.

In order to deploy our proposed algorithms, we implement the
cost-model-based grid index, and apply the incremental updating
strategy for dynamically-changing tasks and workers to our RDB-SC
system. The framework for the incremental updating strategy is shown 
in Figure \ref{alg:incremental} below. In line 6 of the framework (Figure \ref{alg:incremental}), we 
can use our proposed algorithms to assign the available workers 
to the opening tasks, where considering $A$ and $\mathbb{S}_c$ 
means the reliability and diversity of a task $t_i$ is calculated from the received answers, the workers 
assigned to  $t_i$ in $\mathbb{S}_c$, and newly assigned workers. In particular, 
we periodically update the task-and-worker assignments every $t_{interval}$ timestamps. 
For each update, those workers, who either have accomplished the assigned 
tasks or rejected the assignment requests, would be available to receive new tasks. 
In our experiments, we set this length, $t_{interval}$, of the periodic
update interval, from 1 minute to 4 minutes, with an increment
of 1 minute. We hired 10 active users in our experiments and chose 5
sites to ask spatial crowdsourcing questions (tasks) with 15 minutes
opening time. The sites are close to each other, and in general a
user can walk from one site to another one within 2 minutes.

\nop{
	
	Incrementally updating the assignments is important for a real
	deployed algorithm. A worker's status is available when he has no
	task to work, or occupied when he is conducting some task. Once an
	occupied worker submits an answer to the assigned task or reject
	that task, he becomes available. It is a common sense that the
	assigned task for a user should not changing when he is approaching
	to the target place to conduct it. We periodically use our proposed
	algorithms to assign the available workers to the opening tasks.
	During that process, one task's diversity and reliability is
	calculated on already received answers, on-going workers and newly
	assigned workers. In a word, we incrementally update the assignments
	periodically with a constant time interval, $t_{interval}$.
	
}

\nop{
	
	As gMission is still a laboratory application, it just has a small
	group of active users. If gMission became popular in future, there
	will be a large number of available workers each time when we
	incrementally update the assignments.  In this set of experiments,
	we want to compare our approaches in a more realistic situation,
	then we set $t_{interval}$ a little bit longer to have more
	available workers in each incremental update. In the experiments,
	$t_{interval}$ is chosen from 1 minute to 4 minutes with an
	increment of 1 minute.
	
}

\begin{figure}[ht]
	\begin{center}\vspace{-4ex}
		\begin{tabular}{l}
			\parbox{3.1in}{
				\begin{scriptsize}
					\begin{tabbing}
						12\=12\=12\=12\=12\=12\=12\=12\=12\=12\=12\=\kill
						{\bf Procedure {\sf RDB-SC\_Incremental}} \{ \\
						\> {\bf Input:} $m$ time-constrained spatial tasks in $T$, 	$n$ workers in $W$, \\
						\> {\bf Output:} a updated task-and-worker assignment strategy, $\mathbb{S}$, \\
						\>\>\>\>\hspace{2ex}with high reliability and diversity\\
						\> (1) \> \> $\mathbb{S} = \emptyset$\\
						\> (2) \> \> from $W$, retrieve all the available workers to $W_a$\\
						\> (3) \> \> from $T$, retrieve all the opening tasks to $T_a$\\
						\> (4) \> \> obtain the received answers of all the tasks in $W_a$, noted as $A$\\
						\> (5) \> \> obtain the current assignment, noted as $S_c$\\
						\> (6) \> \> assign workers in $W_a$ to tasks in $T_a$ considering $A$ and $S_c$
						(new pairs are \\
						\>\>\>added to $\mathbb{S}$ )\\
						\> (7) \> \> $\mathbb{S} = \mathbb{S} \cup \mathbb{S}_c$\\
						\> (8) \> \> return $\mathbb{S}$\\
						\}
					\end{tabbing}
				\end{scriptsize}
			}
		\end{tabular}
	\end{center}\vspace{-8ex}
	\caption{\small Incremental Updating Strategy.}
	\label{alg:incremental}\vspace{-2.5ex}
\end{figure}

\begin{figure}[t]\centering \vspace{-1ex}
	\subfigure[][{\scriptsize Minimum Reliability}]{
		\scalebox{0.18}[0.18]{\includegraphics{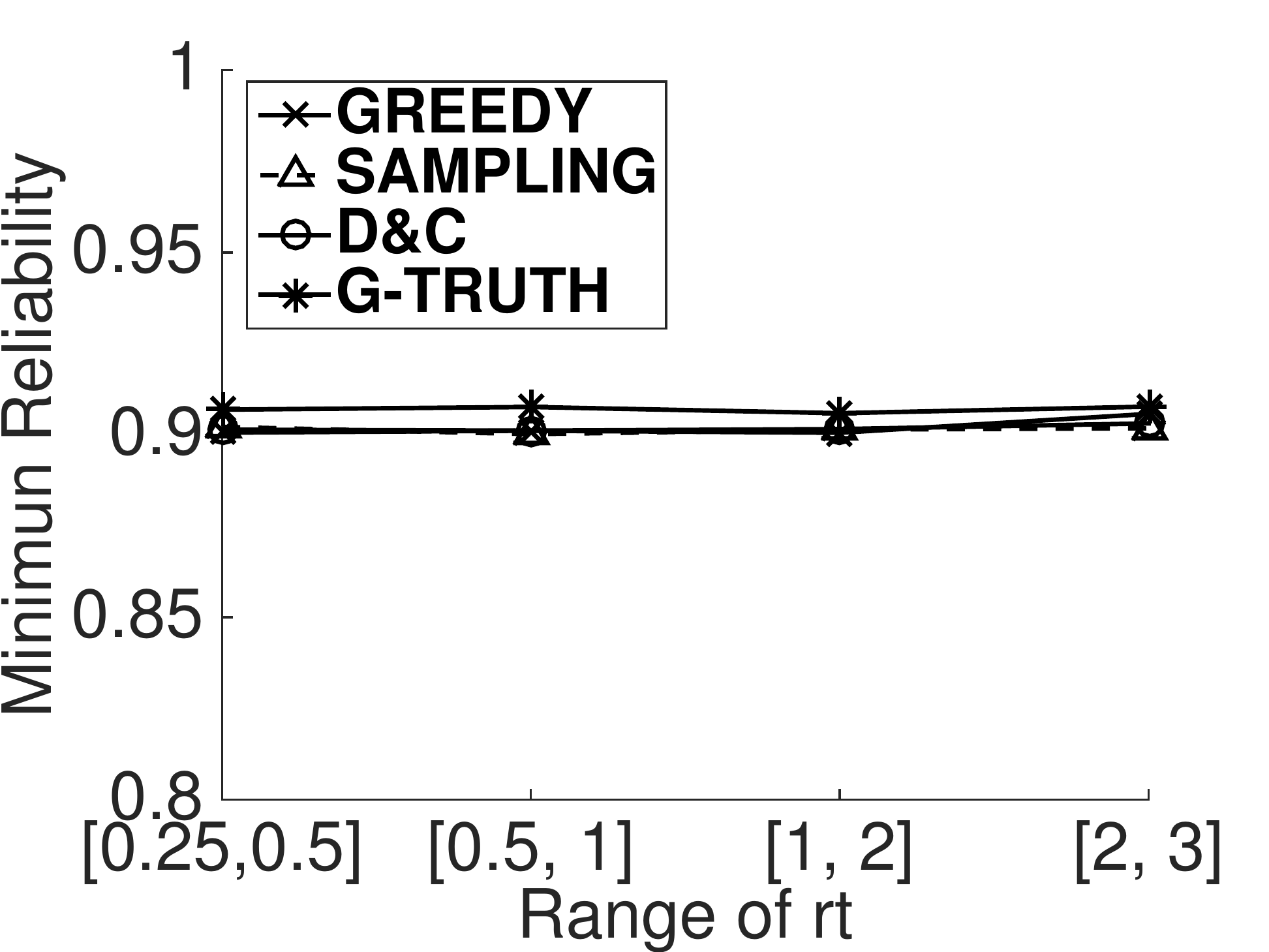}}\vspace{-2ex}
		\label{subfig:tconstrainR}}
	\subfigure[][{\scriptsize Summation of Diversity}]{
		\scalebox{0.18}[0.18]{\includegraphics{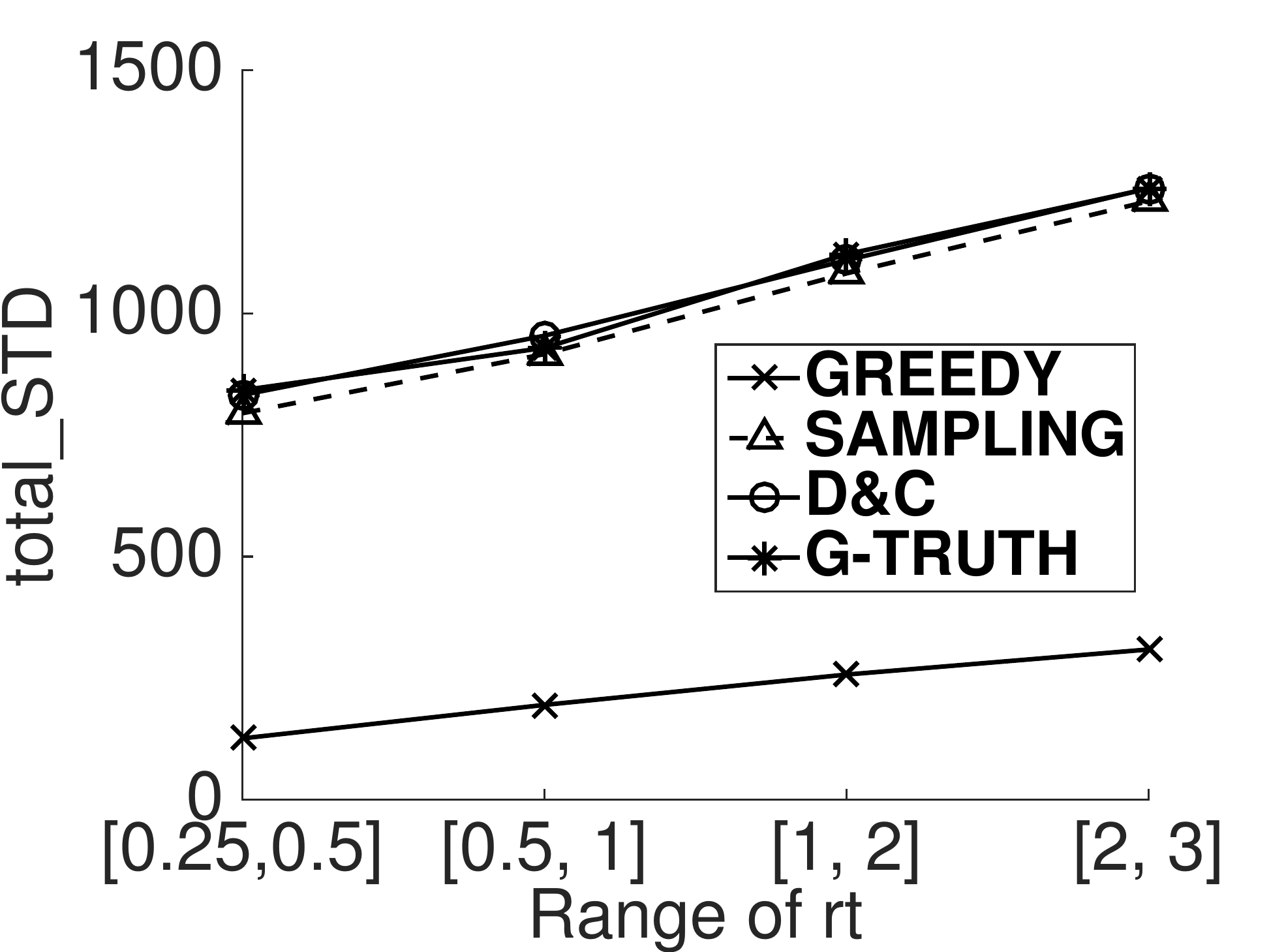}}\vspace{-2ex}
		\label{subfig:tconstrainD}}
	\vspace{-3ex}
	\caption{\small Effect of Tasks' Expiration Time Range of $rt$} \vspace{-3ex}
	\label{fig:tconstrain}
\end{figure}

\begin{figure}[t]\centering 
	\subfigure[][{\scriptsize Minimum Reliability}]{
		\scalebox{0.18}[0.18]{\includegraphics{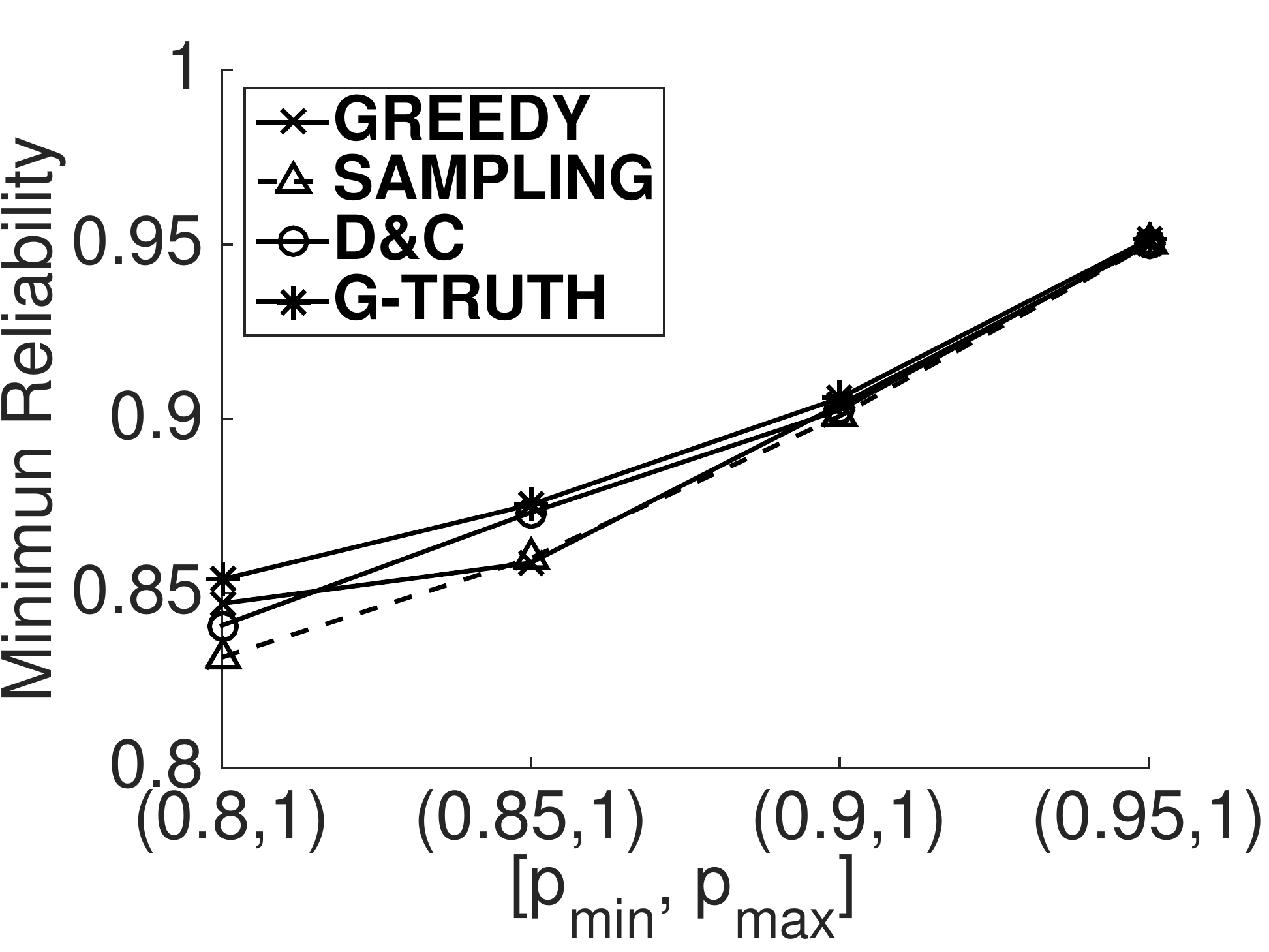}} 
		\label{subfig:ConfidenceReliability}}
	\subfigure[][{\scriptsize Summation of Diversity}]{
		\scalebox{0.18}[0.18]{\includegraphics{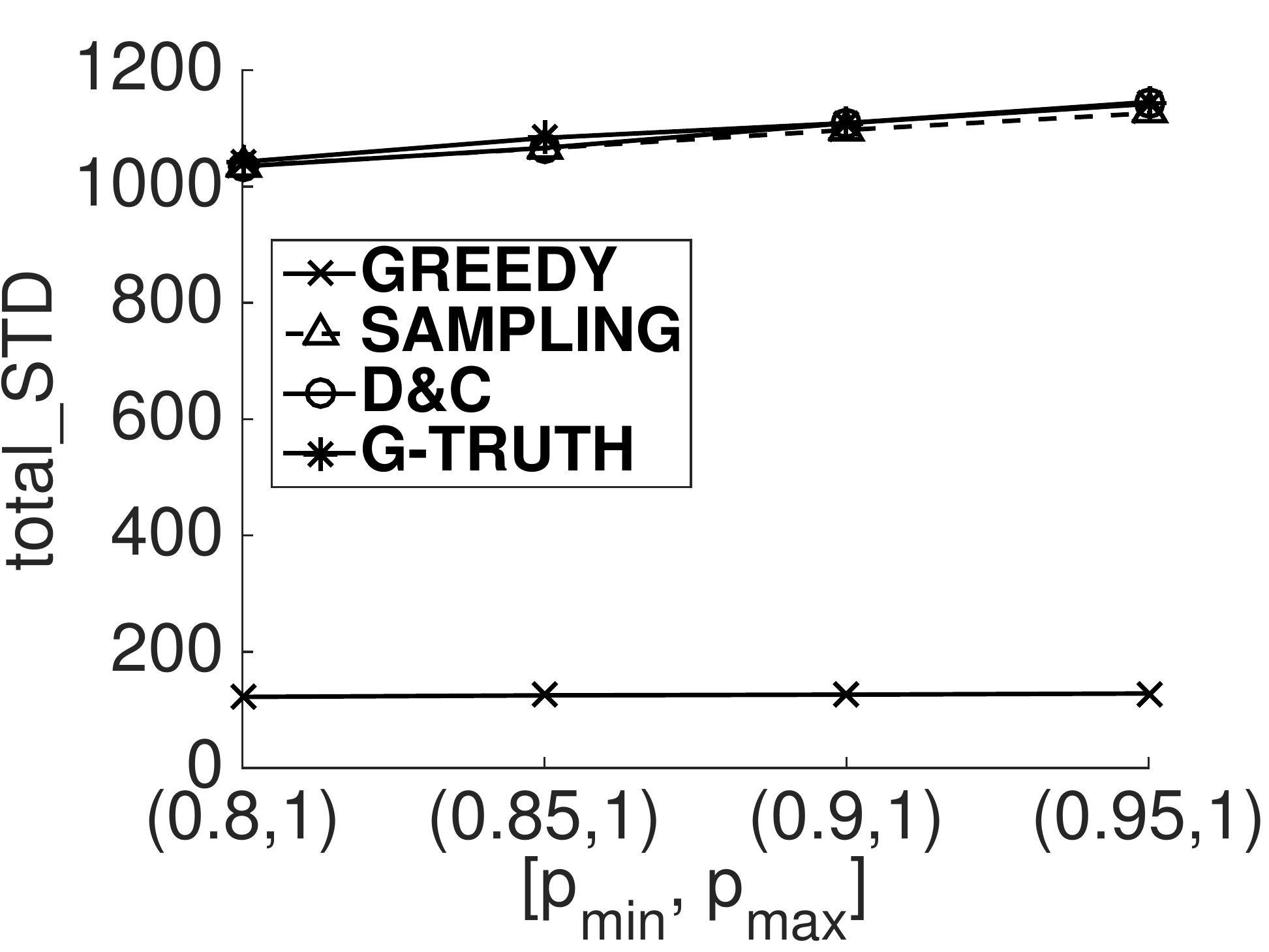}} 
		\label{subfig:ConfidenceDiversity}}
	\vspace{-3ex}
	\caption{\small Effect of Workers' Reliability $[p_{min}, p_{max}]$}
	\vspace{-5ex}
	\label{fig:ConfidenceEffect}
\end{figure}

\vspace{0.5ex}\noindent {\bf RDB-SC Approaches and Measures.}  Greedy (GREEDY) assigns
each worker to a ``best'' task according to the current situation when processing
the worker, which is just a local optimal approach. Sampling (SAMPLING) randomly assigns all the
available workers several times and picks the best test result, using our equations in
Section \ref{subsec:sample_size} to calculate the sampling times to bound the accuracy.
Divide-and-Conquer (D\&C) divides the original problem into subproblems, solves each
one and merges their results. To accelerate D\&C,  we use SAMPLING to solve subproblems
of D\&C, which will sacrifice a little accuracy. Nonetheless, this trade-off is effective, we will
show SAMPLING has a good performance when the problem space is small in our experiments.
To evaluate our 3 proposed
approaches, we will compare them with the ground truth. However,
since RDB-SC problem is NP-hard (as discussed in Section
\ref{subsec:reduction2}), it is infeasible to calculate the real
optimal result as the ground truth. Thus, we use Divide-and-Conquer
approach with the embedded sampling approach (discussed in Section
\ref{sec:sampling}) to calculate sub-optimal result by setting the
sampling size 10 times larger than D\&C (denoted as G-TRUTH).

Table \ref{table2} depicts the experimental settings, where the
default values of parameters are in bold font. In each set of
experiments, we vary one parameter, while setting others to their
default values. We report $\min_{i=1}^m rel(t_i, W_i)$, the minimum
reliability, and $total\_STD$, the summation of the expected
spatial/temporal diversities. All our experiments were run on an
Intel Xeon X5675 CPU @3.07 GHZ with 32 GB RAM.

\subsection{Experiments on Real Data}
\label{subsec:expReal}

In this subsection, we show the effects of workers' confidence
$p$, tasks' valid periods $[s,e]$ and balancing parameter $\beta$ on the real data.
We use the locations of the POIs as the locations of tasks.
To initialize a worker based on trajectory records of a taxi, we use the start
point of the trajectory as the worker's location, use the average speed of the taxi as the
worker's speed. For the moving angle's range of the worker, we draw a sector at the
start point and contain all the other points of the trajectory in
the sector, then we use the sector as the moving angle's range of the worker.
We uniformly sample 10,000 POIs from the 74,013 POIs in
the area of Beijing and the sampled POI date set follows the original data set's distribution.
In other words, we have 10,000 tasks and 9,748 workers in the experiments on real data.

\vspace{0.5ex}\noindent {\bf Effect of the Range of Tasks'
	Expiration Times $rt$.} Figure \ref{fig:tconstrain} shows the effect
of varying the range of tasks' expiration times $rt$. When this
range increases from $[0.25, 0.5]$ to $[2,3]$, the minimum
reliability is very stable, and the diversities $total\_STD$ of all
the approaches gradually increase. Intuitively, longer expiration
time for a task $t_i$ means more workers can arrive at location
$l_i$ to accomplish $t_i$. From the perspective of workers, each
worker can have more choices in his/her reachable area. Thus, each
worker can choose a better target task with higher diversity.
Similar to previous results, SAMPLING and D\&C approaches can
achieve higher diversities than GREEDY, and slightly lower
diversities compared with G-TRUTH. The requester can use this parameter 
to constrain the range of the opening time of a task. For example, if one wants to know 
the situation of a car park in a morning, he/she can set the time range as 
the period of the morning.

\vspace{0.5ex}\noindent {\bf Effect of the Range of Workers'
	Reliabilities $[p_{min}, p_{max}]$.} Figure \ref{fig:ConfidenceEffect}
reports the effect of the range, $[p_{min}, p_{max}]$, of workers'
reliabilities on the reliability/diversity of our proposed RDB-SC
approaches. For the minimum reliability, the reliabilities of workers
may greatly affect the reliability of spatial tasks (as given by
Eq.~(\ref{eq:eq1})). Thus, as shown in Figure
\ref{subfig:ConfidenceReliability}, for the range with higher
reliabilities, the minimum reliability of tasks also becomes larger.
For diversity $total\_STD$, according to Lemma \ref{lemma:lem1},
when the workers assigned to tasks $t_i$ have higher reliabilities,
the expected spatial/temporal diversity will be higher. Therefore,
we can see the slight increases of $total\_STD$ in Figure
\ref{subfig:ConfidenceDiversity}. Similar to previous results,
SAMPLING and D\&C show reliability and diversity similar to G-TRUTH,
and have higher diversities than GREEDY.

We test the effect of the requester-specified weight range. Due to 
space limitations, please refer to the experimental
results with different $\beta$ values in Appendix J of the technical report \cite{aixivReport}.
Requesters can use this parameter to reflect their preference. 
The valid value of $\beta$ is from 0 to 1. The 
bigger $\beta$ is, the more spatial diverse the answers are. The 
smaller $\beta$ is, the more temporal diverse the answers are.  
If one has no preference, he/she can simply set $\beta$ to 0.5.

\subsection{Experiments on Synthetic Data}
\label{subsec:expSynthetic}

\begin{figure}[t]\centering \vspace{-1ex}
	\centering
	\subfigure[][{\scriptsize Minimum Reliability}]{
		\scalebox{0.18}[0.18]{\includegraphics{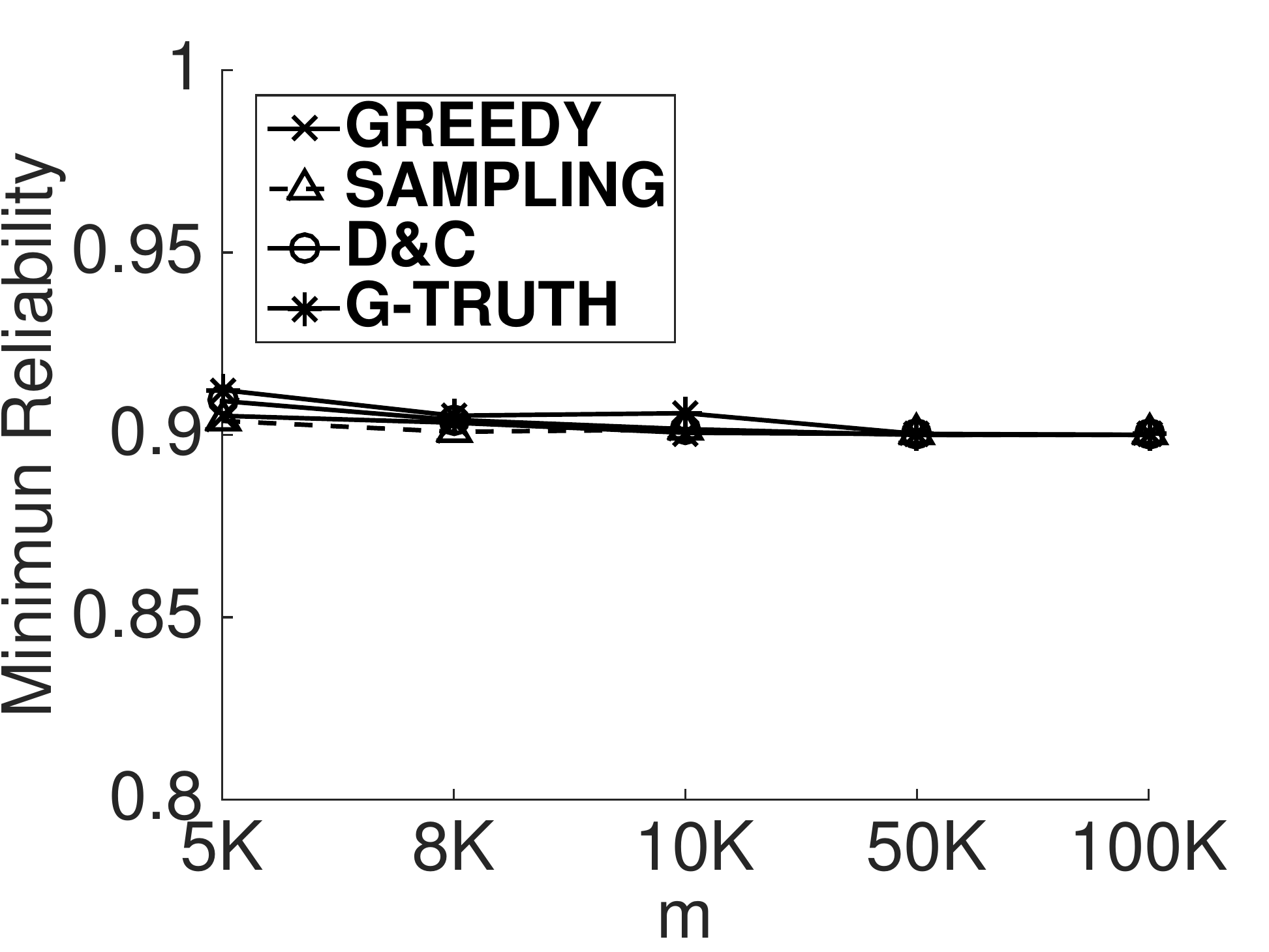}}\vspace{-2ex}
		\label{subfig:taskNReliability}}
	\subfigure[][{\scriptsize Summation of Diversity}]{
		\scalebox{0.18}[0.18]{\includegraphics{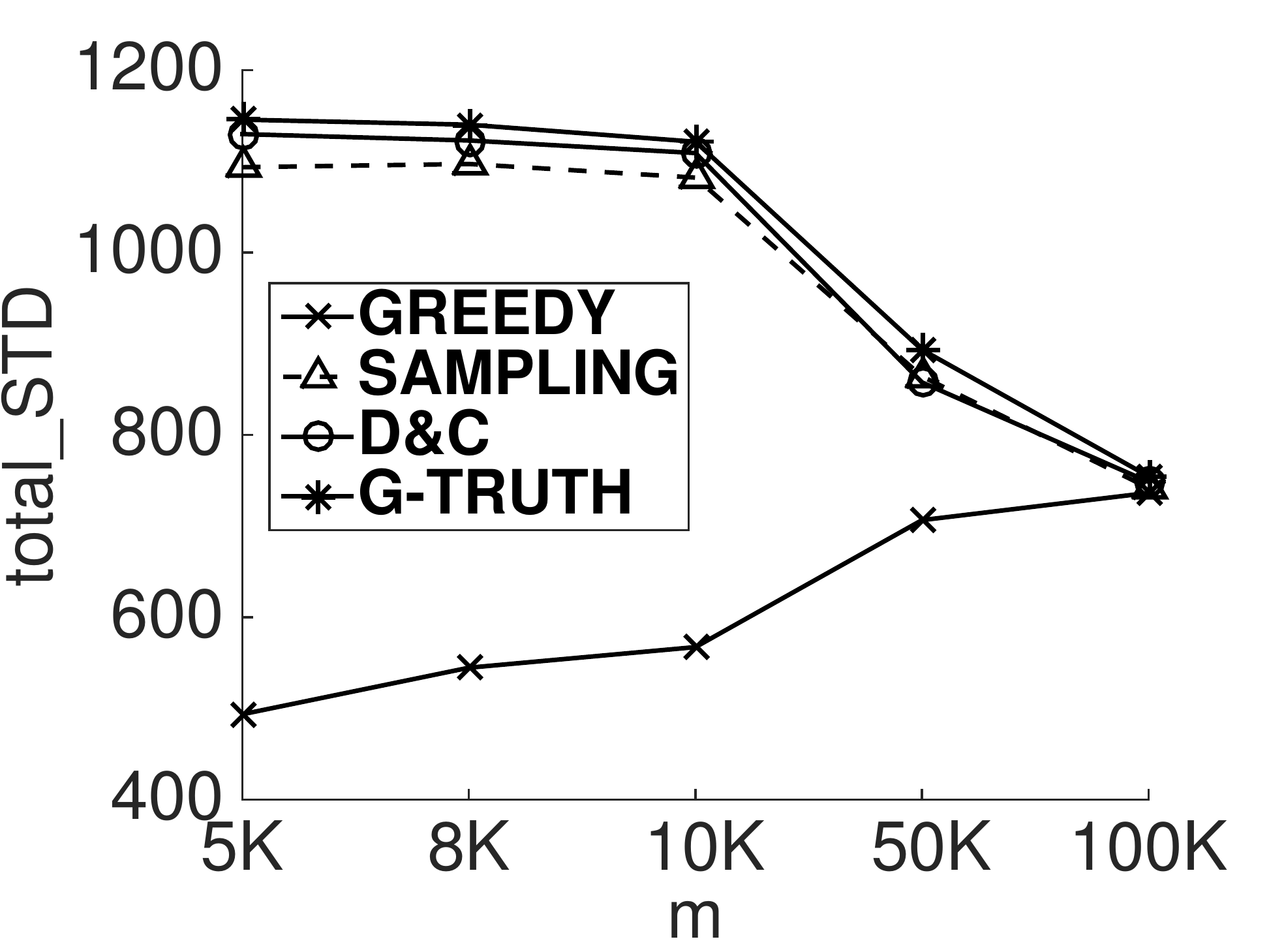}}\vspace{-2ex}
		\label{subfig:taskNDiversity}}
	\vspace{-3ex}
	\caption{\small Effect of the Number of Tasks $m$ (UNIFORM)} \vspace{-3ex}
	\label{fig:taskN}
\end{figure}

\begin{figure}[t]\centering 
	\subfigure[][{\scriptsize Minimum Reliability}]{
		\scalebox{0.18}[0.18]{\includegraphics{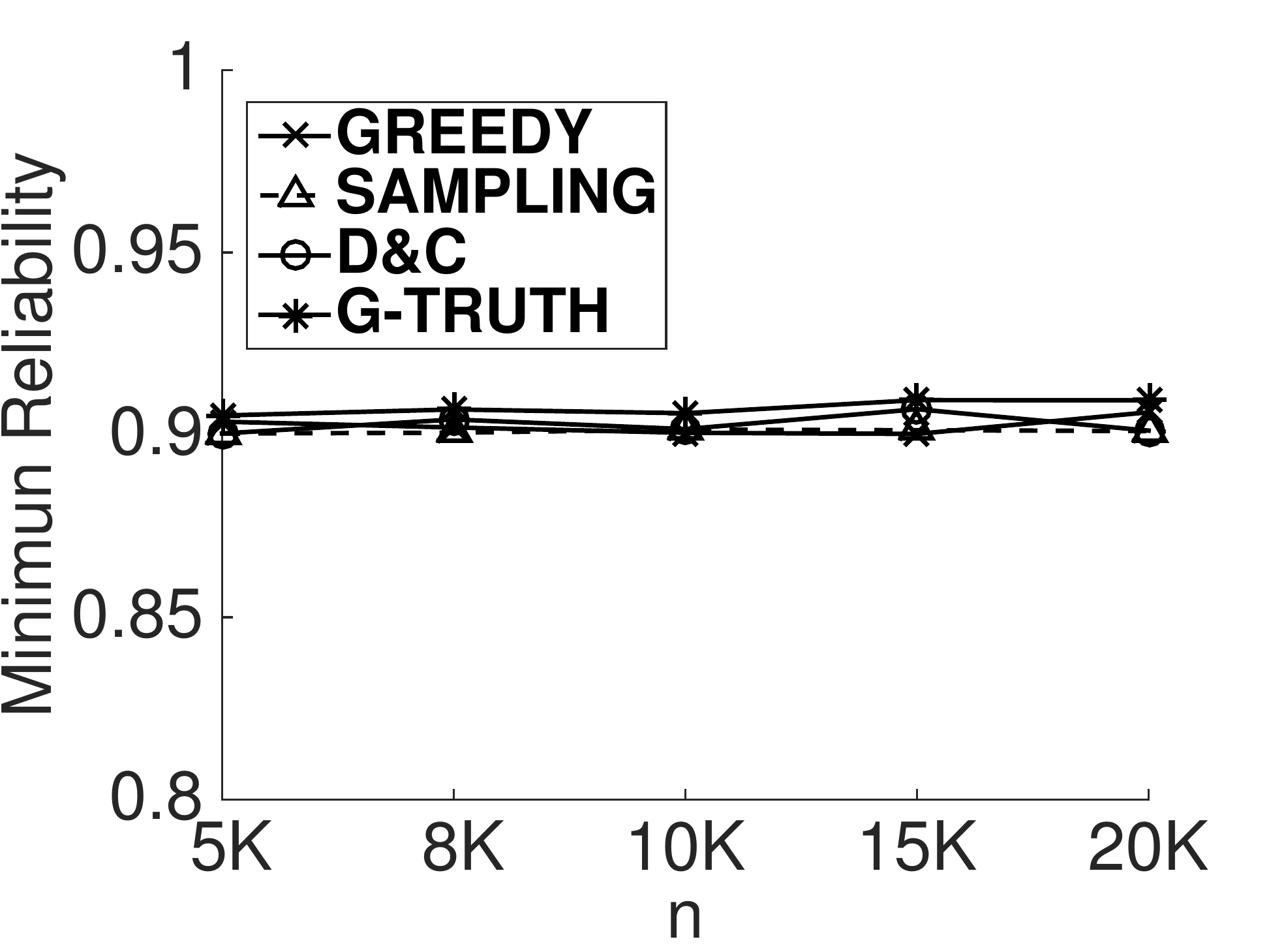}}\vspace{-3ex}
		\label{subfig:workerNReliability}}
	\subfigure[][{\scriptsize Summation of Diversity}]{
		\scalebox{0.18}[0.18]{\includegraphics{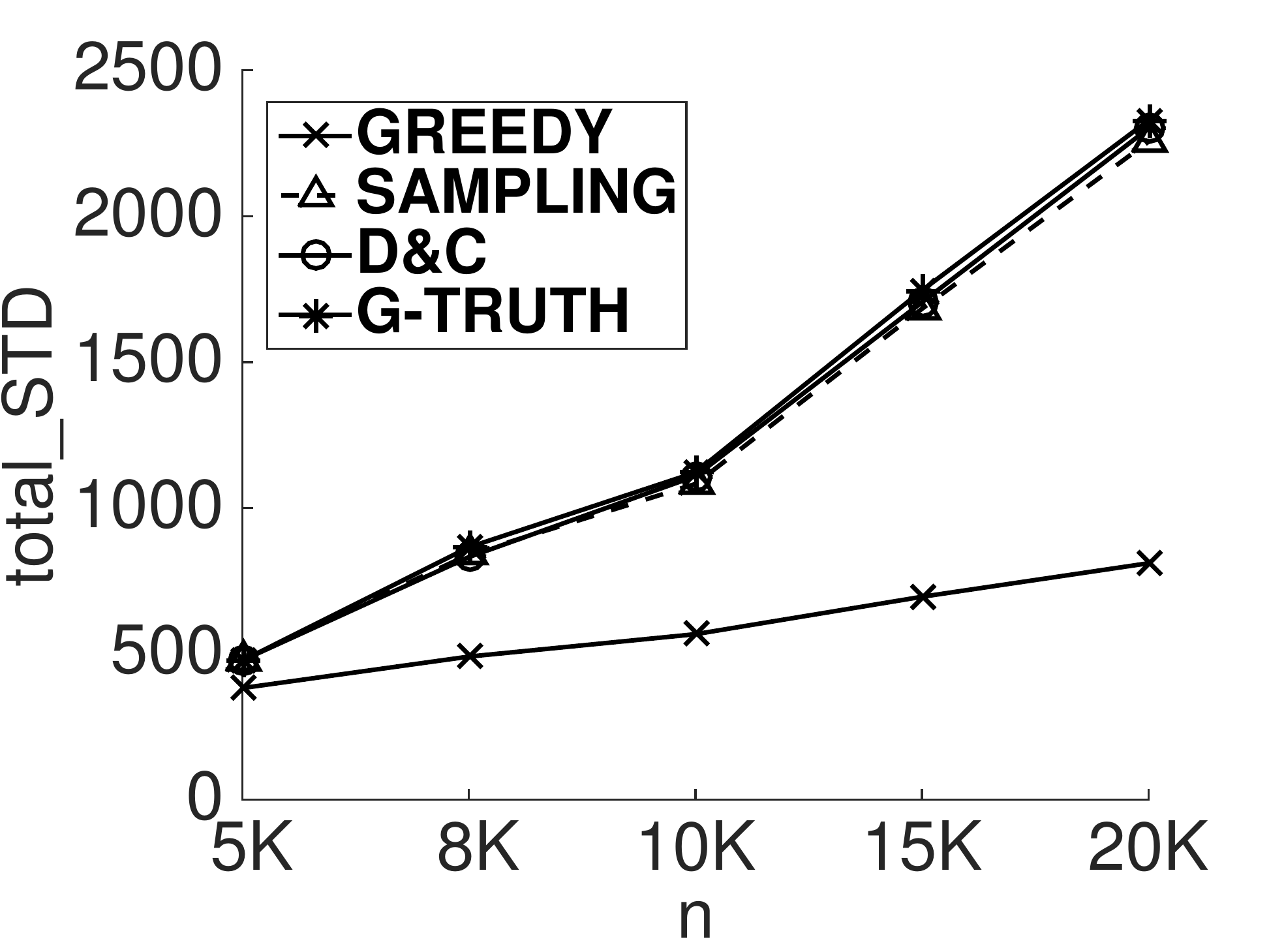}}\vspace{-3ex}
		\label{subfig:workerNDiversity}}
	\vspace{-3ex}
	\caption{\small Effect of the Number of Workers $n$ (UNIFORM)}
	\vspace{-3ex}
	\label{fig:workerN}
\end{figure}

In this subsection, we test the effectiveness and robustness of our
proposed 3 RDB-SC approaches, GREEDY, SAMPLING, and D\&C, compared
with G-TRUTH, by varying different parameters.

As we already see the effects of $p$ and $[s,e]$, we will focus on the rest
four parameters in Table \ref{table2} in this subsection. We first report the
experimental results on Uniform task/worker distributions. Please
refer to more (similar) experimental results over data sets with
Uniform/Skew distributions in Appendix J of the technical report \cite{aixivReport}.

\vspace{0.5ex}\noindent {\bf Effect of the Number of Tasks $m$.}
Figure \ref{fig:taskN} shows the effect of the number, $m$, of
spatial tasks on the reliability and diversity of RDB-SC answers,
where we vary $m$ from 5K to 100K. In Figure \ref{subfig:taskNReliability},
all the 3 approximation approaches can achieve good minimum
reliability, which are close to G-TRUTH, and remain high (i.e., with
reliability around 0.9). The reliability of D\&C is higher than that
of the other two approaches. When the number, $m$, of tasks
increases, the minimum reliability slightly decreases. This is
because given a fixed (default) number of workers, our assignment
approaches trade a bit the reliability for more accomplished tasks.

For the diversity, our 3 approaches have different trends for larger
$m$. In Figure \ref{subfig:taskNDiversity}, for large $m$,
the total diversity, $total\_STD$, of GREEDY becomes larger, while
that of the other two approaches decreases. For GREEDY, more tasks
means more possible task targets for each worker on average. This
can make a particular worker to choose one possible task such that
high diversity is obtained. In contrast, for SAMPLING, when the
number of tasks increases, the size of possible combinations
increases dramatically. Thus, under the same accuracy setting, the
result will be relatively worse (as discussed in Section
\ref{subsec:sample_size}). In D\&C, we divide the original problem
into several subproblems of smaller scale. Since the number of
possible combinations in subproblems decreases quickly, we can
achieve good solutions to subproblems. After merging answers to
subproblems, D\&C can obtain a slightly higher $total\_STD$ than
SAMPLING (about 3\% improvement). We can see that, both SAMPLING and
D\&C have  $total\_STD$ very close to G-TRUTH, which
indicates the effectiveness of SAMPLING and D\&C.

In Figure \ref{subfig:taskNDiversity}, when $m$ is small, SAMPLING
and D\&C can achieve much higher $total\_STD$ than that of GREEDY.
The reason is that GREEDY has a bad start-up performance. That is,
when most reachable tasks of a worker are not assigned with workers
(namely, empty tasks), he/she is prone to join those tasks that
already have workers. In particular, when a worker joins an empty
task, he/she can only improve the temporal diversity (TD) of that
task, and has no contribution to the spatial diversity (SD),
according to the definitions in Section \ref{subsec:RDB-SC}. On the
other hand, if a worker joins a task that has already been assigned
with some workers, then his/her join can improve both SD and TD,
which leads to higher STD. Since GREEDY always chooses
task-and-worker pairs that increase the diversity most, GREEDY will
always exploit those non-empty tasks, which may potentially miss the
good assignment with high diversity $total\_STD$. Thus, $total\_STD$
of GREEDY is low when $m$ is 5K - 10K.

\begin{figure}[t]\centering \vspace{-1ex}
	\subfigure[][{\scriptsize Minimum Reliability}]{
		\scalebox{0.18}[0.18]{\includegraphics{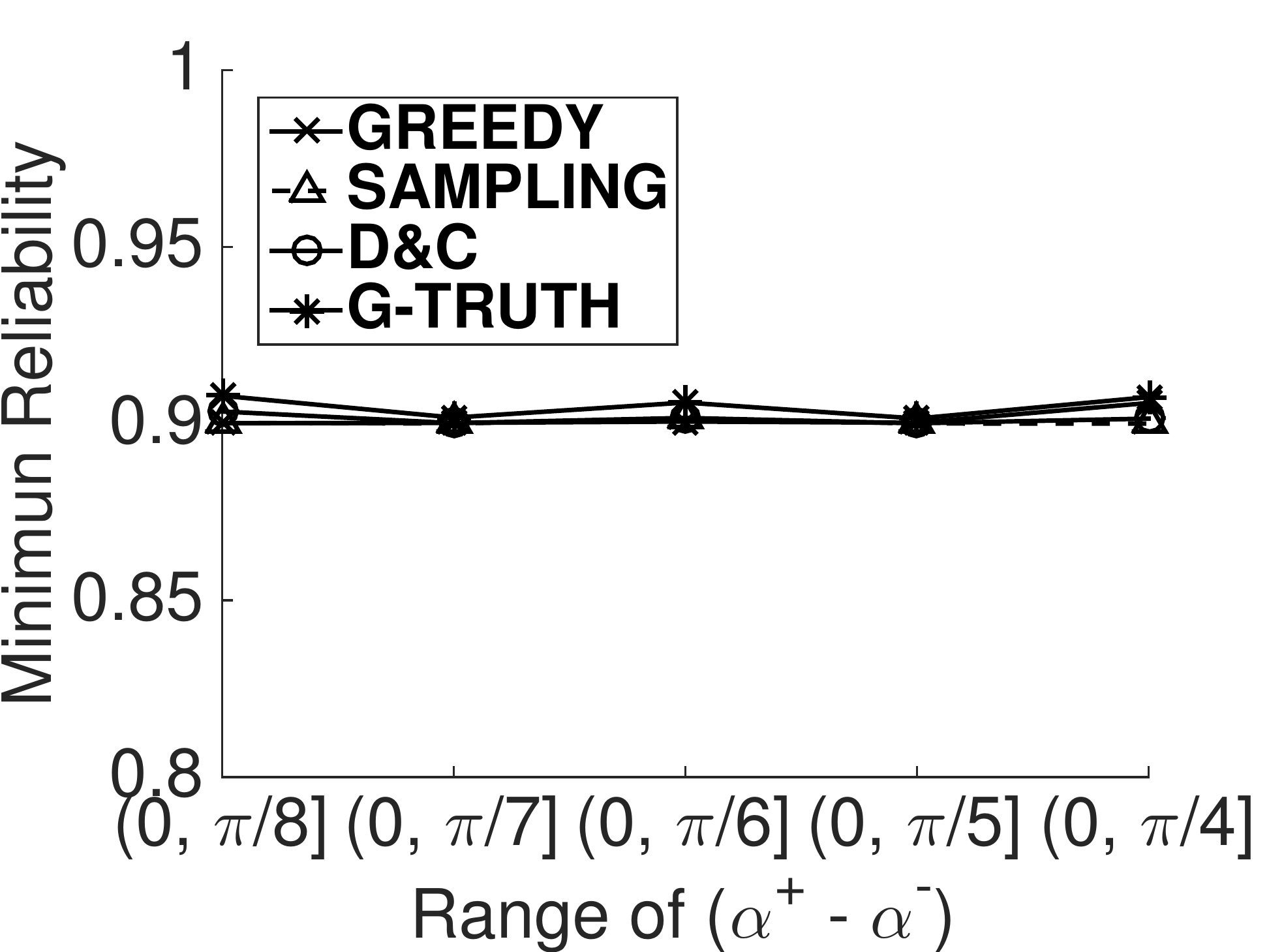}}\vspace{-2ex}
		\label{subfig:AngleReliability}}
	\subfigure[][{\scriptsize Summation of Diversity}]{
		\scalebox{0.18}[0.18]{\includegraphics{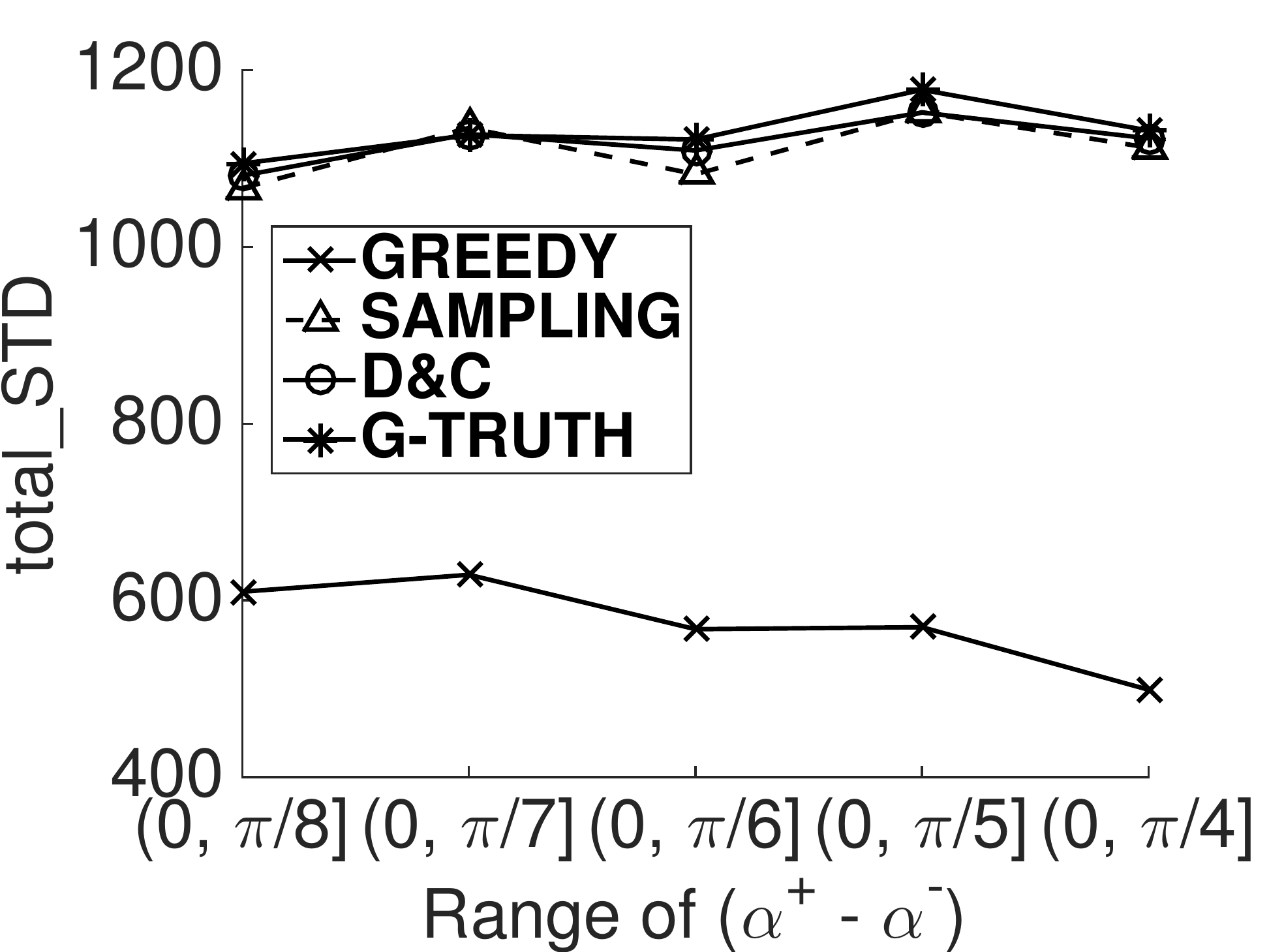}}\vspace{-2ex}
		\label{subfig:AngleDiversity}}
	\vspace{-3ex}
	\caption{\small Effect of the Range of Angles $(\alpha_j^+ - \alpha_j^-)$(UNIFORM)} \vspace{-2ex}
	\label{fig:angle}
\end{figure}

\vspace{0.5ex}\noindent {\bf Effect of the Number of Workers $n$.}
Figure \ref{fig:workerN} illustrates the experimental results on
different numbers, $n$, of workers from $5K$ to $20K$. In Figure
\ref{subfig:workerNReliability}, the minimum reliability is not very
sensitive to $n$. This is because, although we have more workers,
there always exist tasks that are assigned with just one worker.
According to Eq.~(\ref{eq:eq1}), the minimum reliability among tasks
is very close to the lower bound of workers' confidences. Thus, the
reliability slightly changes with respect to $n$.

On the other hand, the diversities, $total\_STD$, of all the four
approaches increase for larger $n$ value. In particular, as depicted
in Figure \ref{subfig:workerNDiversity}, the diversity of SAMPLING
increases more rapidly than that of GREEDY. Recall from Lemma
\ref{lemma:lem4} that, more workers means a higher diversity for
each task. When the number of workers increases, the average number
of workers of each task also increases, which leads to a higher
$total\_STD$. Similar to previous results, SAMPLING and D\&C have
diversities very close to G-TRUTH, which confirms the effectiveness
of our approaches.

\vspace{0.5ex}\noindent {\bf Effect of the Range of Moving Angles
	$(\alpha_j^+ - \alpha_j^-)$.} Figure \ref{fig:angle} varies the
range, $(\alpha_j^+ - \alpha_j^-)$, of moving angles for workers
$w_j$ from $(0, \pi/8]$ to $(0, \pi/4]$. From figures, we can see
that the minimum reliability is not very sensitive to this angle
range. With different angle ranges, the reliability of our proposed
approaches remains high (i.e., above 0.9). Moreover, both SAMPLING
and D\&C approaches can achieve much higher diversities than GREEDY,
and they have diversities similar to G-TRUTH, which indicates
good effectiveness against different angle ranges of moving
directions. On the other hand, $total\_STD$ of GREEDY drops when
angle becomes larger. The reason is similar to the cause of GREEDY's
bad start-up, which is discussed when we show the effect of the number of tasks.
Larger angle range means more reachable tasks, then workers are more
likely to find a task that has been assigned with workers and join 
that task, which leads to low diversity. On a real platform, the workers 
may set this parameter based on their personal interests. For example, 
if a worker would like to deviate more from his/her moving direction, 
he/she can set the range of his/her moving angle wider.

We test the effect of the range of workers' velocities. Due to 
space limitations, please refer to the experimental
results with different ranges, $[v^-, v^+]$, in Appendix J of the technical report \cite{aixivReport}.

\begin{figure}[t]\centering 
	\subfigure[][{\scriptsize CPU Time (vs. $m$)}]{
		\scalebox{0.18}[0.18]{\includegraphics{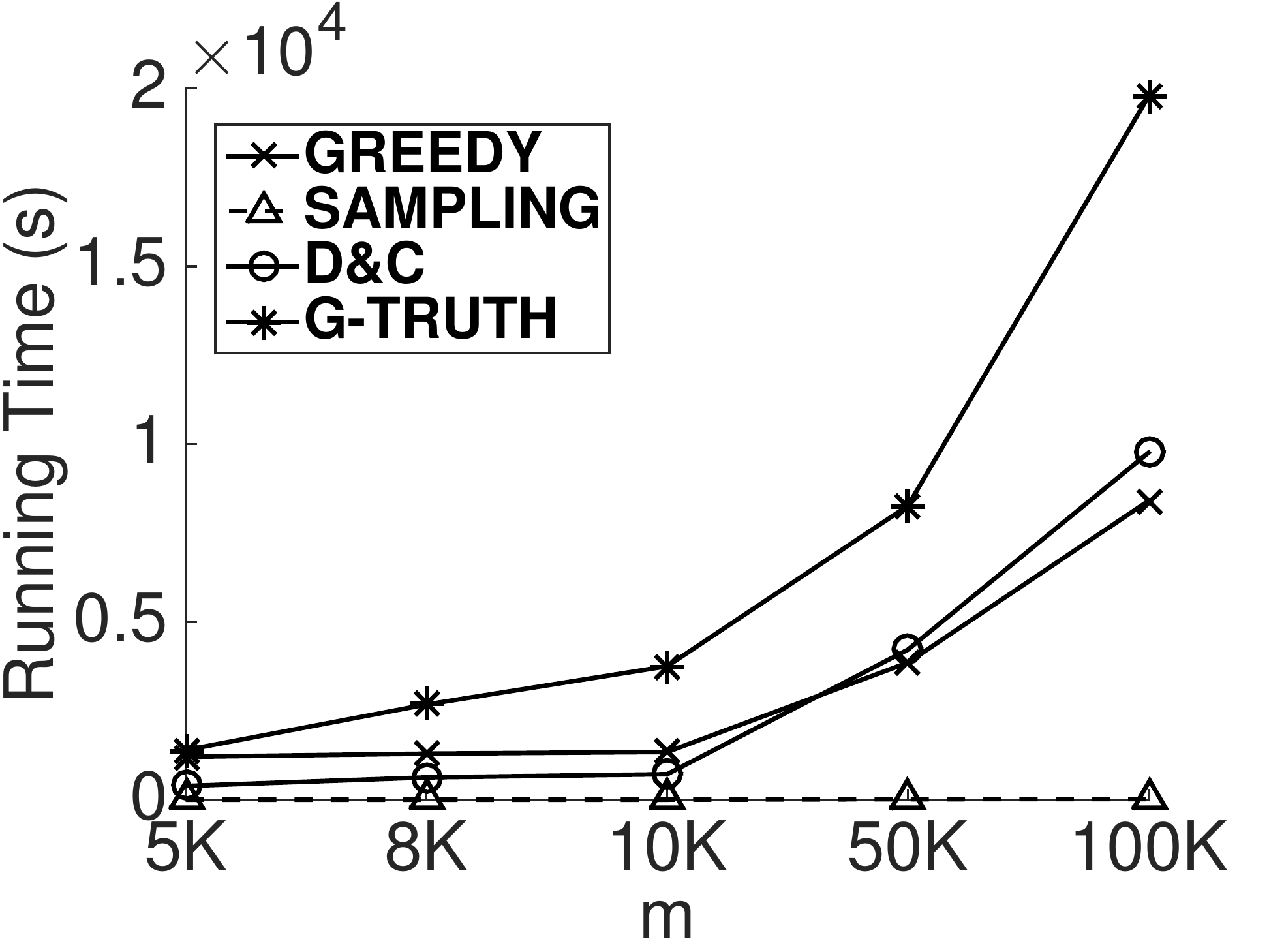}}\vspace{-3ex}
		\label{subfig:runningTimeTasks}}
	\subfigure[][{\scriptsize CPU Time (vs. $n$)}]{ 
		\scalebox{0.18}[0.18]{\includegraphics{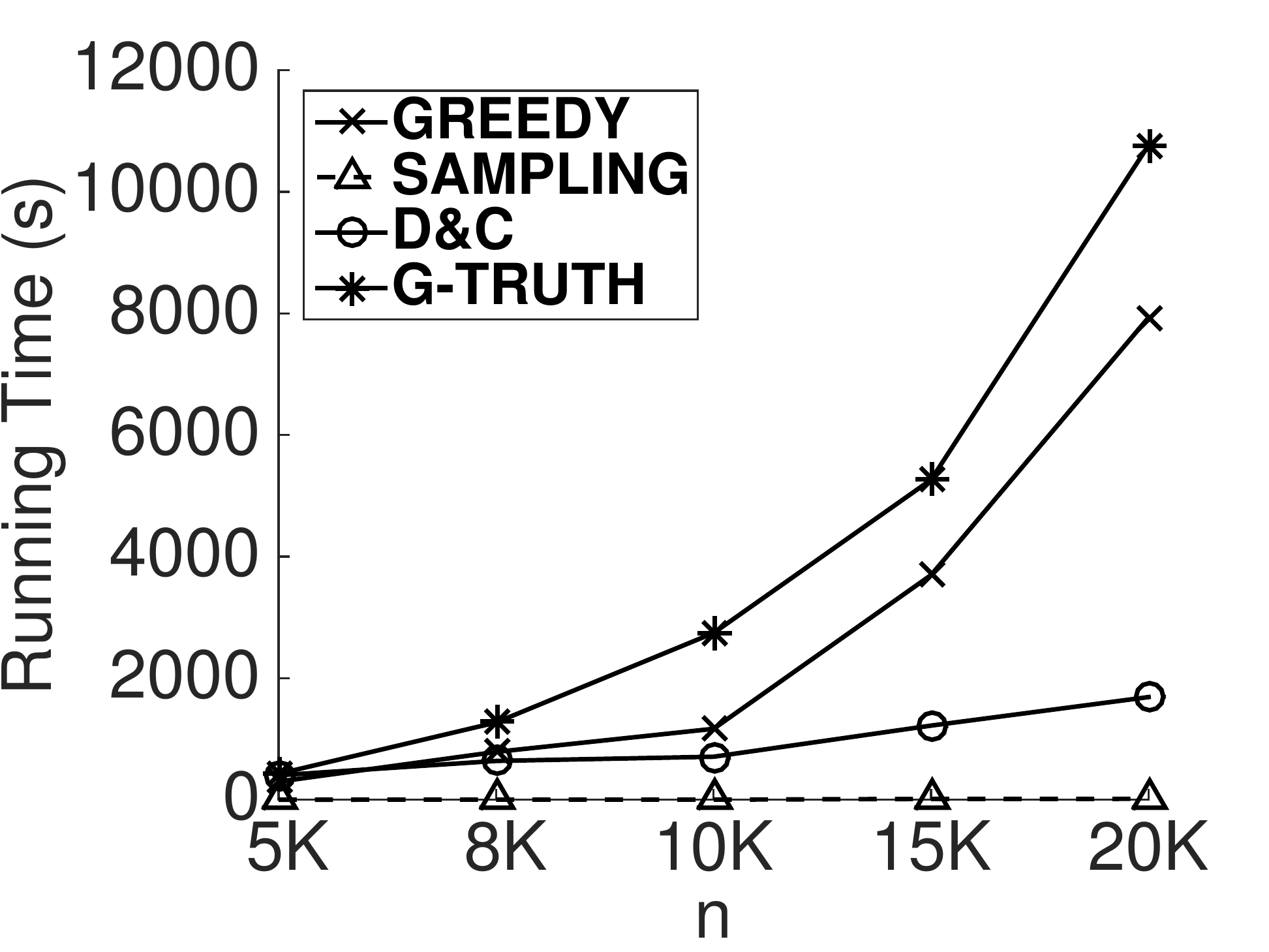}}\vspace{-3ex}
		\label{subfig:runningTimeWorkers}}\vspace{-3ex}
	\caption{\small Comparisons of the CPU Time with RDB-SC Approaches}\vspace{-2ex}
	\label{fig:runningTime}
\end{figure}

\begin{figure}[t]\centering
	\subfigure[][{\scriptsize Index Construction Time}]{
		\scalebox{0.18}[0.18]{\includegraphics{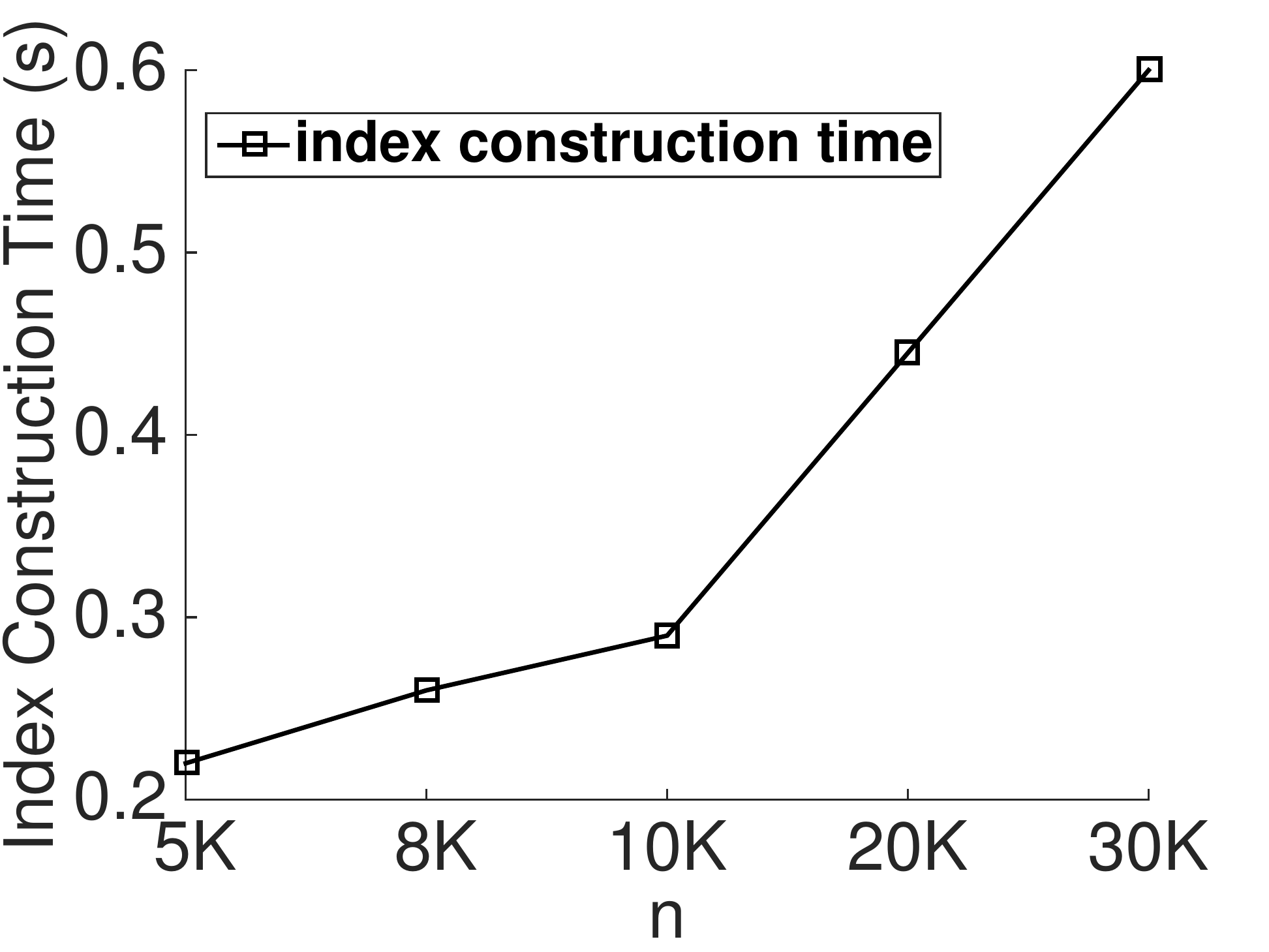}}\vspace{-2ex}
		\label{subfig:buildIndex}}
	\subfigure[][{\scriptsize W-T Pairs Retrieval Time}]{
		\scalebox{0.18}[0.18]{\includegraphics{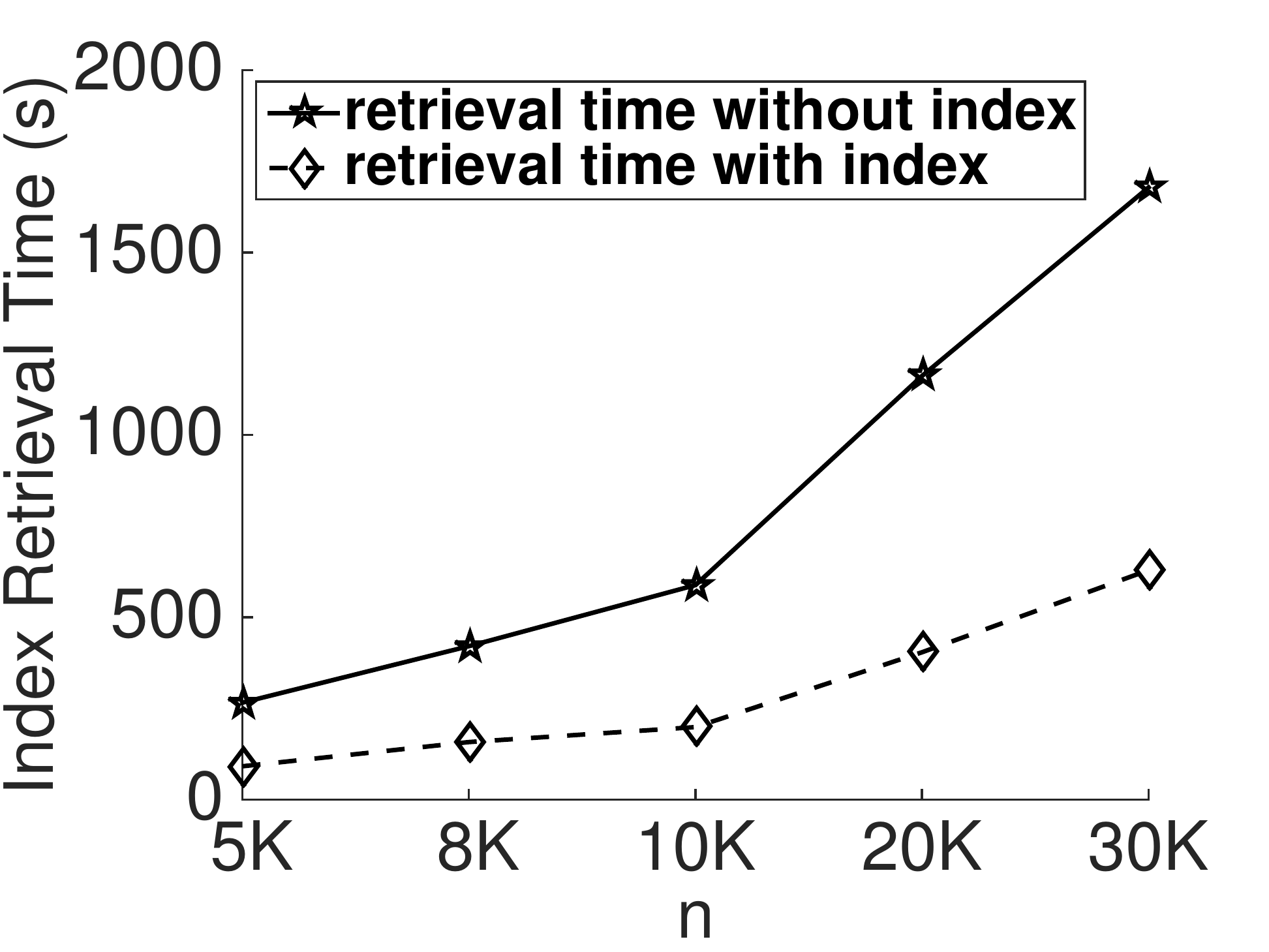}}\vspace{-2ex}
		\label{subfig:IndexEfficiency}}\vspace{-3ex}
	\caption{\small Efficiency of the RDB-SC-Grid Index}\vspace{-1ex}
	\label{fig:index}
\end{figure}

\vspace{0.5ex}\noindent {\bf Running
	Time Comparisons and Efficiency of RDB-SC-Grid.} We  report the running time of our approaches by varying $m$ and
$n$ in Figures \ref{subfig:runningTimeTasks} and
\ref{subfig:runningTimeWorkers}, respectively. We can see that, when
$m$ increases, the running times of all approaches, except for
SAMPLING, grow quickly. For GREEDY, when $m$ increases, each worker
has more reachable tasks, and thus the total running time grows
(since more tasks should be checked). For large $m$, D\&C needs to
run more rounds for the divide-and-conquer process, which leads to
higher running time. On the other hand, when $n$ increases, only
GREEDY's running time grows dramatically. This is because GREEDY
needs to run more rounds to assign workers. Under both situations,
SAMPLING only takes several seconds (due to small sample size). In
contrast, D\&C has higher CPU cost than SAMPLING, however, higher
reliability and diversity (as confirmed by Figures \ref{fig:taskN}
and \ref{fig:workerN}). This indicates that D\&C can trade the
efficiency for effectiveness (i.e., reliability and diversity).

Figure \ref{fig:index} presents the 
\textit{index construction time} and \textit{index retrieval time}
(i.e., the time cost for retrieving task-and-worker pairs, denoted
as W-T pairs, from the index) over UNIFORM data, where $m=10K$ and
$n$ varies from $5K$ to $30K$. As shown in Figure
\ref{subfig:buildIndex}, the construction time of the RDB-SC-Grid
index is small (i.e., less than 0.7 sec). In Figure
\ref{subfig:IndexEfficiency}, the RDB-SC-Grid index can dramatically
reduce the time of finding W-T pairs (up to 67 \%), compared with
that of retrieving W-T pairs without index.

\subsection{Experiments on Real RDB-SC Platform}
\label{subsec:expRealPlatform}

Figure \ref{fig:realPlatform} shows the RDB-SC performance of
GREEDY, SAMPLING, D\&C, and G-TRUTH over the real RDB-SC system,
where the length of the time interval, $t_{interval}$, between every
two consecutive incremental updates varies from 1 minute to 4
minutes. In Figure \ref{subfig:realPlatformR}, when $t_{interval}$
becomes larger, the minimum reliability remains high, except for
GREEDY. This is because when $t_{interval}$ is larger than 1 minute, 
GREEDY assigns just one worker to some tasks,
and it leads to sensitive change of the minimum reliability which is
much more than that of other algorithms. Each user is assigned with 
fewer tasks in the entire testing period, when $t_{interval}$ increases. 
At the same time, GREEDY is prone to assign workers to those tasks 
already have workers or are answered, which has been discussed when we 
show the effect of the number of tasks in Section \ref{subsec:expSynthetic}. 
When each user is assigned with fewer tasks in the entire testing period, it 
is more likely to assign only one worker to some task whose reliability will 
equal to the reliability of that worker. Thus, the minimum reliability of GREEDY varies much.

In Figure \ref{subfig:realPlatformD}, we can see that for all the
approaches, when $t_{interval}$ increases, the total
spatial/temporal diversity $total\_STD$ decreases. This is
reasonable, since each user is assigned with fewer tasks in the
entire testing period. Meanwhile, SAMPLING and D\&C are much better
than GREEDY from the perspective of the diversity, and their
diversities are close to that of G-TRUTH, which indicates the
effectiveness of our RDB-SC approaches.

\begin{figure}[t]\centering 
	\subfigure[][{\scriptsize Minimum Reliability}]{
		\scalebox{0.18}[0.18]{\includegraphics{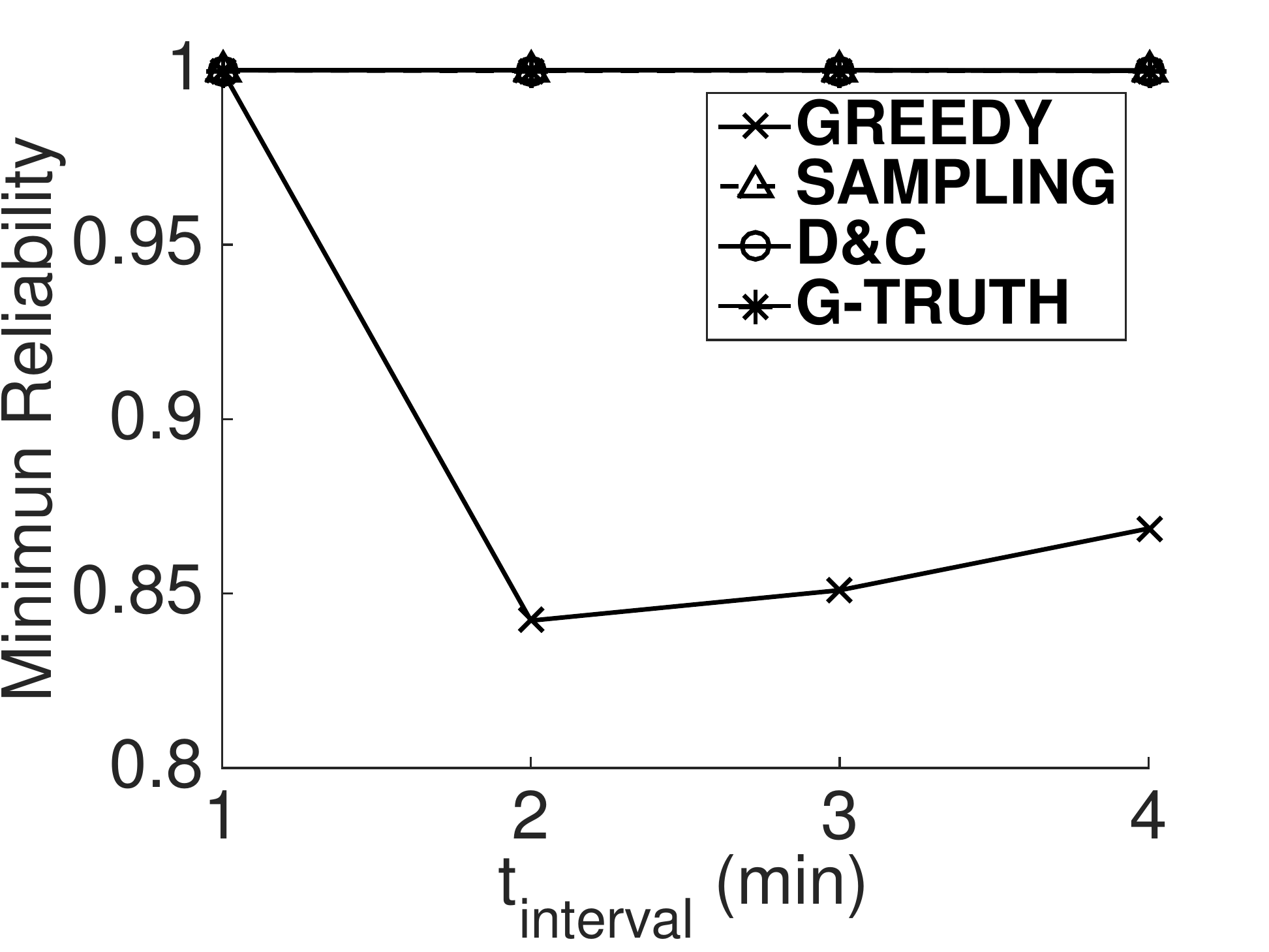}}\vspace{-2ex}
		\label{subfig:realPlatformR}}
	\subfigure[][{\scriptsize Summation of Diversity}]{
		\scalebox{0.18}[0.18]{\includegraphics{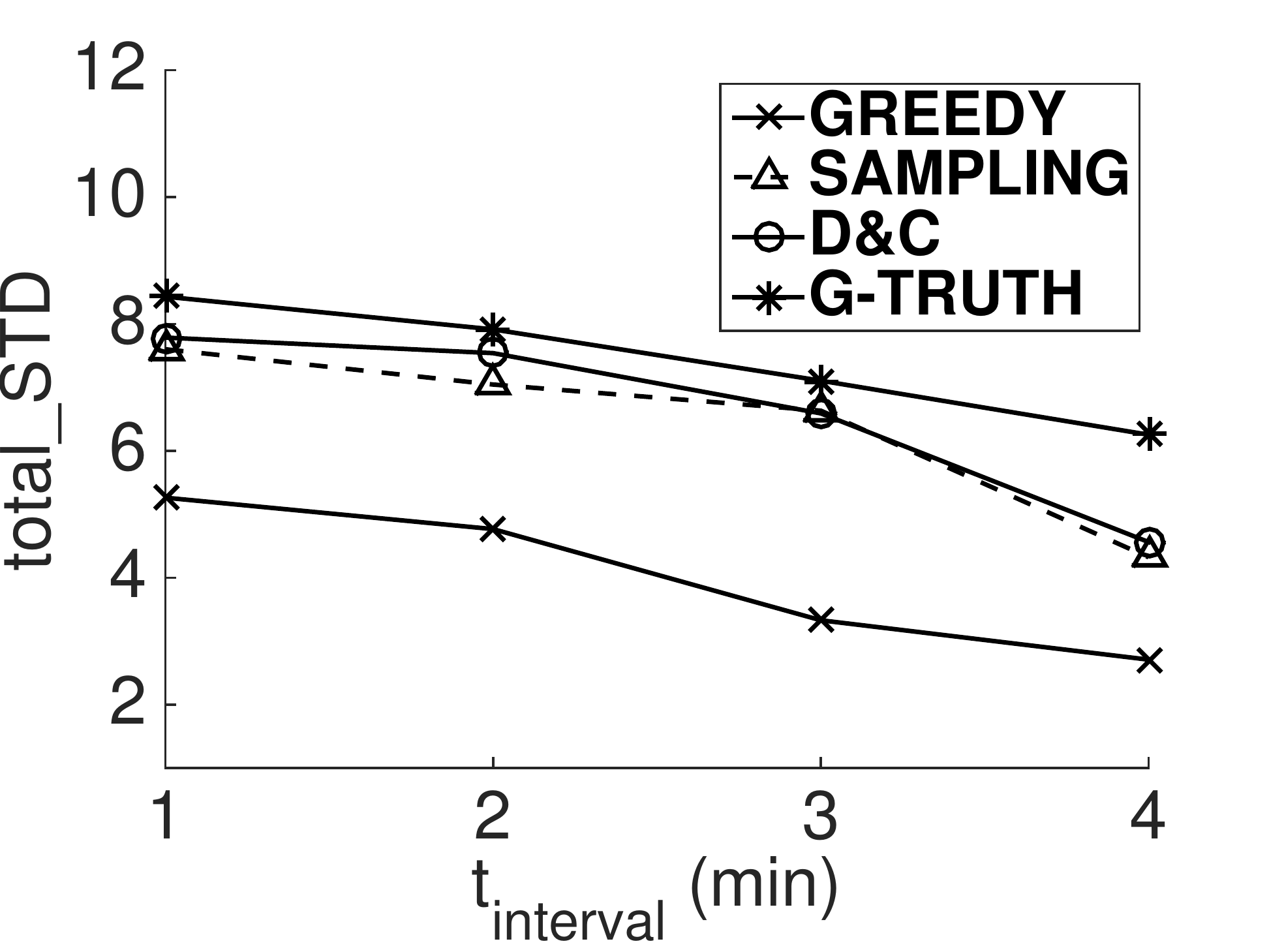}}\vspace{-2ex}
		\label{subfig:realPlatformD}}
	\vspace{-3ex}
	\caption{\small Effect of the updating time interval $t_{interval}$} \vspace{-2ex}
	\label{fig:realPlatform}
\end{figure}

To show the potential value of our model, we present a 3D reconstruction showcase on gMission's homepage \cite{gmissionhkust}.
We use VisualSFM \footnotemark to build 3D models from unordered photos. 
As shown in Figure \ref{fig:3dsparse}, the triangles are the cameras, and their sizes 
represent the resolutions of those photos. We can see, our approaches can assign 
workers to the task with a good spatial diversity. 
Figure \ref{fig:3dFull1} compares our final model from the side direction with that of the ground truth.
Though we just received 23 photos of the 
garden in the experiments, the 3D model reconstructed from our test can present a general 
shape of the garden, as shown in Figure \ref{subfig:fullExp1}.
 The showcase video can be found on Youtube \cite{3dvideo}.

\footnotetext{http://ccwu.me/vsfm/}

\begin{figure}[ht]\centering
	\subfigure[][{\scriptsize Experiemental Sparse Model}]{
		\scalebox{0.085}[0.085]{\includegraphics{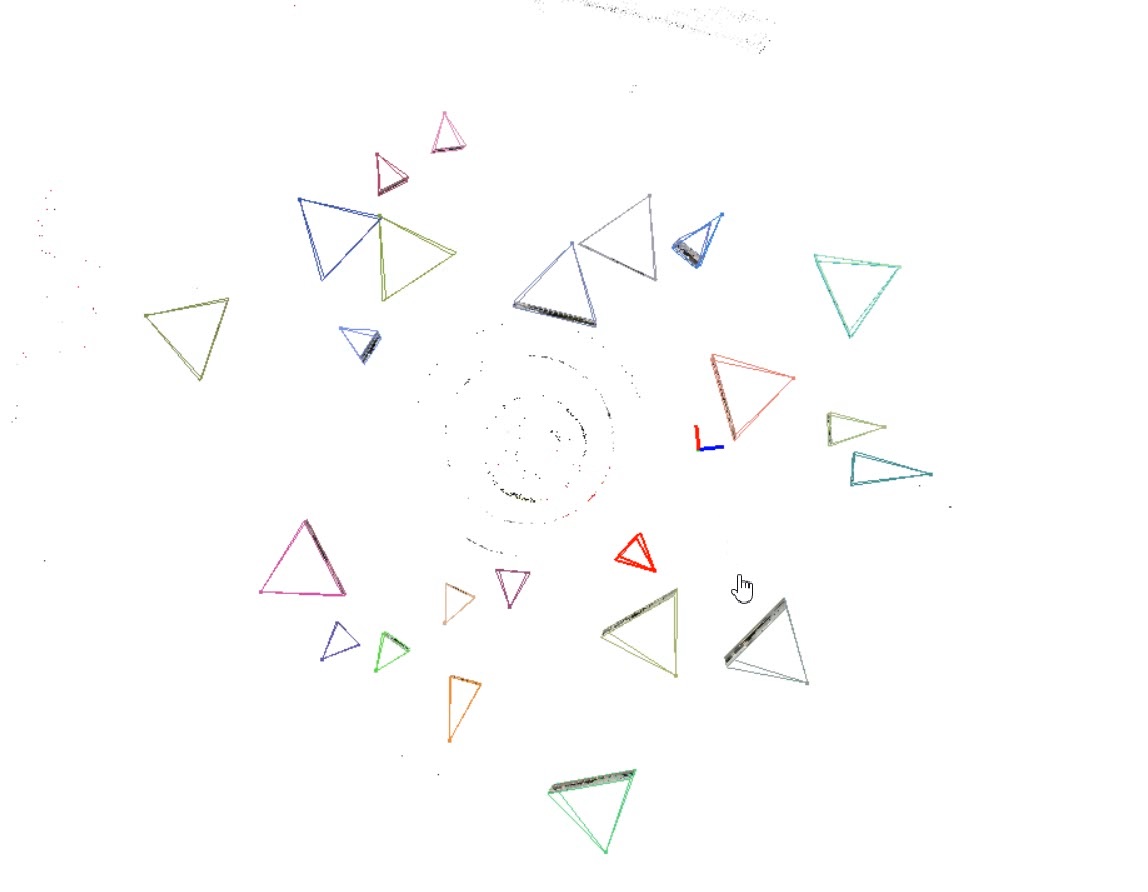}}\vspace{-2ex}
		\label{subfig:sparseExp}}
	\subfigure[][{\scriptsize Ground Truth Sparse Model}]{
		\scalebox{0.085}[0.085]{\includegraphics{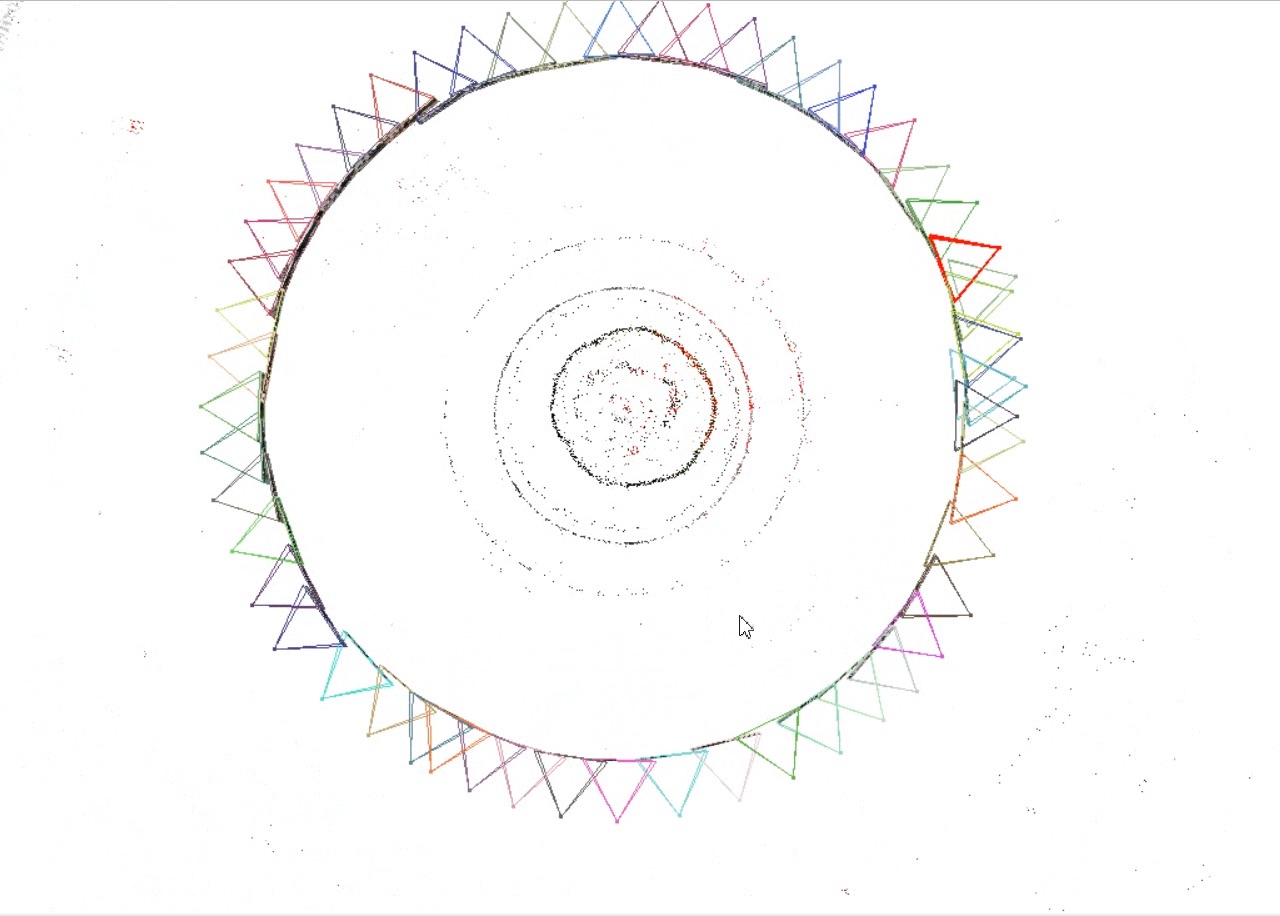}}\vspace{-2ex}
		\label{subfig:sparseTruth}}
	\vspace{-3ex}
	\caption{\small Comparison of Sparse Models of 3D Reconstruction} \vspace{-2ex}
	\label{fig:3dsparse}
\end{figure}

\begin{figure}[ht]\centering
	\subfigure[][{\scriptsize Experimental Dense Model}]{
		\scalebox{0.08}[0.08]{\includegraphics{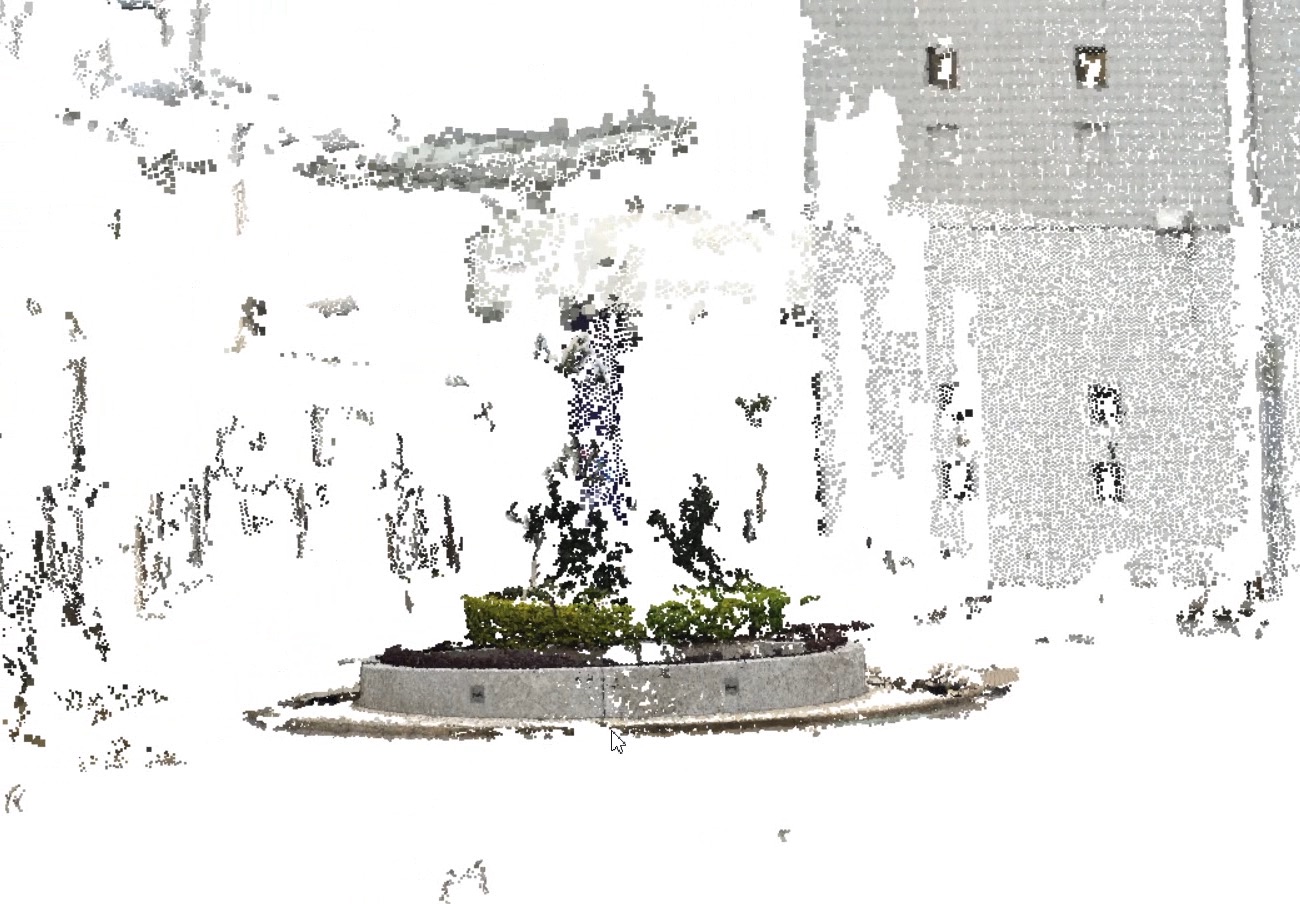}}\vspace{-2ex}
		\label{subfig:fullExp1}}
	\subfigure[][{\scriptsize Ground Truth Dense Model}]{
		\scalebox{0.08}[0.08]{\includegraphics{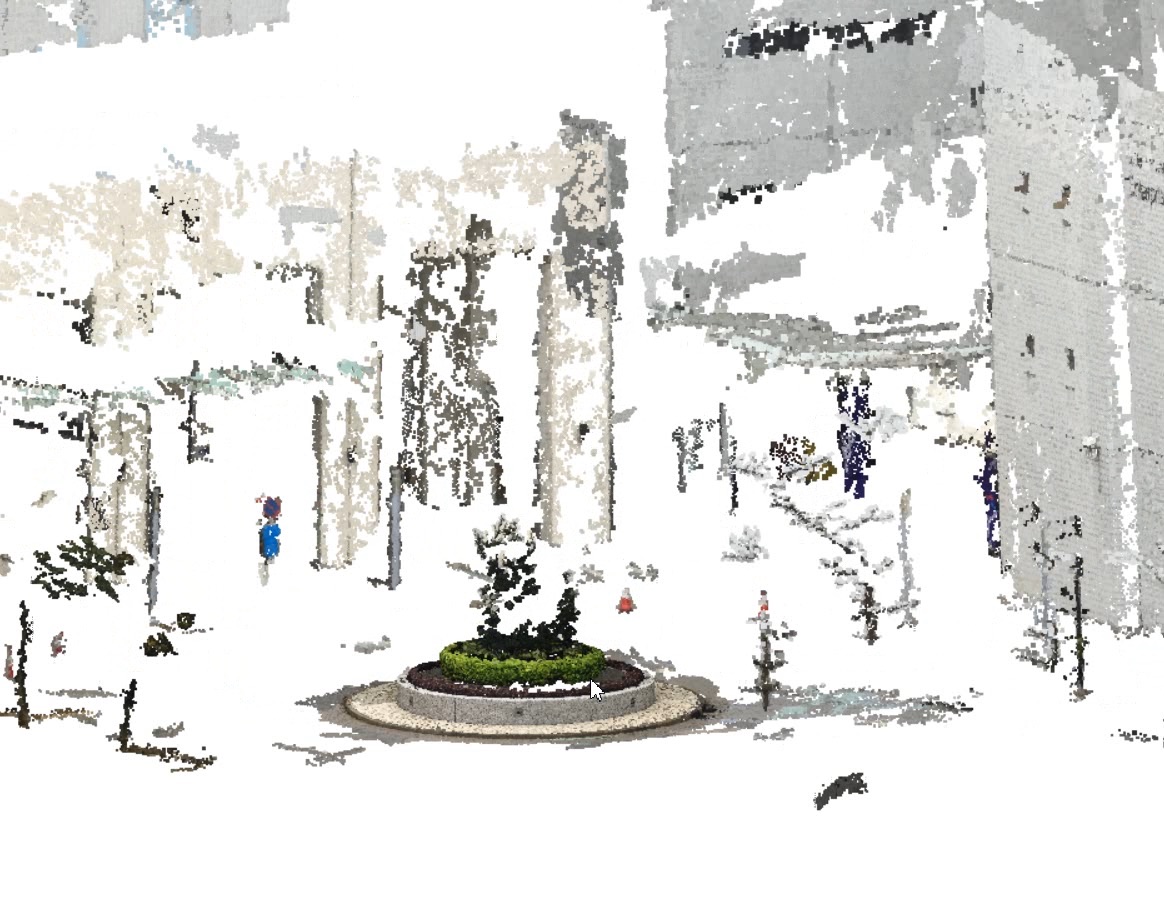}}\vspace{-2ex}
		\label{subfig:fullTruth1}}
	\vspace{-3ex}
	\caption{\small Comparison of Dense Models of 3D Reconstruction} \vspace{-2ex}
	\label{fig:3dFull1}
\end{figure}

\balance

\vspace{-1ex}
\section{Related Work}
\label{sec:related}

Recently, with rapid
development of GPS-equipped mobile devices, the spatial
crowdsourcing \cite{deng2013maximizing, kazemi2012geocrowd} that
sends location-based requests to workers (based on their spatial
positions) has become increasingly important in real applications,
such as monitoring real-world scenes (e.g., street view of Google
Maps \cite{GoogleMapStreetView}), local hotspots (e.g., Foursquare
\cite{foursquare}), and the traffic (e.g., Waze \cite{waze}).

Some prior works \cite{alt2010location,bulut2011crowdsourcing}
studied the crowdsourcing problems which treat location information
as the parameter, and distribute tasks to workers. However, in these
works, workers do not have to visit spatial locations physically (in
person) to complete the assigned tasks. In contrast, the spatial
crowdsourcing usually needs to employ workers to conduct tasks
(e.g., sensing jobs) by physically going to some specific positions.
For example, some previous works \cite{cornelius2008anonysense,
	kazemi2011privacy} studied the single-campaign or small-scale
participatory sensing problems, which focus on particular
applications of the participatory sensing.

According to people's motivation, Kazemi and Shahabi
\cite{kazemi2012geocrowd} classified the spatial crowdsourcing into
two categories: reward-based and self-incentivised. That is, in the
reward-based spatial crowdsourcing, workers can receive a small
reward after completing a spatial task; oppositely, for the
self-incentivised one, workers perform the tasks voluntarily (e.g.,
participatory sensing). In this paper, we consider the
self-incentivised spatial crowdsourcing.

Furthermore, based on publishing modes of spatial tasks, the spatial
crowdsourcing problems can be partitioned into another two classes:
\textit{worker selected tasks} (WST) and \textit{server assigned
	tasks} (SAT) \cite{kazemi2012geocrowd}. In particular, WST publishes
spatial tasks on the server side, and workers can choose any tasks
without contacting with the server; SAT collects location
information of all workers to the server, and directly assigns
workers with tasks. For example, in the WST mode, some existing
works \cite{alt2010location,deng2013maximizing} allowed users to
browse and accept available spatial tasks. On the other hand, in the
SAT mode, previous works \cite{kazemi2011privacy,
	kazemi2012geocrowd} assumed that the server decides how to assign
spatial tasks to workers and their solutions only consider simple
metrics such as maximizing the number of assigned tasks on the server
side and maximizing the number of worker's self-selected tasks.
In this paper, we not only consider the
SAT mode, but also take into account constrained features of
workers/tasks (e.g., moving directions of workers and valid period
of tasks), which make our problem more complex and unsuitable for
borrowing existing techniques.

Kazemi and Shahabi \cite{kazemi2012geocrowd} studied the spatial
crowdsourcing with the goal of static maximum task assignment, and proposed several heuristics
approaches to enable fast assignment of workers to tasks. Similarly,
Deng et al. \cite{deng2013maximizing} tackled the problem of
scheduling spatial tasks for a single worker such that the number of
completed tasks by this worker is maximized. In contrast, our work
has a different goal of maximizing the reliability and
spatial/temporal diversity that spatial tasks are accomplished. As
mentioned in Section \ref{sec:introduction}, the reliability and
diversity of spatial tasks are very important criteria in
applications like taking photos or checking whether or not parking
spaces are available. Moreover, while prior works often consider
static assignment, our work considers dynamic updates of spatial
tasks and moving workers, and proposes a cost-model-based index.
Therefore, previous techniques
\cite{deng2013maximizing,kazemi2012geocrowd} cannot be directly
applied to our RDB-SC problem.

Another important topic about the spatial crowdsourcing is the
privacy preserving. This is because workers need to report their
locations to the server, which thus may potentially release some
sensitive location/trajectory data. Some previous works
\cite{cornelius2008anonysense,kazemi2011privacy} investigate how to
tackle the privacy preserving problem in spatial crowdsourcing,
which is however out of the scope of this paper.

\vspace{-1ex}

\section{Conclusion}
\label{sec:conclusion}

In this paper, we propose the problem of reliable diversity-based
spatial crowdsourcing (RDB-SC), which assigns time-constrained
spatial tasks to dynamically moving workers, such that tasks can be
accomplished with high reliability and spatial/temporal diversity.
We prove that the processing of the RDB-SC problem is NP-hard, and
thus we propose three approximation algorithms (i.e., greedy,
sampling, and divide-and-conquer). We also design a cost-model-based
index to facilitate worker-task maintenance and RDB-SC answering.
Extensive experiments have been conducted to confirm the efficiency
and effectiveness of our proposed RDB-SC approaches on both real and
synthetic data sets.

\section{acknowledgment}
This work is supported in part by the Hong Kong RGC Project N\_HKUST637/13; National Grand Fundamental Research 973 Program of China under Grant 2014CB340303; NSFC under Grant No. 61328202, 61325013, 61190112, 61373175, and 61402359; and Microsoft Research Asia Gift Grant.

\bgroup\small

\bibliographystyle{abbrv}
\let\xxx=\bibitem\def\bibitem{\par\vspace{0.0mm}\xxx}
\bibliography{all,add}
\egroup

\balance
\section*{Appendix}

\noindent {\bf A. Proof of Lemma \ref{lemma:lem1}.}

\begin{proof}
A spatial diversity entry $M_{SD} [j] [k]$ represents the summation
of the related elements from all the possible worlds containing
angle $A_{j,k}$. In each $pw(W_i)$ of those possible worlds, the
related element is $-(\frac{A_{j,k}}{2\pi}) \cdot
log(\frac{A_{j,k}}{2\pi}) \cdot \prod_{w_x \in pw(W_i)}p_x \cdot
\prod_{w_y \notin pw(W_i)}(1 - p_y)$. It is easy to find that when
we sum up all the related elements, the summation is equal to
$M_{SD} [j] [k]$. $E(SD(t_i))$ is the summation of the diversities
of all the possible worlds multiply their probabilities
respectively. At the same time, all the possible angles will be
taken into consideration when we calculate  $E(SD(t_i))$. Thus,
after combining all the related elements for every angle, we get
$E(SD(t_i)) = \sum_{\forall j, k} M_{SD}[j] [k]$. Similarly, we can
prove that $E(TD(t_i)) = \sum_{\forall j, k} M_{TD} [j] [k]$.
\end{proof}

\noindent {\bf B. Proof of Lemma \ref{lemma:lem2}.}
\begin{proof}
We prove the lemma by a reduction from the number partition problem. A number partition problem can be described as follows: Given a set
$A = \{a_1, a_2, ..., a_N\}$ of positive integers, find a
partitioning strategy, i.e., $A_1 \subset A$ and $A_2 = A - A_1$,
such that the discrepancy $E(A)=\abs{\sum_{i \in A_1}a_i - \sum_{j
\in A_2} a_j}$ is minimized. For this given number partition
problem, we construct an instance of RDB-SC as follows: first, we
give 2 tasks and $n$ workers (each assigned to one of these 2
tasks), and let all the workers and tasks be on the same line, as
Figure \ref{fig:reduceNP}. shows. Moreover, let the beginning times of
the 2 tasks is late enough such that every worker can arrive at any
task before its expiration time. Under this setting, no matter how
to assign the workers, the $total\_STD$ is always zero.

\begin{figure}[ht]\vspace{-2ex}
\centering
       \scalebox{0.25}[0.25]{\includegraphics{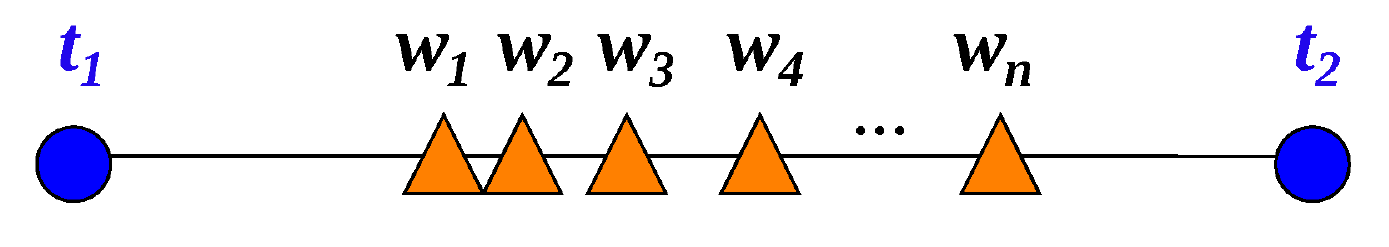}}\vspace{-2ex}
\caption{\small Illustration of the Task-and-Worker Locations.}
  \label{fig:reduceNP}\vspace{-2ex}
\end{figure}

Next, we set the reliability of the workers based on $A$. Let
$a_{max} = \max(A)$ and $a_i^\prime = a_i / a_{max}$. Then, let $p_i
= 1 - e^{a_i^\prime}$.

We now show that we can reduce the number partition problem to our
RDB-SC problem. In other words, the result of the number partition
problem is also the answer to our designed RDB-SC instance. The
answer of this RDB-SC is:\vspace{-2ex}

{\scriptsize
\begin{flalign}
&\text{maximize }  \min(\sum_{w_i \in W_1}-\ln(1-p_i), \sum_{w_j \in W_2}-\ln(1-p_j)) \notag\\
\implies &\text{minimize } \max(\sum_{w_i \in W_1}\ln(1 - p_i), \sum_{w_j \in W_2}\ln(1 - p_j))\notag\\
\implies &\text{minimize } \max(\sum_{w_i \in W_1}a_i^\prime,
\sum_{w_j \in W_2}a_j^\prime )\notag
\end{flalign}\vspace{-3ex}
}

As $\sum_{a_i \in A}  a_i/a_{max}$ is a constant, to minimize the
bigger subset sum is same as minimizing the discrepancy to the two
subsets, which is the objective of the number partition problem.
Given this mapping it is easy to show that the number partition problem instance can be solved if and only if the transformed RDB-SC problem instance can be solved. It completes our proof.
\end{proof}

\noindent {\bf C. Proof of Lemma \ref{lemma:lem3}.}
\begin{proof} According to the equivalent definition of reliability in
Eq.~(\ref{eq:eq8}), it holds that: $$R(t_i, W_i) =  \sum_{w_j\in
W_i} -ln (1-p_j).$$ Moreover, for the new set
$(W_i\cup\{w_{r+1}\})$, we have: $$R(t_i, W_i\cup \{w_{r+1}\}) =
\sum_{w_j\in (W_i\cup \{w_{r+1}\})} \hspace{-4ex}-ln (1-p_j).$$

By combining the two formulae above, we can infer that: $R(t_i,
W_i\cup \{w_{r+1}\}) = R(t_i, W_i) - ln (1-p_{r+1})$. Hence, the
lemma holds.
\end{proof}

\noindent {\bf D. Proof of Lemma \ref{lemma:lem4}.}
\begin{proof} Based on Eq.~(\ref{eq:eq6}), we have the
expected spatial/temporal diversity as follows:\vspace{-2ex}

{\scriptsize
\begin{eqnarray}
&&\hspace{-2ex}E(STD(t_i, W_i\cup\{w_{r+1}\}))\notag\\
&\hspace{-4ex}=& \hspace{-5ex}\sum_{\forall pw(W_i\cup \{w_{r+1}\})}
\hspace{-4ex}(Pr\{pw(W_i\cup \{w_{r+1}\})\}\cdot STD(t_i, pw(W_i\cup
\{w_{r+1}\})).\notag
\end{eqnarray}\vspace{-3ex}
}

Next, we divide all possible worlds of the worker set $W_i\cup
\{w_{r+1}\}$ into two disjoint parts $PW_1$ and $PW_2$, where all
possible worlds in $PW_1$ contain the new worker $w_{r+1}$, and
other possible worlds in $PW_2$ do not contain $w_{r+1}$. Therefore,
we can rewrite $E(STD(t_i, W_i\cup\{w_{r+1}\}))$ in the formula
above as:\vspace{-2ex}

{\scriptsize
\begin{eqnarray}
&&E(STD(t_i, W_i\cup\{w_{r+1}\}))\notag\\
&=&\sum_{\forall pw_1\in PW_1} (Pr\{pw_1\}\cdot STD(t_i, pw_1))\notag\\
&& + \sum_{\forall pw_2\in PW_2} (Pr\{pw_2\}\cdot
STD(t_i, pw_2))\notag\\
&=& p_{r+1} \cdot \hspace{-4ex}\sum_{\forall pw_1\in PW_1-\{w_{r+1}\}} \hspace{-4ex}(Pr\{pw_1\}\cdot STD(t_i, pw_1\cup \{w_{r+1}\}))\notag\\
&& + (1-p_{r+1}) \cdot  \sum_{\forall pw_2\in PW_2} (Pr\{pw_2\}\cdot
STD(t_i, pw_2))\notag
\end{eqnarray}\vspace{-3ex}
}

From the formula above, it is sufficient to prove that $STD(t_i,
pw_1\cup \{w_{r+1}\})\geq STD(t_i, pw_1)$ under a single possible
world $pw_1$. Due to Eq.~(\ref{eq:eq5}), alternatively, we prove
that $SD(t_i, pw_1\cup \{w_{r+1}\})\geq SD(t_i, pw_1)$ and $TD(t_i,
pw_1\cup \{w_{r+1}\})\geq TD(t_i, pw_1)$.

For the spatial diversity $SD(t_i, \cdot)$, without loss of
generality, we assume that worker $w_{r+1}$ divides the angle $A_r$
into two angles $A_{r1}$ and $A_{r2}$, where $A_r = A_{r1}+A_{r2}$.
Then, from Eq.~(\ref{eq:eq3}), we have:\vspace{-2ex}

{\scriptsize
\begin{eqnarray}
SD(t_i, pw_1\cup \{w_{r+1}\}) &=& - \sum_{j=1}^{r-1}
\frac{A_j}{2\pi} \cdot log \left(\frac{A_j}{2\pi}\right) -
\frac{A_{r1}}{2\pi} \cdot log
\left(\frac{A_{r1}}{2\pi}\right)\notag\\
&& - \frac{A_r - A_{r1}}{2\pi} \cdot log \left(\frac{A_r -
A_{r1}}{2\pi}\right)\notag
\end{eqnarray}\vspace{-3ex}
}

We next prove that $-x\cdot log x - (c-x)\cdot log (c-x) \geq -
c\cdot log c$ (for positive constant $c$). Since it holds that
$0\leq x\leq c\leq 1$ and $c-x\leq c$, we have $- log x \geq -log c$
and $- log (c-x) \geq -log c$. Thus, $-x\cdot log x - (c-x)\cdot log
(c-x) \geq - (x + (c-x)) \cdot log c = - c \cdot log c$. Hence, let
$c = \frac{A_r}{2\pi}$ in the formula above. We can
obtain:\vspace{-2ex}

{\scriptsize
\begin{eqnarray}
SD(t_i, pw_1\cup \{w_{r+1}\}) &\geq& - \sum_{j=1}^{r-1}
\frac{A_j}{2\pi} \cdot log \left(\frac{A_j}{2\pi}\right) -
\frac{A_{r}}{2\pi} \cdot log
\left(\frac{A_{r}}{2\pi}\right)\notag\\
&=& SD(t_i, pw_1).\notag
\end{eqnarray}\vspace{-3ex}
}

The case of temporal diversity can be proved similarly, and thus we
omit it. Therefore, the lemma holds.
\end{proof}

\noindent {\bf E. Proof of Lemma \ref{lemma:lem5}.}
\begin{proof}
Only when a worker is assigned to the smallest reliability task,
$\Delta min\_R(t_i, w_j)\geq 0$. Inequality (1) means the pair
$(t_i, w_j)$ can increase the smallest reliability not less than the
pair $(t_i', w_j')$. According to the reliability objective, we
should choose the pair $(t_i, w_j)$. Moreover, inequality (2) means
the pair $(t_i, w_j)$ can improve the diversity of task $t_i$ more
than the pair $(t_i', w_j')$ to task $t_i'$. To increase the
$total\_STD$ we should choose the pair $(t_i, w_j)$. Hence, when the
two inequalities are satisfied at the same time, we can prune the
pair $(t_i', w_j')$.
\end{proof}

\noindent {\bf F. Derivation of Eq. \ref{eq:eq20}.}
Given parameters $p$, $\epsilon$, and
$\delta$, we want to decide the value of parameter $K$ with high
confidence. That is, we have:$Pr\{X > (1-\epsilon)\cdot N\} > \delta. $

Equivalently, we can rewrite it as: \vspace{-1.5ex}
{\scriptsize
	\begin{eqnarray}
	Pr\{X \leq (1-\epsilon)\cdot N\} \leq 1- \delta.\label{eq:eq15}
	\end{eqnarray}\vspace{-5ex}
}

Let $M = (1-\epsilon)\cdot N$. By combining Eqs.~(\ref{eq:eq13}),
~(\ref{eq:eq14}), and ~(\ref{eq:eq15}), we can derive the
probability below:\vspace{-2ex}

\text{ }\vspace{-1ex} {\scriptsize
	\begin{eqnarray}
	Pr\{X \leq M\} &=& \sum_{i=1}^M \dbinom{i-1}{K-1} \cdot p^K \cdot
	(1-p)^{N-K}\label{eq:eq16} \vspace{-5ex}\\
	&=& p^K \cdot (1-p)^{N-K}\cdot \sum_{i=K}^M
	\dbinom{i-1}{K-1}\notag \vspace{-5ex}\\
	&=& p^K \cdot (1-p)^{N-K}\cdot\left(\dbinom{K}{K} + \sum_{i=K+1}^M
	\dbinom{i-1}{K-1}\right)\notag
	\end{eqnarray} \vspace{-1ex}
}

Since it holds that:{\scriptsize $\dbinom{Y}{X} + \dbinom{Y}{X-1} =
	\dbinom{Y+1}{X}$}, we can rewrite Eq.~(\ref{eq:eq16})
as:\vspace{-4ex}

{\scriptsize
	\begin{eqnarray}
	Pr\{X \leq M\} &=& (1-p)^N \left(\frac{p}{1-p}\right)^K\cdot
	\dbinom{M}{K} \notag\\
	&=& C_1 \cdot C_2^K \cdot \frac{M!}{K! \cdot (M-K)!}\label{eq:eq17}
	\end{eqnarray}\vspace{-3ex}
}

\noindent where constants $C_1 = (1-p)^N$ and $C_2 =\frac{p}{1-p}$.

Up to now, our problem is reduced to the one to find $K$ such
that:\vspace{-2ex}

{\scriptsize
	\begin{eqnarray}
	C_1 \cdot C_2^K \cdot \frac{M!}{K! \cdot (M-K)!} \leq 1-\delta.
	\label{eq:eq18}
	\end{eqnarray}\vspace{-2ex}
	
}

Since the factorial is not a continuous function and cannot take the
derivatives, we convert the formula above into a continuous function
by the Gamma function.

Let $F(x) =C_1\cdot C_2^x\cdot \frac{M!}{\Gamma(x+1)\cdot
	\Gamma(M-x+1)}$. 

\text{}\vspace{-3ex}
{\scriptsize
	\begin{eqnarray}
	\frac{\partial F(x)}{\partial x} &=& C_1 \cdot M! \cdot
	(\Gamma(x+1)\cdot\Gamma(M-x+1)\cdot ln C_2\cdot C_2^x\notag\\
	&& - C_2^x\cdot (x! \cdot (-e+\sum_{i=1}^x \frac{1}{i})\cdot
	\Gamma(M-x+1)\notag\\
	&& - (M-x)!(-e+\sum_{i=1}^{M-x} \frac{1}{i})\cdot
	\Gamma(x+1)))\notag\\
	&&/(\Gamma(x+1)\cdot \Gamma(M-x+1))^2.\notag
	\end{eqnarray}\vspace{-3.5ex}
}

We want to find $K$ value, such that $F(K)\leq 1-\delta$ (equivalent
to Eq.~(\ref{eq:eq18})). Thus, we let $\frac{\partial F(x)}{\partial
	x}<0$ (i.e., $F(x)$ decreases with the increase of $x$), and can
obtain:\vspace{-3ex}

{\scriptsize
	\begin{eqnarray}
	ln C_2 + \sum_{j=K+1}^{M-K} \frac{1}{j} <0 \label{eq:eq19}
	\end{eqnarray}\vspace{-3ex}
}

Based on the lower/upper bounds of \textit{Harmonic series}, we
have: \vspace{-3ex}

{\scriptsize
	$$ln (x+1) < \sum_{j=1}^x \frac{1}{j} <ln (x) +1. $$\vspace{-3ex}
}

In order to let Inequality (\ref{eq:eq19}) hold, we relax the
Harmonic series in Inequality (\ref{eq:eq19}) with their lower/upper
bounds, and obtain the condition below:\vspace{-3.5ex}

{\scriptsize
	\begin{eqnarray}
	\sum_{j=K+1}^{M-K} \frac{1}{j} &=& \sum_{j=1}^{M-K} \frac{1}{j} -
	\sum_{j=1}^{K} \frac{1}{j}\notag\\
	&\leq& ln (M-K) + 1 - ln (K+1)= ln \frac{M-K}{K+1} + 1 \notag\\
	&<& - ln C_2.\notag
	\end{eqnarray}\vspace{-4ex}
}

\noindent which can be simplified as: 

{\scriptsize
	\vspace{-2.5ex}
	\begin{eqnarray}
	K>\frac{p\cdot M\cdot e -1+p}{1-p+e\cdot p}. 
	\end{eqnarray}\vspace{-3.5ex}
}

Since $\frac{\partial F(x)}{\partial x}<0$ holds for $K$ satisfying
Eq.~(\ref{eq:eq20}), $F(K)$ monotonically decreases with the
increase of $K$. We can thus conduct a binary search for
$\widehat{K}$ value within $\left(\frac{p\cdot M\cdot e
	-1+p}{1-p+e\cdot p}, M\right]$, such that $\widehat{K}$ is the
smallest $K$ value such that $F(K)\leq 1-\delta$, where $p =
\prod_{j=1}^n \frac{1}{deg(w_j)}$.

\noindent {\bf G. Proof of Lemma \ref{lemma:lem6}.}
\begin{proof}
Assume that we remove a conflicting worker $w_j$ from task $t_i$. We
try to proof that the increase of diversity of any other worker
$w_k$ will not decrease after $w_j$ is removed when there have other
workers assigned to task $t_i$.

We divide all possible worlds of the worker set $W_i$ into two
disjoint parts $PW_1$ and $PW_2$, where all possible worlds in
$PW_1$ contain $w_k$, and in $PW_2$ not contain $w_k$. Then the
diversity increase of worker $w_k$ is presented as:\vspace{-2ex}

{\scriptsize
\begin{eqnarray}
\bigtriangleup E(STD(t_i, W_i, w_k))
&=& E(STD(t_i, W_i)) - E(STD(t_i, W_i - w_k))\notag\\
&=& p_{k} \cdot \hspace{-2ex}\sum_{\forall pw_1\in PW_2} \hspace{-3ex}(Pr\{pw_1\}\cdot STD(t_i, pw_1\cup \{w_{k}\}))\notag\\
&& + (1-p_{k}) \cdot \hspace{-3ex} \sum_{\forall pw_2\in PW_2} \hspace{-3ex} (Pr\{pw_2\}\cdot STD(t_i, pw_2))\notag\\
&& - \hspace{-3ex}\sum_{\forall pw_2\in PW_2} \hspace{-3ex}(Pr\{pw_2\}\cdot STD(t_i, pw_2))\notag\\
&=& p_k \cdot \hspace{-2ex}\sum_{\forall pw\in PW_2} \hspace{-3ex}(Pr\{pw\}\cdot (STD(t_i, pw\cup \{w_{k}\}).\notag \\
&& -STD(t_i, pw))) \notag
\end{eqnarray}\vspace{-3ex}
}

Let $PW_3$ represent all the possible worlds that do not contain
$w_i$ or $w_k$. Similarly, we can get the differential increase of
worker $w_i$ as:\vspace{-2ex}

{\scriptsize
\begin{eqnarray}
&& \bigtriangleup E(STD(t_i, W_i - w_j, w_k)) - \bigtriangleup E(STD(t_i, W_i, w_k))\notag\\
&=& p_k \cdot p_j \cdot \hspace{-2ex}\sum_{\forall pw\in PW_3} \hspace{-3ex}(Pr\{pw\}\notag\\
&&\cdot ((STD(t_i, pw\cup \{w_{k}\})-STD(t_i, pw )))\notag \\
&& - (STD(t_i, pw\cup \{w_{k}, w_{j}\})-STD(t_i, pw\cup
\{w_{j}\})))). \notag
\end{eqnarray}\vspace{-3ex}
}

Like the proof of Lemma \ref{lemma:lem4}, we prove under a single
possible world $pw_3 \in PW_3$, $(STD(t_i, pw_3\cup
\{w_{k}\})-STD(t_i, pw_3 ))) \ge (STD(t_i, pw_3\cup \{w_{k},
w_{j}\})-STD(t_i, pw_3\cup \{w_{j}\})$. Due to Eq.~(\ref{eq:eq5}),
we prove the spatial diversity and temporal diversity one by one.
For the spatial diversity $SD(t_i,\cdot)$, without loss of
generality, we assume that worker $w_k$ divides the angle $A_k$ into
two angles $A_{k1}$ and $A_{k2}$, where $A_k = A_{k1} + A_{k2}$.
$w_j$ divides the angle $A_j$ into $A_{j1}$ and $A_{j2}$. If $A_k$
and $A_j$ have no overlap, we can directly know  $(SD(t_i, pw_3\cup
\{w_{k}\})-SD(t_i, pw_3))) = (SD(t_i, pw_3\cup \{w_{k},
w_{j}\})-SD(t_i, pw_3\cup \{w_{j}\})$. Otherwise, when there is not
$w_j$, either angle $A_{k1}$ or $A_{k2}$ will be enlarged. We assume
angle $A_{k1}$ will be enlarged as $A_{k1}^\prime$. At the same
time, let $A_k^\prime = A_{k1}^\prime + A_{k2}$. According to
Eq.~(\ref{eq:eq3}), we have:\vspace{-2ex}

{\scriptsize
\begin{eqnarray}
&&  SD(t_i, pw_3\cup \{w_{k}, w_j\})-SD(t_i, pw_3 \cup \{w_j\})))\notag\\
&=& (\frac{A_{k}}{2\pi})\cdot\log(\frac{A_{k}}{2\pi}) -
(\frac{A_{k1}}{2\pi})\cdot\log(\frac{A_{k1}}{2\pi}) -
(\frac{A_{k2}}{2\pi})\cdot\log(\frac{A_{k2}}{2\pi}).\notag
\end{eqnarray}\vspace{-3ex}
}

Then, we have:\vspace{-2ex}

{\scriptsize
\begin{eqnarray}
&&  (SD(t_i, pw_3\cup \{w_{k}\})-SD(t_i, pw_3)))\notag\\
&& - (SD(t_i, pw_3\cup \{w_{k}, w_{j}\})-SD(t_i, pw_3\cup \{w_{j}\})\notag \\
&=& ((\frac{A_{k}^\prime}{2\pi})\cdot\log(\frac{A_{k}^\prime}{2\pi}) - (\frac{A_{k}}{2\pi})\cdot\log(\frac{A_{k}}{2\pi})) \notag \\
&& -
((\frac{A_{k1}^\prime}{2\pi})\cdot\log(\frac{A_{k1}^\prime}{2\pi}) -
(\frac{A_{k1}}{2\pi})\cdot\log(\frac{A_{k1}}{2\pi})).\notag
\end{eqnarray}\vspace{-3ex}
}

We next prove that $f(x) = (x+\epsilon)\log(x+\epsilon) - x\log(x)$
is an increasing function, where $\epsilon > 0$. As $f(x)^\prime =
\log(x+\epsilon) - \log(x) > 0$, $f(x)$ is an increasing function.
Thus, $(SD(t_i, pw_3\cup \{w_{k}\})-SD(t_i, pw_3)))\notag - (SD(t_i,
pw_3\cup \{w_{k}, w_{j}\})-SD(t_i, pw_3\cup \{w_{j}\}) > 0$ when
$A_k$ overlap $A_j$. Similarly, we can prove the case of temporal
diversity, and we omit it to save the space.

Next, we prove that, when task $t_i$ just has worker $w_j$ and
$w_k$, $w_k$ should not leave task $t_i$ if $w_j$ is removed. If
both $w_j$ and $w_k$ are removed from $t_i$, the reliability of
$t_i$ will become zero, which must be the minimum reliability. This
situation goes against our goal of maximizing the minimum
reliability and should not happen. In conclusion, this lemma holds.
\end{proof}

\noindent {\bf H. Proof of Lemma \ref{lemma:lem7}.}
\begin{proof}
According to Lemma \ref{lemma:lem6}, deleting copies of conflicting
workers will not affect the assignments of non-conflicting workers.
The only situation we need to consider is that some conflicting
workers are assigned to some tasks. In other words, they can connect
with each other through some tasks. For these DCWs, we just need to
enumerate all the possible combinations and pick the best one. If
the size of a group of DCWs is $k$, we have $2^k$ combinations to
check.
\end{proof}

\begin{figure}[ht!]\centering\vspace{-1ex}
	\subfigure[][{\scriptsize Minimum Reliability}]{
		\scalebox{0.18}[0.18]{\includegraphics{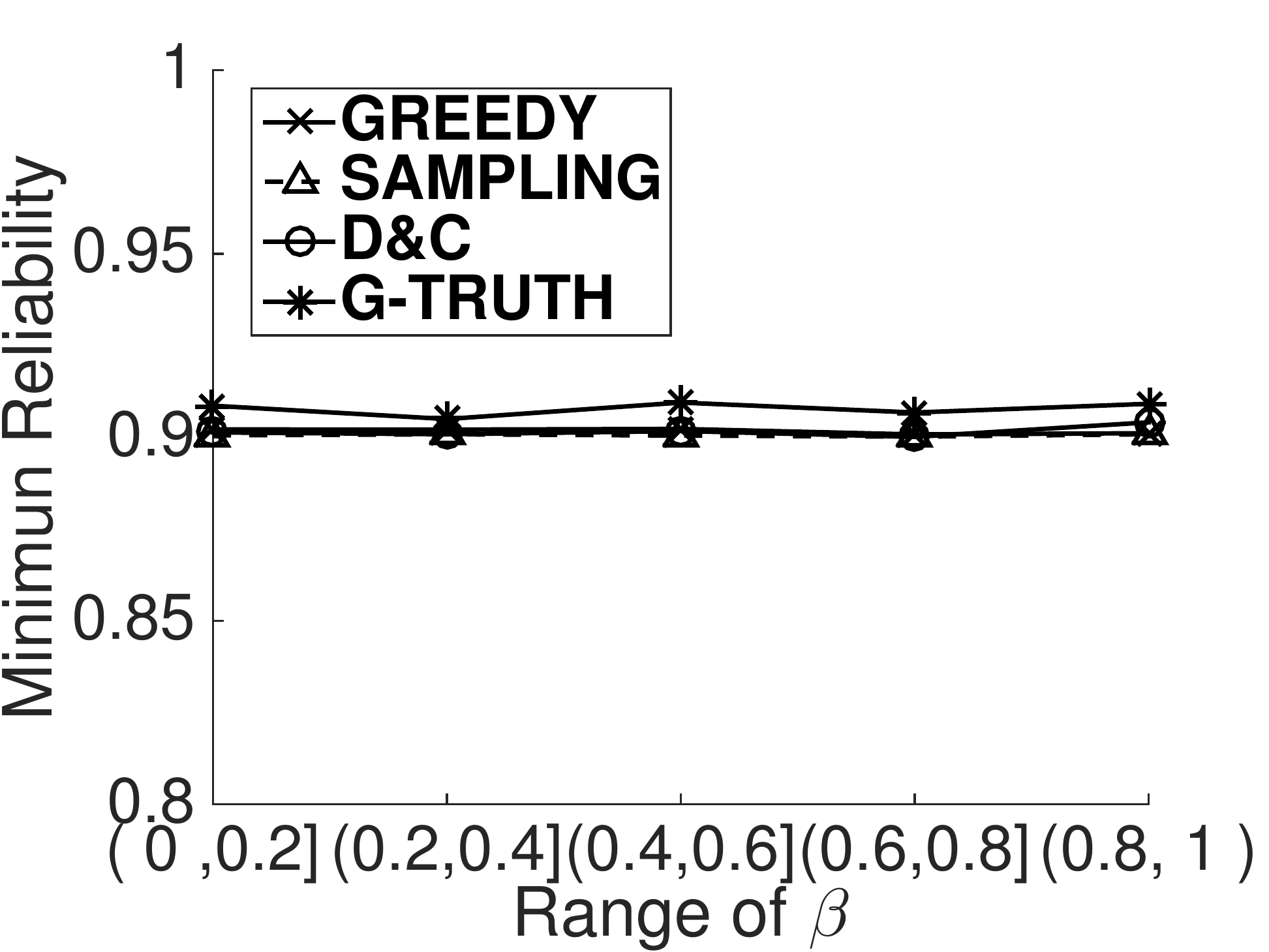}}\vspace{-2ex}
		\label{subfig:BetaReliability}}
	\subfigure[][{\scriptsize Summation of Diversity}]{
		\scalebox{0.18}[0.18]{\includegraphics{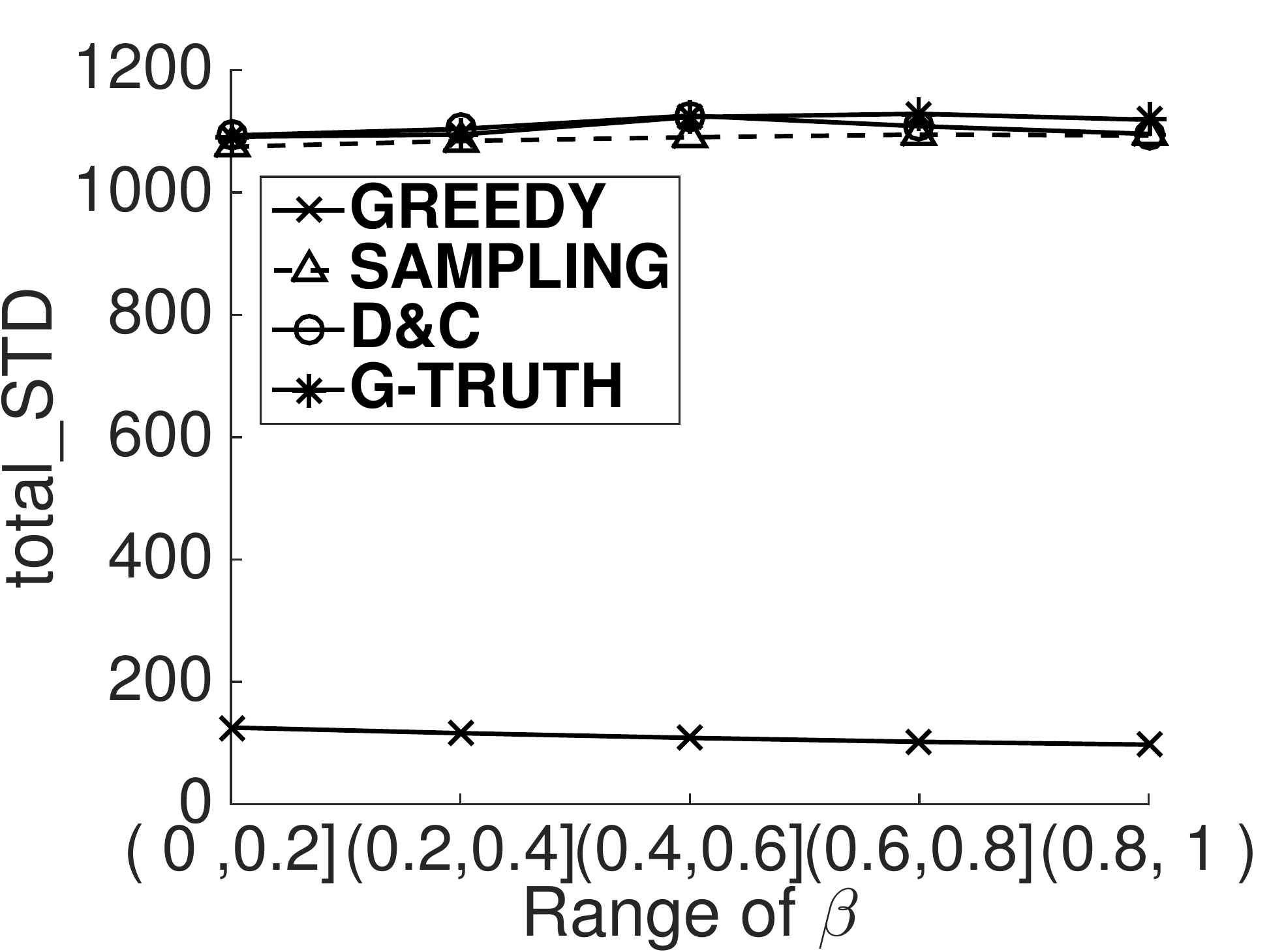}}\vspace{-2ex}
		\label{subfig:BetaDiversity}}
	\vspace{-2ex}
	\caption{\small Effect of the Range of Requester-Specified Weight $\beta$} \vspace{-2ex}
	\label{fig:betaEffect}
\end{figure}

\noindent {\bf I. Cost Model of RDB-SC-Grid.}

In this subsection, we illustrate how to set the best cell size
$\eta$ based on our proposed cost model. In the sequel, we first
propose a cost model to formalize the updating cost of the
RDB-SC-Grid index. Then, we use the cost model to guide how to set
$\eta$.

In general, the update cost of the RDB-SC-Grid index includes two
parts, the cost of retrieving cells in the reachable area and that
of checking every task in each candidate cell. Since workers do not
move very far in general, we design the cost model based on the
moving distance history. According to the moving history of workers,
we can get the maximum moving distance, $L_{max}$, for all the
workers. In the worst case, we need to check $\frac{\pi\cdot(L_{max}
	+ \eta)^2}{\eta^2}$ cells. To estimate the number of tasks,
$\overline{nt}$, in these cells, we use the \textit{power law}
\cite{belussi1998self} for the \textit{correlation fractal
	dimension} $D_2$ of the tasks in the 2D data space. Then, we have
$\overline{nt} = (N - 1)\cdot(\pi\cdot(L_{max}+\eta)^2)^{\frac{D_2}{2}}$. As the power law
is applicable to both uniform data and nonuniform data in real world
\cite{belussi1998self}, we use it to estimate the tasks in the
reachable area. With the cost of these two parts, we can get the
updating cost, $cost_{update}$, of the RDB-SC-Grid index as
follows:\vspace{-3ex}

{\scriptsize
	\begin{eqnarray}
	cost_{update} = \frac{\pi\cdot(L_{max} + \eta)^2}{\eta^2} + (N -
	1)\cdot(\pi\cdot(L_{max}+\eta)^2)^{\frac{D_2}{2}}.\label{eq:eq21}
	\end{eqnarray}\vspace{-3ex}
}

Therefore, our goal is to set the best $\eta$ value for our grid
index, such that $cost_{update}$ is minimized in
Eq.~(\ref{eq:eq21}). In particular, we take the derivative of
$cost_{update}$ with respect to $\eta$, and let it equal to 0, that
is, $\frac{\partial cost_{update}}{\partial \eta}=0$, which can be
simplified as:\vspace{-2ex}

{\scriptsize
	\begin{eqnarray}
	(L_{max} + \eta)^{D_2 - 2}\cdot\eta^3 = \frac{2\pi^{1-D_2/2}\cdot
		L_{max}}{D_2 \cdot(N - 1)}.\label{eq:eq22}
	\end{eqnarray}\vspace{-3ex}
}

Next, we can calculate $L_{max}$ and $D_2$ by collecting statistics
from historical data, and then estimate a proper cell size $\eta$
according to Eq.~(\ref{eq:eq22}). The resulting $\eta$ can minimize
the updating cost of the RDB-SC-Grid index. When we do not have the
historical data at the beginning, we can only assume that data are
uniform such that $D_2=2$. As a result, we have $\eta =
\sqrt[3]{\frac{L_{max}}{N-1}}$.

\noindent {\bf J. Figures on UNIFORM/SKEWED synthetic data.}

We first present the effect of the range of Requester-Specified
weight $\beta$ to our approaches on real data sets.
Then, we show the  effects of the parameters listed in Table \ref{table2} to our approaches on
SKEWED synthetic data sets.

\vspace{0.5ex}\noindent {\bf Effect of the Range of
Requester-Specified Weight $\beta$.} Figure \ref{fig:betaEffect}
illustrates the effect of the range of requester-specified weight,
$\beta$, on the reliability/diversity over real data set, where
other parameters use their default values. In figures, when the
range changes from $(0, 0.2]$ to $(0.8, 1)$, both reliability and
diversity are not very sensitive, which indicates the robustness of
the 3 approaches against weight $\beta$ for the spatial/temporal
diversity. For all the tested ranges, both minimum reliability and
diversity remain high (i.e., with reliability above 0.9, and the
diversity of SAMPLING and D\&C close to that of G-TRUTH).

\begin{figure}[ht!]\centering
    \vspace{-1ex}
     \subfigure[][{\scriptsize Minimum Reliability}]{
       \scalebox{0.18}[0.18]{\includegraphics{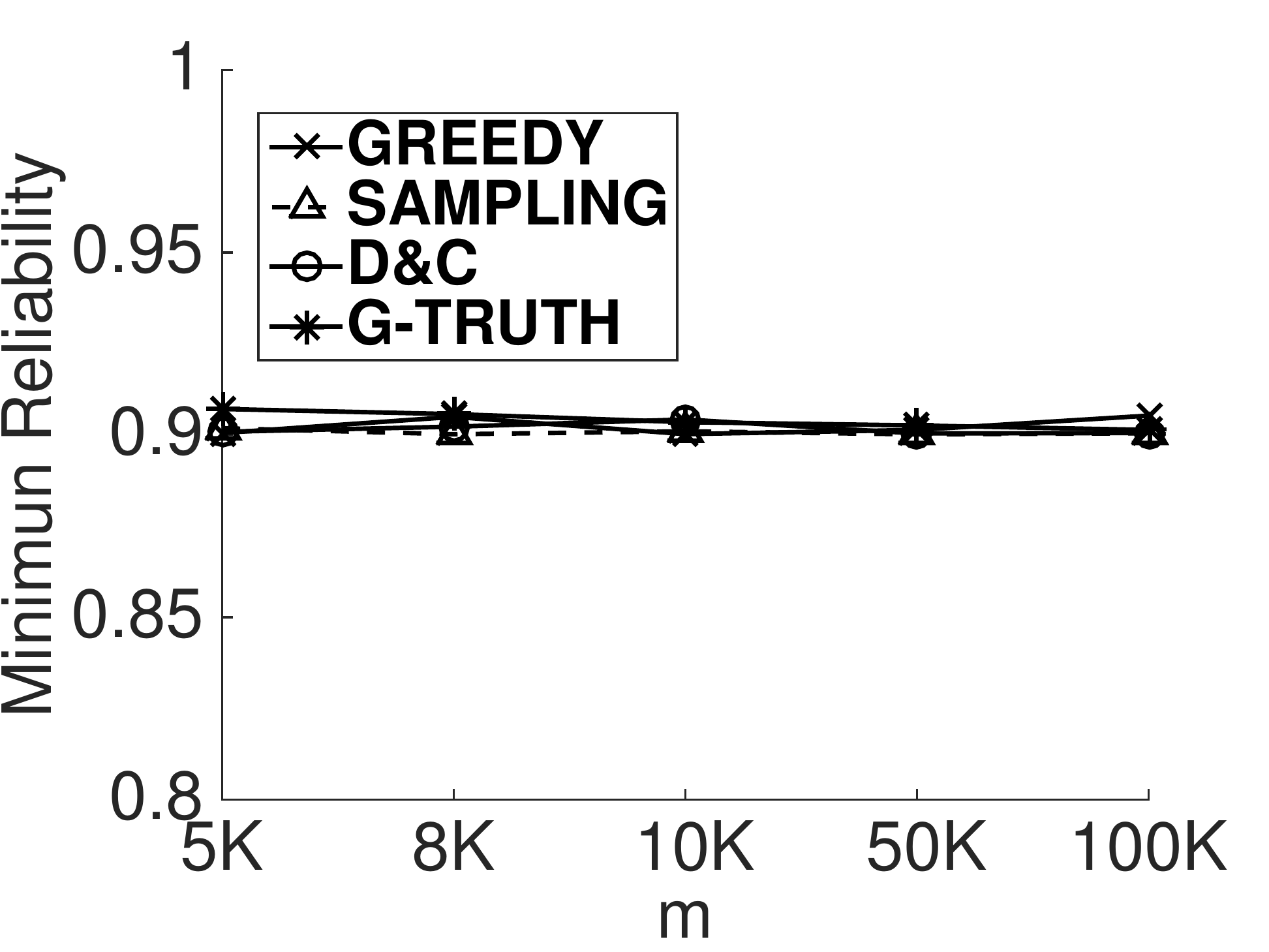}}
     \label{subfig:taskNSkewedReliability}}
     \subfigure[][{\scriptsize Summation of Diversity}]{
       \scalebox{0.18}[0.18]{\includegraphics{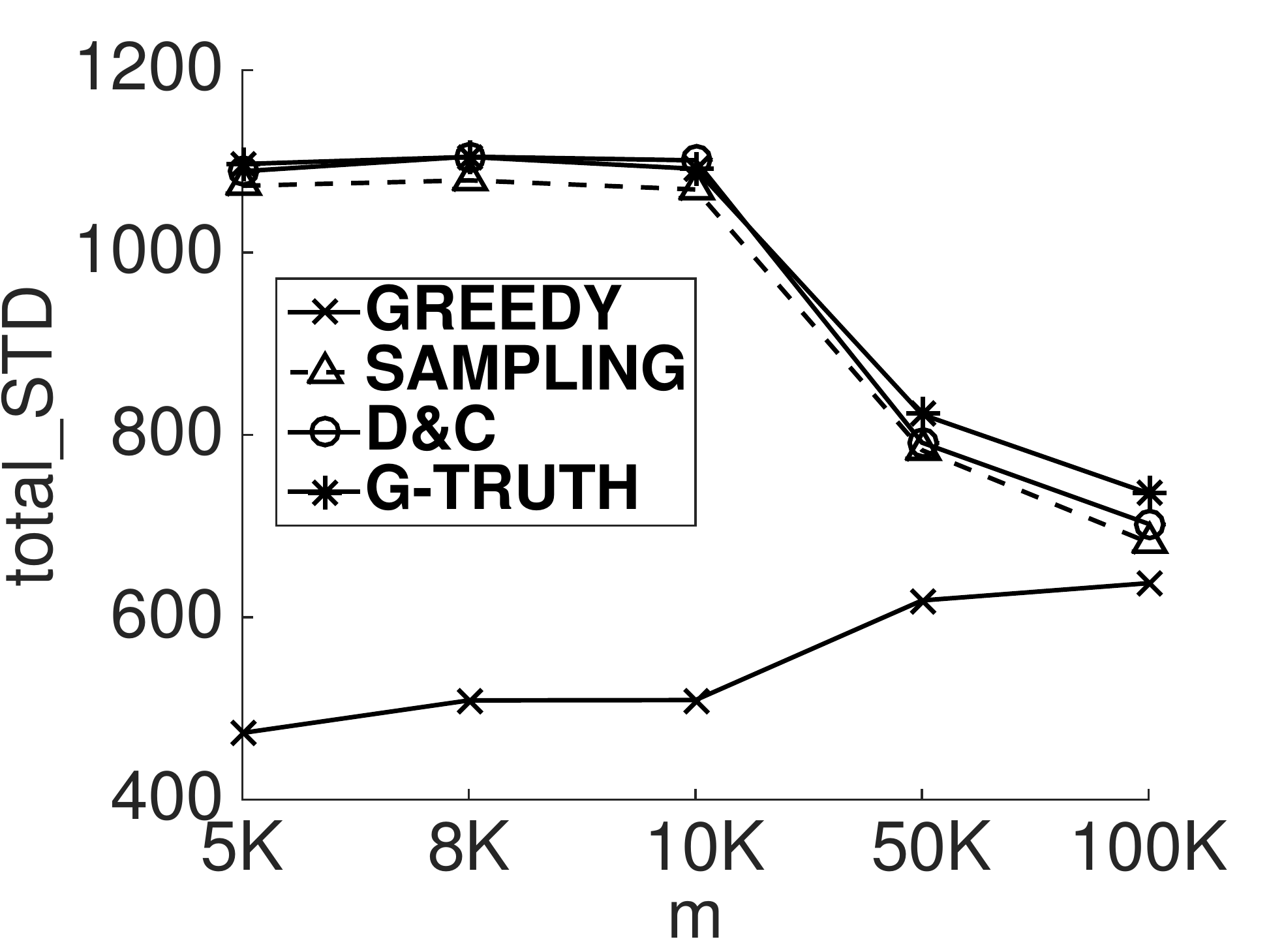}}
     \label{subfig:taskNSkewedDiversity}}
     \vspace{-3ex}
\caption{\small Effect of $m$ (SKEWED)} 
  \label{fig:taskNSkewed}
\end{figure}

\begin{figure}[ht!]\vspace{-1ex}
\centering
     \subfigure[][{\scriptsize Minimum Reliability}]{
       \scalebox{0.18}[0.18]{\includegraphics{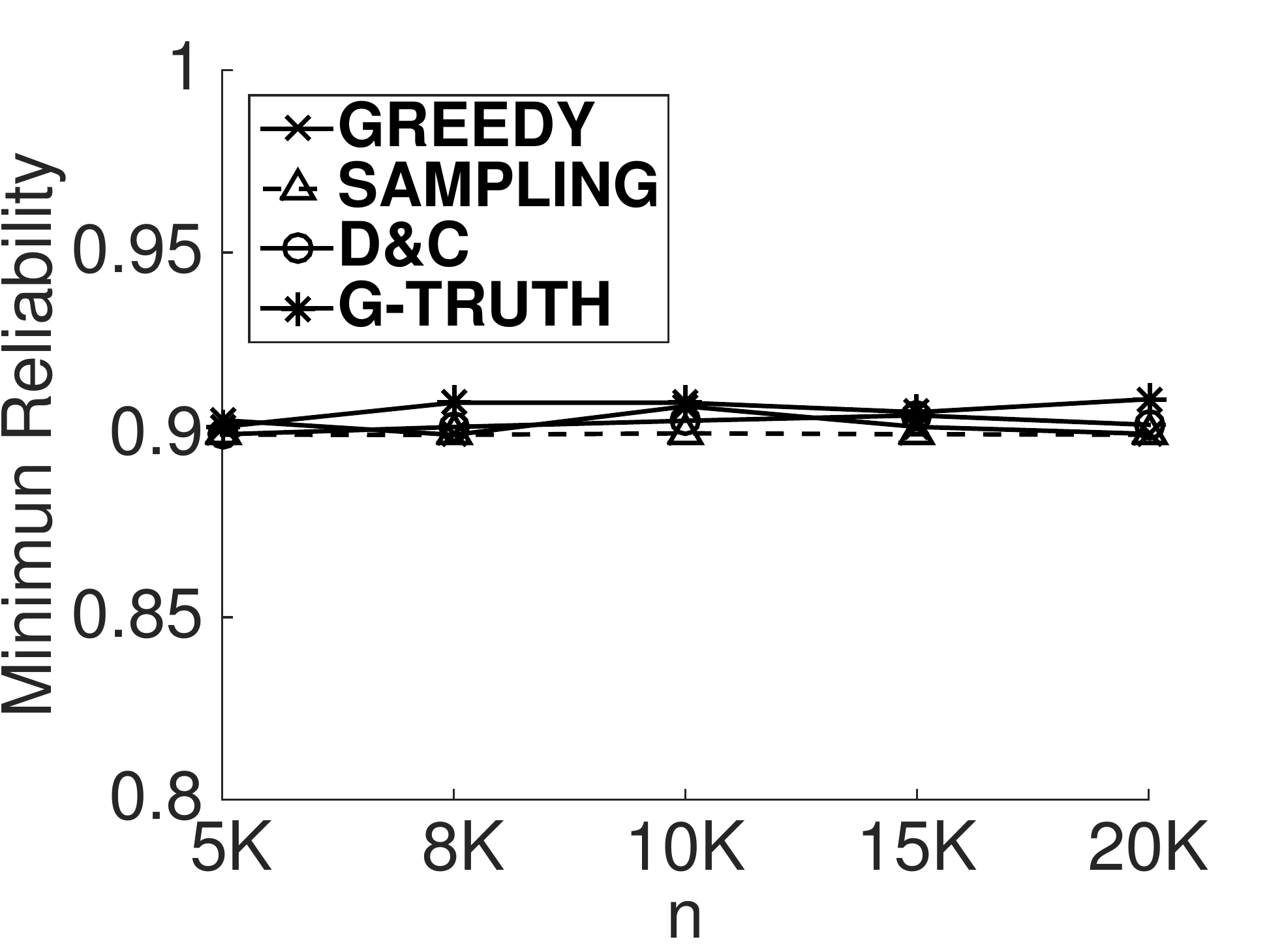}}
     \label{subfig:workerNSkewedReliability}}
     \subfigure[][{\scriptsize Summation of Diversity}]{
       \scalebox{0.18}[0.18]{\includegraphics{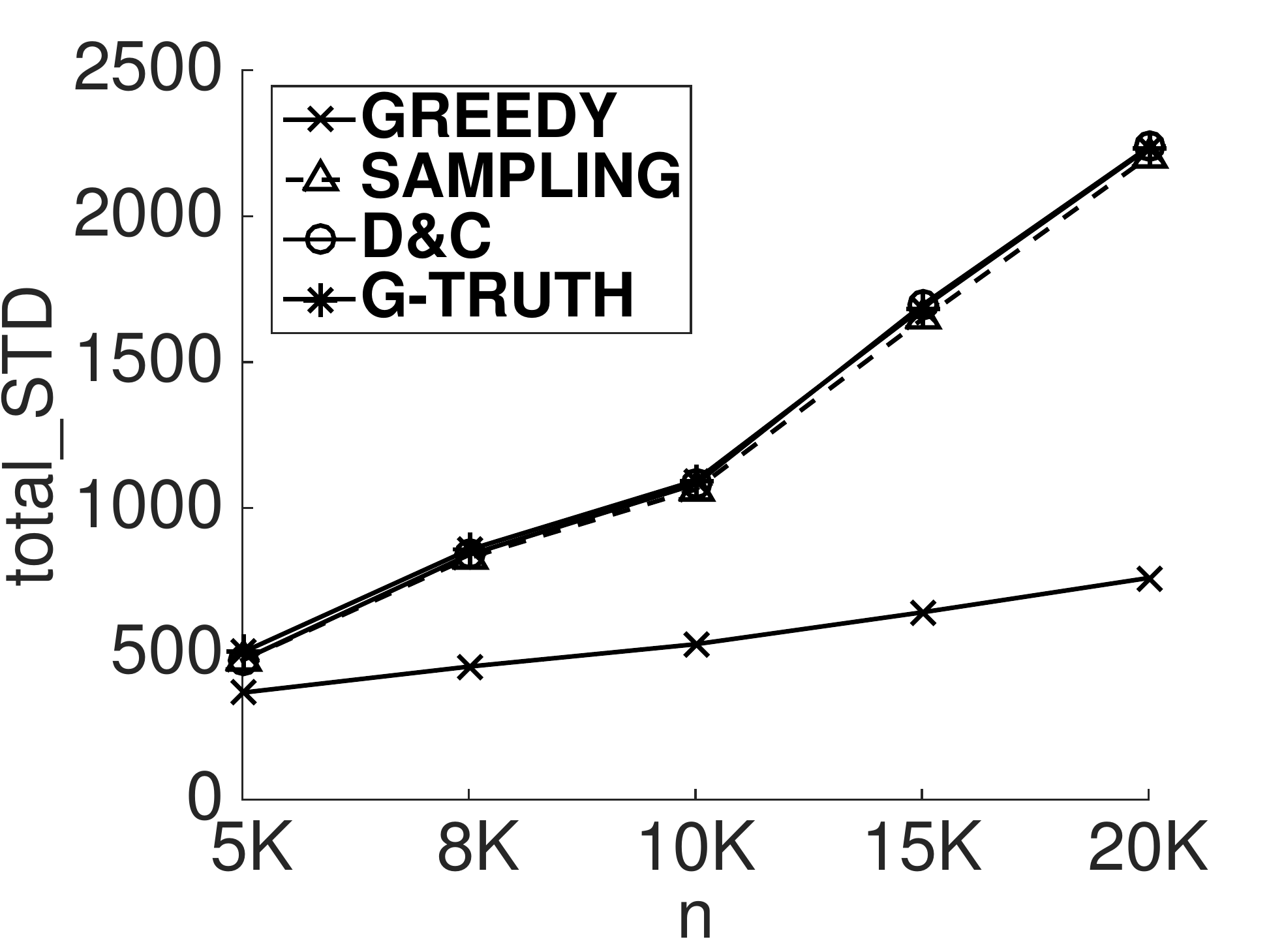}}
     \label{subfig:workerNSkewedDiversity}}
     \vspace{-3ex}
\caption{\small Effect of $n$ (SKEWED)}
  \label{fig:workerNSkewed}
\end{figure}

\vspace{0.5ex}\noindent {\bf Effect of the number of the tasks, $m$
(SKEWED).} Figure \ref{fig:taskNSkewed} shows the effect of the
number of the tasks, $m$, using SKEWED synthetic data. In figures,
the minimum reliability is not very sensitive to $m$. With the same
reasons discussed in Section \ref{subsec:expSynthetic}, when the
number of tasks $m$ increases, the diversities of SAMPLING and D\&C
decrease, whereas that of GREEDY increases.

\vspace{0.5ex}\noindent {\bf Effect of the number of the workers,
$n$ (SKEWED).} Figure \ref{fig:workerNSkewed} presents the effect of
the number of the workers, $n$, on SKEWED synthetic data set. We can
see that the minimum reliability is not sensitive to $n$. On the
other hand, the diversities of all the approaches increase when $n$
becomes larger, which is same with the situation in UNIFORM data
set. The reasons are same with the discussed reasons in Section
\ref{subsec:expSynthetic}.

\begin{figure}[ht!]\centering \vspace{-1ex}
	\subfigure[][{\scriptsize Minimum Reliability}]{
		\scalebox{0.18}[0.18]{\includegraphics{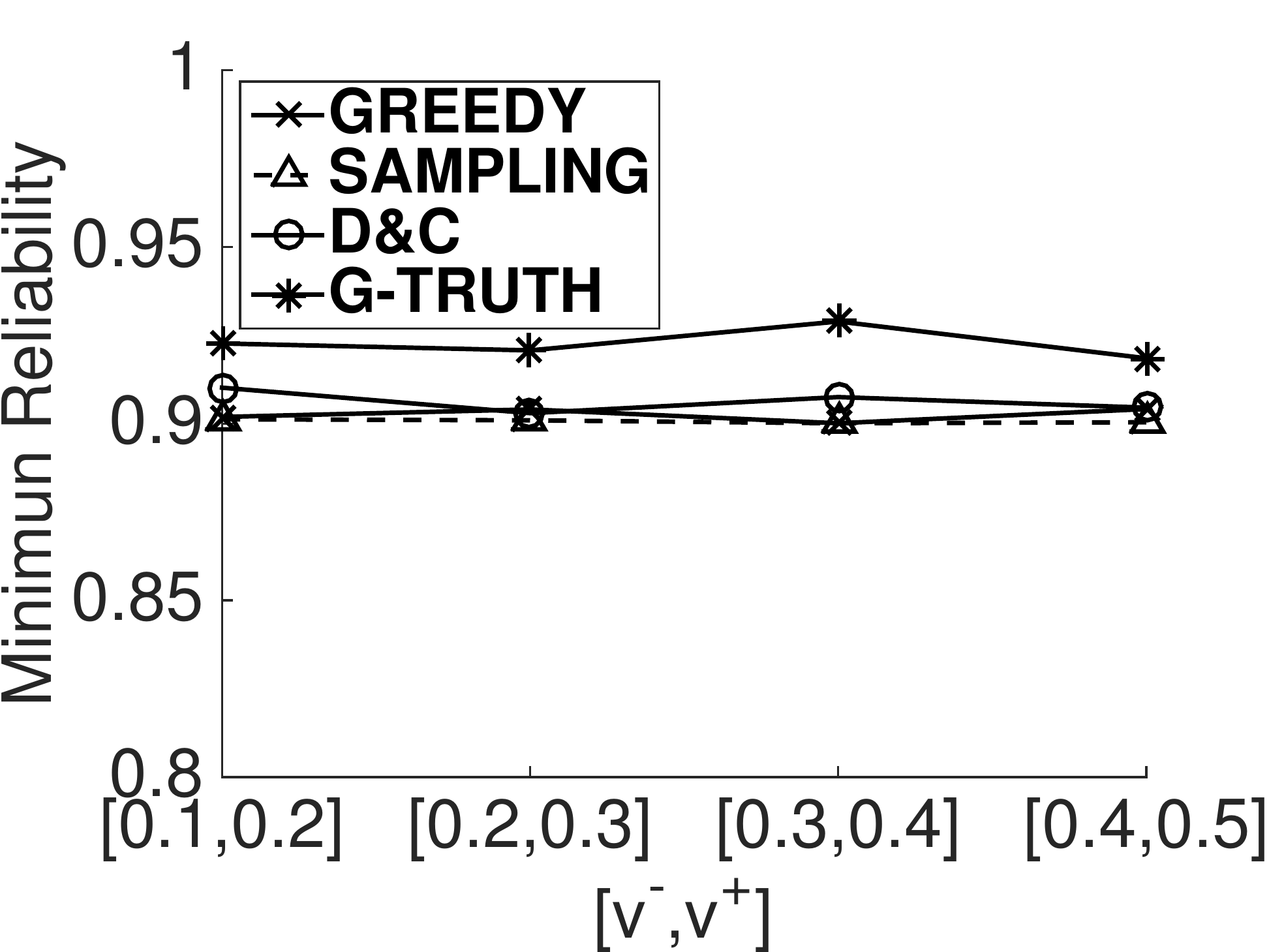}}\vspace{-2ex}
		\label{subfig:velocityR}}
	\subfigure[][{\scriptsize Summation of Diversity}]{
		\scalebox{0.18}[0.18]{\includegraphics{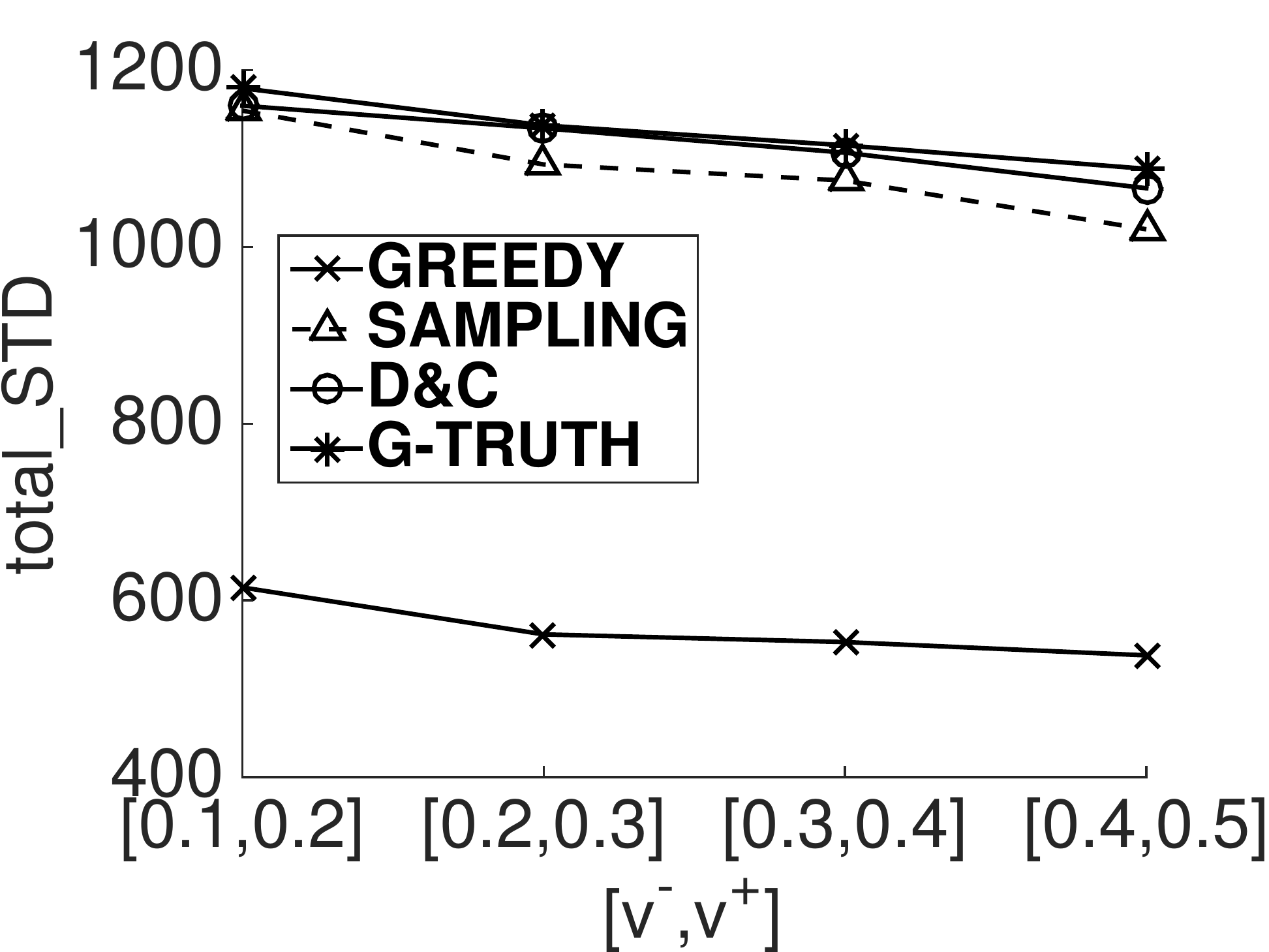}}\vspace{-2ex}
		\label{subfig:velocityD}}
	\vspace{-3ex}
	\caption{\small Effect of the Range of the Velocities $[v^-, v^+]$ (UNIFORM)} 
	\label{fig:velocityRange}
\end{figure}

\vspace{0.5ex}\noindent {\bf Effect of the Range of Workers'
	Velocities $[v^-, v^+]$. (UNIFORM)} Figure \ref{fig:velocityRange} reveals the
effect of workers' velocities range, $[v^-, v^+]$, from
$[0.1, 0.2]$ to $[0.4, 0.5]$. In figures, when the velocity in the range becomes
higher, the minimum reliability remains high (i.e., around 0.9).
Furthermore, the diversities of our approaches gradually decrease
for higher velocity range. Intuitively, faster moving workers can
potentially reach more tasks before their deadlines, which, however,
lowers the total spatial/temporal diversity. Similarly, SAMPLING and
D\&C are much better than GREEDY in diversity perspective, and their
diversities closely approach that of G-TRUTH.

\begin{figure}[ht!]
  \centering\vspace{-1ex}
       \subfigure[][{\scriptsize Minimum Reliability}]{
         \scalebox{0.18}[0.18]{\includegraphics{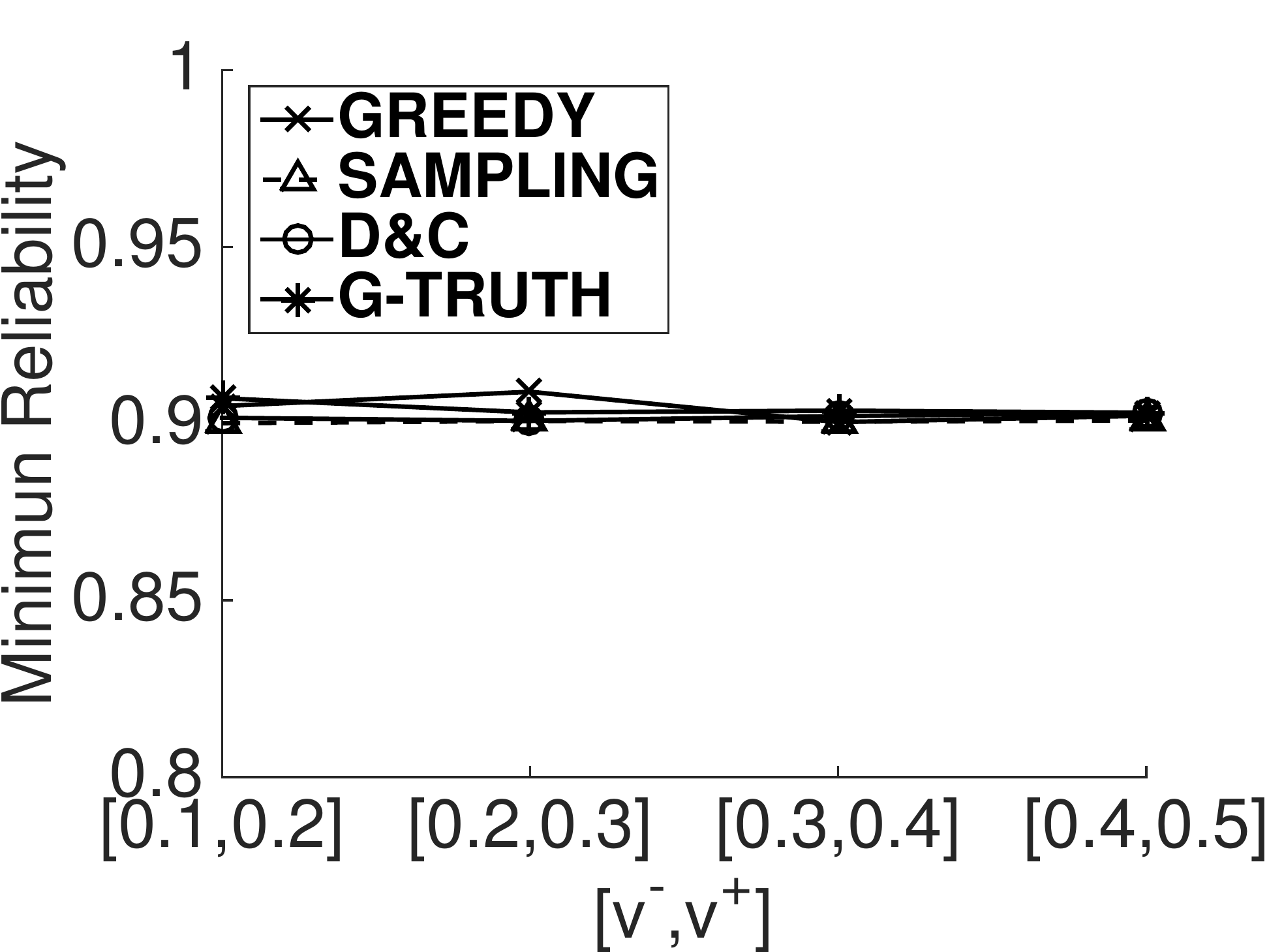}}\vspace{-2ex}
       \label{subfig:velocityRSkewed}}
       \subfigure[][{\scriptsize Summation of Diversity}]{
         \scalebox{0.18}[0.18]{\includegraphics{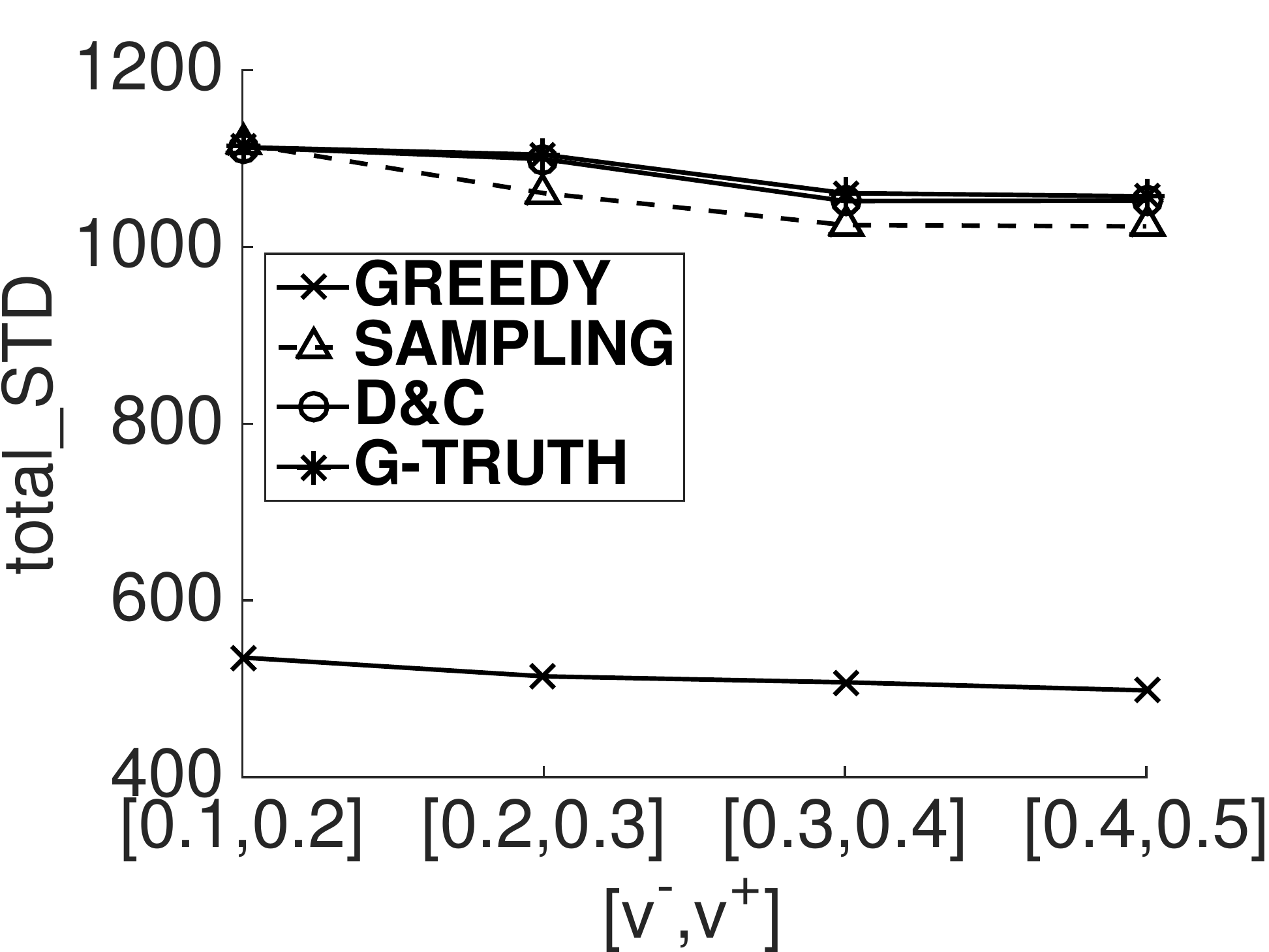}}\vspace{-2ex}
       \label{subfig:velocityDSkewed}}
        \vspace{-3ex}
  \caption{\small Effect of $[v^-, v^+]$ (SKEWED)}
    \label{fig:velocityRangeSkewed}
\end{figure}

\vspace{0.5ex}\noindent {\bf Effect of Workers' Velocity Range,
$[v^-, v^+]$ (SKEWED).} Figure \ref{fig:velocityRangeSkewed} shows
the effect of workers' velocity ranges over SKEWED synthetic data
sets, where other parameters are set to default values. The minimum
reliability of our approaches is not very sensitive to the velocity.
For diversities, the results are similar to that on UNIFORM
synthetic data. The same reason can explain why the diversities of
SAMPLING, D\&C, and G-TRUTH decrease, when workers' moving speeds
are improved.

\vspace{0.5ex}\noindent {\bf Effect of  the Range of Moving Angle,
$(\alpha_j^+ - \alpha_j^-)$ (SKEWED).} Figure \ref{fig:angleSkewed}
shows the experimental results by varying the range of workers'
moving angles, $(\alpha_j^+ - \alpha_j^-)$, from $(0,\pi/8]$ to $(0,
\pi/4]$, on SKEWED synthetic data sets, where other parameters take
default values. From figures, the minimum reliability is not very
sensitive to the range of worker' moving direction. As shown in
Figure \ref{subfig:AngleDiversitySkewed}, SAMPLING and D\&C can
achieve much higher diversity, $total\_STD$, than GREEDY, and they
have diversities similar to G-TRUTH.

\begin{figure}[ht!]\centering\vspace{-1ex}
     \subfigure[][{\scriptsize Minimum Reliability}]{
       \scalebox{0.18}[0.18]{\includegraphics{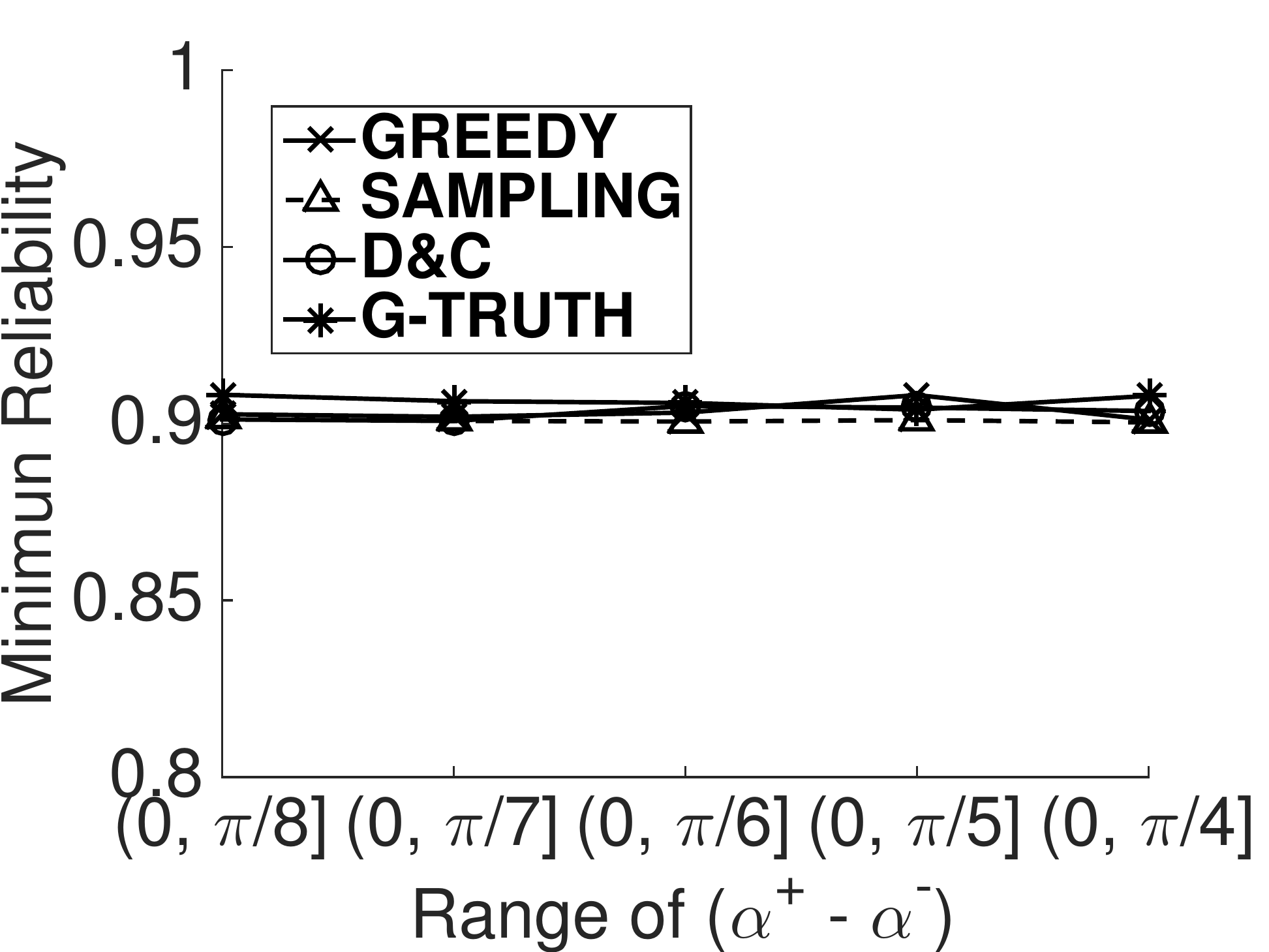}}\vspace{-2ex}
     \label{subfig:AngleReliabilitySkewed}}
     \subfigure[][{\scriptsize Summation of Diversity}]{
       \scalebox{0.18}[0.18]{\includegraphics{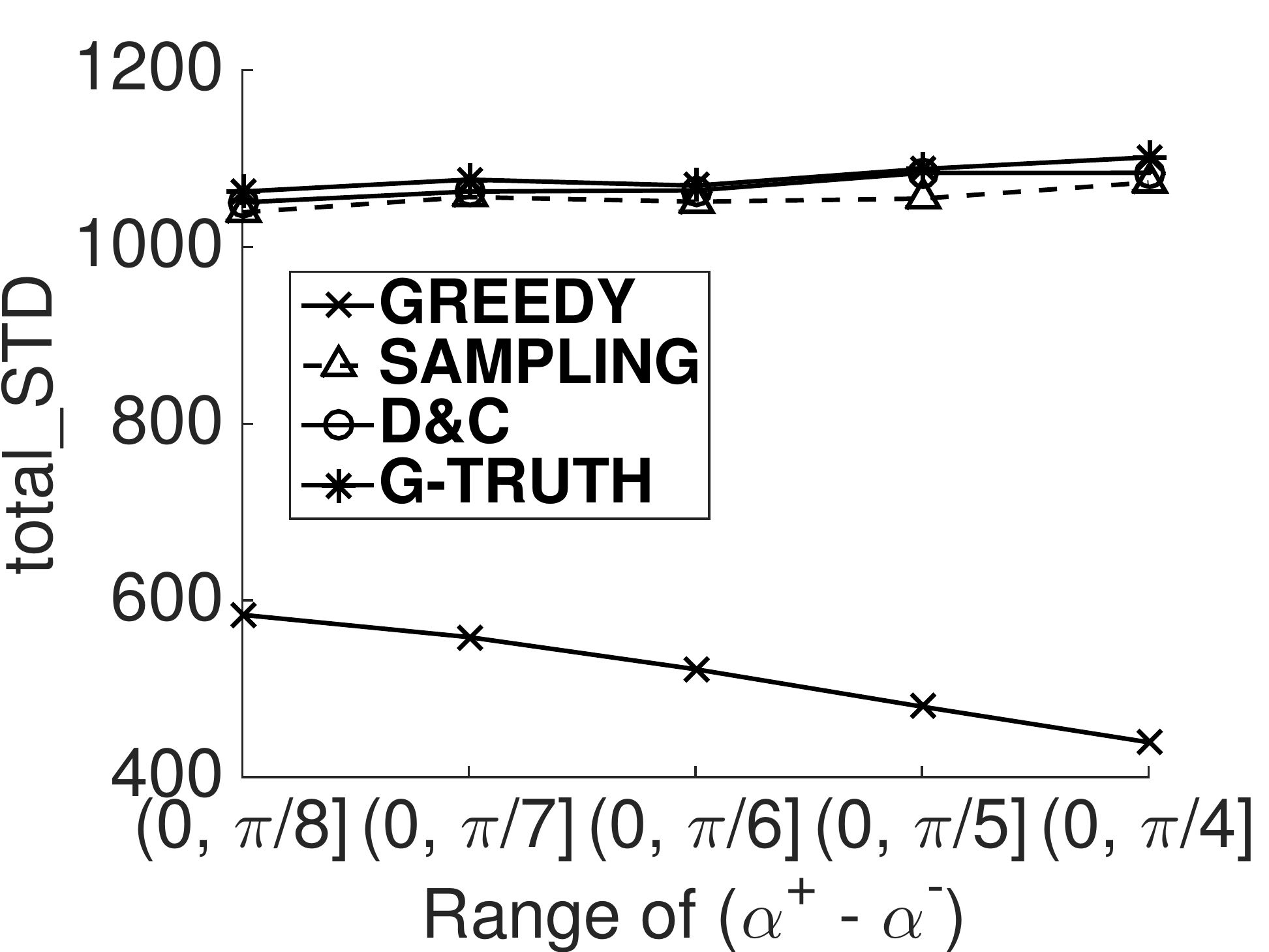}}\vspace{-2ex}
     \label{subfig:AngleDiversitySkewed}}
      \vspace{-3ex}
\caption{\small Effect of the Range of $(\alpha_j^+ - \alpha_j^-)$ (SKEWED)} 
  \label{fig:angleSkewed}
  \end{figure}

\end{document}